%% file: main.tex
\documentclass[11pt,letterpaper]{article}
\usepackage{hyperref}
\hypersetup{
pagebackref=true,
colorlinks=true,
linktocpage=true,
pdfstartview=FitH,
breaklinks=true,
pdfpagemode=UseNone,
pageanchor=true,
pdfpagemode=UseOutlines,
plainpages=false,
bookmarksnumbered,
bookmarksopen=false,
bookmarksopenlevel=1,
hypertexnames=true,
pdfhighlight=/O,
urlcolor=blue,linkcolor=blue,citecolor=blue,	
pdftitle={},
pdfauthor={},
pdfsubject={},
pdfkeywords={},
pdfcreator={pdfLaTeX},
pdfproducer={LaTeX with hyperref}
}

\usepackage{thm-restate}
\usepackage{verbatim}
\usepackage{amsfonts}
\usepackage{amsmath,amssymb,amsthm}
\usepackage{multirow}
\usepackage{url}
\usepackage{setspace}
\usepackage{times}
\usepackage{fullpage}
\usepackage{epsfig}
\usepackage[usenames]{color}
\usepackage{graphicx}
\usepackage{times}
\usepackage{enumerate}
\usepackage[sc]{mathpazo}
\usepackage[T1]{fontenc}
\usepackage{mathtools}
\usepackage{wrapfig}
\input{insbox}
\usepackage{caption}
\input{tikz-packages}

\usepackage[square,numbers]{natbib}
\input{macros}

\newtheorem{theorem}{Theorem}[section]
\newtheorem{lemma}[theorem]{Lemma}

\newtheorem{fact}[theorem]{Fact}
\newtheorem{corollary}[theorem]{Corollary}

\newtheorem{claim}[theorem]{Claim}

\newtheorem{definition}[theorem]{Definition}
\newtheorem{remark}[theorem]{Remark}

\usepackage[left=0.77in,right=0.77in,top=0.9in, bottom=0.92in]{geometry}
\usepackage{footmisc}

\spacing{1.16}

\def\obj{\textsf{obj}}
\def\calS{\mathcal{S}}
\def\calM{\mathcal{M}}
\def\calX{\mathcal{X}}
\def\arcs{\textsf{arcs}}
\def\cycles{\textsf{cycles}}
\def\rank{\textsf{rank}}
\def\adarcs{\arcs}
\def\len{\textsf{len}}

\newcommand{\valid}{valid }

\newcommand{\manolis}[1]{{\color{orange}Manolis:#1}}

\title{Smoothed Complexity of SWAP  in Local Graph Partitioning\vspace{0.2cm}}

\author{
Xi Chen\thanks{Supported by NSF IIS-1838154, CCF-2106429 and CCF-2107187.}\\
Columbia University\\
\tt{xichen@cs.columbia.edu}
\and
Chenghao Guo\thanks{Supported by NSF TRIPODS program award DMS-2022448 and by NSF Career Award CCF-1940205, CCF- 2131115.} \\
MIT\\
\tt{chenghao@mit.edu}
\and
Emmanouil V. Vlatakis-Gkaragkounis\thanks{Supported by Postdoctoral FODSI Simons-Fellowship.}\\
University California, Berkeley\\
\tt{emvlatakis@berkeley.edu}
\and
Mihalis Yannakakis\thanks{Supported by NSF CCF-2107187 and CCF-2212233.} \\
Columbia University\\
\tt{mihalis@cs.columbia.edu}
}
\begin{document}

\maketitle


\begin{abstract}
We give the first quasipolynomial upper bound $\phi n^{\text{polylog}(n)}$ for the smoothed complexity of the 
SWAP algorithm for local Graph Partitioning (also known as Bisection Width), 
where $n$ is the number of nodes in the graph and $\phi$ is a parameter that measures the magnitude of perturbations applied on its edge weights.
More generally, we show that the same quasipolynomial upper bound holds for the smoothed complexity of the 2-FLIP algorithm for any binary Maximum Constraint Satisfaction Problem, including local Max-Cut, for which similar bounds were only known for $1$-FLIP.
Our results are based on an analysis
  of cycles formed in long sequences of double flips,
  showing that it is unlikely for every move in a long sequence to incur a positive but small improvement in the cut weight.

\end{abstract}

\input{introduction}
\input{Preliminaries}

\input{Main-lemma-and-proof-of-theorem}

\input{Lemma-for-finding-good-windows}

\input{finding-cycles}

\input{General-case}

\input{csp}

\input{Conclusion}

\clearpage

\appendix
\input{Appendix-Preliminaries-Missing-Proofs}
\input{Appendix-Rank-Invariance}

\clearpage
\begingroup
\setlength{ \bibsep} {12pt}
\bibliography{ref.bib}
\endgroup

\end{document}

%% file: tikz-packages
\usepackage{physics}
\usepackage{amsmath}
\usepackage{tikz}
\usepackage{mathdots}
\usepackage{yhmath}
\usepackage{cancel}
\usepackage{color}
\usepackage{siunitx}
\usepackage{array}
\usepackage{multirow}
\usepackage{amssymb}
\usepackage{gensymb}
\usepackage{tabularx}
\usepackage{extarrows}
\usepackage{booktabs}
\usetikzlibrary{fadings}
\usetikzlibrary{patterns}
\usetikzlibrary{shadows.blur}
\usetikzlibrary{shapes}

%% file: macros.tex
\def\00{\mathbf{0}} \def\11{\mathbf{1}}
\def\order{\text{order}}\def\eps{\epsilon} \def\rank{\text{rank}}

\newcommand{\E}{{\rm I}\kern-0.18em{\rm E}}

   \def\11{\mathbf{1}}

\def\wit{\mathsf{wit}}

%% file: introduction.tex
\section{Introduction}\label{sec:intro}
\setlength{\parindent}{15pt}

\emph{Local search} has been a powerful machinery for a plethora of problems in combinatorial optimization,
  from the classical Simplex algorithm for linear programming to 
  the gradient descent method for modern machine learning problems, to effective heuristics (e.g. Kernighan-Lin) for basic combinatorial problems such as the Traveling Salesman Problem and Graph Partitioning.
A local search algorithm begins with an initial candidate solution and then follows a
path by iteratively moving to a better neighboring solution until a local optimum~in~its neighborhood is reached. The quality of the
obtained solutions depends of course on how rich is the neighborhood structure that is explored by the algorithm. Local search is a popular approach to optimization because of the general applicability of the method and the fact that the algorithms typically run fast in practice. 
In contrast to their empirical fast convergence, however, many local search algorithms have exponential running time in the worst case 
due to delicate pathological instances that one may never
encounter in practice.
To bridge this striking discrepancy, Spielman and Teng \cite{spielman2004smoothed} proposed the framework of \emph{smoothed analysis}, a hybrid of the classical worst-case and average-case analyses. 
They used it to provide  
 rigorous justifications for the empirical performance of the Simplex algorithm
 by showing its smoothed complexity to be polynomial.
Since then, the smoothed analysis of algorithms and problems from combinatorial optimization \cite{beier2003random,englert2016smoothed,roglin2007smoothed}, among many other research areas such as numerical methods \cite{dadush2018friendly,bhaskara2014smoothed,roglin2007smoothed,farrell2016smoothed}, machine learning \cite{blum2002smoothed,arthur2011smoothed,arthur2006worst,sivakumar2020structured} and algorithmic game theory\cite{boodaghians2020smoothed,coordination,beier2022smoothed,xia2020smoothed}, has been studied extensively .

In this paper we study the smoothed complexity of local search algorithms for the classical problem of \emph{Graph Partitioning} (also known as \emph{Bisection Width} in the literature).
In the problem we are given 
  edge weights $X=(X_e:e\in E_{2n})$
  of a complete graph  $K_{2n}=(V_{2n},E_{2n})$ with $X_e\in [-1,1]$,
  and the goal is to find a \emph{balanced} partition $(U,V)$ of $V_{2n}$ into two equal-size subsets $U$ and $V$ to minimize the weight of the corresponding cut (i.e., the sum of weights of edges with one node in $U$ and the other node in $V$). 
Graph Partitioning has been studied extensively, especially in practice. It forms the basis of divide and conquer algorithms and is used in various application
domains, for example in laying out circuits in VLSI.
It has also served as a test bed for 
  algorithmic ideas \cite{JAMS89}.


Given its NP-completeness \cite{NPC}, heuristics have been developed to solve Graph Partitioning in practice.
A commonly used approach is based on local search: starting with an initial balanced partition,
local improvements on the cut are made iteratively until a balanced partition that mimimizes the cut  within its \emph{neighborhood} is reached.
The simplest 
neighborhood is the SWAP neighborhood, where two balanced partitions
are neighbors if one can be obtained from the other by swapping two nodes, one from each part.
A locally optimal solution under the SWAP neighborhood can be found naturally by the SWAP algorithm, which keeps swapping two nodes  as long as the swap improves the cut.
A more sophisticated neighborhood structure, which yields much better locally optimal solutions in practice, is that of the Kernighan-Lin (KL) algorithm which performs in each move a sequence of swaps \cite{KL70}. 

These local search algorithms for Graph Partitioning typically converge fast in practice. 
(For a thorough experimental analysis of their performance, and comparison with simulated annealing, regarding both the quality of solutions and the running time, see \cite{JAMS89}.)
In contrast, it is also~known~that~the worst-case complexity is exponential. (Finding a locally optimal solution for  Graph~Partitioning  \mbox{under} the sophisticated Kernighan-Lin neighborhood, and even under the  SWAP neighborhood is complete in PLS \cite{johnson1988easy,schaffer1991simple}. 
The hardness reductions give instances on which these algorithms take exponential time to converge.)
\emph{This significant gap in our understanding motivates us to work on the smoothed complexity of the SWAP algorithm  for Graph Partitioning in this paper.}

\def\calZ{\mathcal{Z}}

We work on the full perturbation model, under which edge weights are drawn independently from a collection of distributions $\calX=(\calX_e:e\in E_{2n})$. Each $\calX_e$ is supported on $[-1,1]$, and has~its   density function bounded from above by a parameter $\phi>0$.
Our goal is to understand the expected number of steps the SWAP algorithm takes, as a function of $n$ and $\phi$, against any  edge weight distributions $\calX$.\footnote{Note that any upper bound under the full perturbation model applies to the alternative, simpler model where an adversary commits to a weight $w_e$ for each edge and then 
  all edge weights are perturbed independently by a random noise $\calZ_e$ (for example, drawn uniformly from a small interval), i.e. the weights are $\calX_e = w_e + \calZ_e$. The parameter $\phi$ in the full perturbation model is a bound on the pdf of the perturbations $\calZ_e$.}
Note that the SWAP algorithm, similar
to the Simplex algorithm, is a family of algorithms since one can implement it using  different pivoting rules, deterministic or
randomized, to pick the next pair of nodes to swap when more than one pairs  can improve
the cut. We would like to establish upper bounds that hold for any implementation of the SWAP algorithm.

\subsection{Related work: Smoothed analysis of 1-FLIP  for Max-Cut}

There has not been any previous analysis on SWAP under the smoothed setting, as far as we are aware.
In contrast, much progress has been made on the smoothed analysis of the~\emph{1-FLIP algorithm for Max-Cut} \cite{elsasser2011settling,etscheid2015smoothed,angel2017local,bibak2019improving,chen2020}.
The major challenge for the analysis of  SWAP, as we discuss in more details in Section \ref{sec:approach}, is to overcome substantial new obstacles posed by the richer neighborhood structure of SWAP, which are not present in the simpler \emph{$1$-change neighborhood} behind $1$-FLIP.


 
Recall  
in Max-Cut, we are given  edge weights $X=(X_e:e\in E_n)$ 
of a complete graph $K_n=(V_n,E_n)$ with $X_e\in [-1,1]$ and the goal is to find a (\emph{not necessarily balanced}) partition of $V_n$ to maximize
 the~cut.
\footnote{Since we allow weights in $[-1,1]$,
  maximizing the cut is the same as minimizing the cut after negating all edge weights.
Hence the only difference of Max-Cut, from Graph Partitioning,  is that the partition does not have to be balanced.}
The simplest neighborhood structure for local search on Max-Cut is the so-called \emph{$1$-change neighborhood}, where two partitions are neighbors if one can be obtained from the other by moving a single node to the other side. 
The $1$-FLIP algorithm finds such a locally optimal solution by keeping moving nodes, one by one, as
long as each move improves the cut. 
For the structured perturbation model, where~a~graph~$G$ (not necessarily a complete graph) is given and only weights of edges in $G$ are perturbed,
\cite{etscheid2015smoothed}~showed that the expected number of steps $1$-FLIP takes to  terminate is at most $\phi n^{\log n}$. Subsequently, the bound  was improved by \cite{angel2017local} to $\smash{\phi \cdot \text{poly}(n)}$ for the full perturbation model,
with further improvements in \cite{bibak2019improving} on the polynomial part of $n$. 
The upper bound of \cite{etscheid2015smoothed} for the structured model was recently  improved to $\smash{\phi n^{\sqrt{\log n}}}$ in \cite{chen2020}.

\subsection{Our Contributions}

We present the first smoothed analysis of the SWAP algorithm for Graph Partitioning.
Our main result for SWAP is a 
quasipolynomial upper bound on its expected running time:\hspace{0.04cm}\footnote{We did not make an attempt to optimize the constant $10$ in the $\text{polylog}$ exponent.} 

\begin{restatable}{theorem}{swapthm}\label{thm:swap}
Let $\calX=(\calX_e:e\in E_{2n})$ be   distributions of edge weights such that 
each $\calX_e$ is supported on $[-1,1]$ and has its density function bounded from above by a parameter $\phi>0$.
Then with probability at least $1-o_n(1)$ over the draw of edge weights $X\sim \calX$, any implementation of  SWAP 
takes at most $\phi n^{O(\log^{10} n)}$ steps to terminate.
\end{restatable}

The proof of Theorem \ref{thm:swap} for SWAP is based on techniques we develop for a more challenging problem: the smoothed analysis of \emph{$2$-FLIP  for Max-Cut}.
Starting with an initial partition (not necessarily balanced), in each round, $2$-FLIP can move either one node (like $1$-FLIP) or two nodes (not necessarily in different parts) as long as the cut is improved. If we restrict the algorithm to only use double flips in every move, then we call this variant {\em pure} 2-FLIP.
Feasible moves in SWAP are clearly feasible in pure 2-FLIP as well but not vice versa. Thus, an improving sequence of SWAP in the Graph Partitioning problem is also an improving sequence of pure 2-FLIP in the Max-Cut problem on the same instance.

We do not make again any assumption on the pivoting rule used by 2-FLIP (i.e., which move is selected in each step if there are multiple improving moves), except that if both single and double flips are allowed, then the algorithm never moves a pair of nodes when moving only one of the two nodes would yield a better cut. Clearly, any reasonable implementation of 2-FLIP  satisfies this property.
Our main result on 2-FLIP is a similar quasipolynomial upper bound on its expected running time. 
The same result holds also for any implementation of the pure 2-FLIP algorithm that performs only 2-flips.
This is the first smoothed analysis of $2$-FLIP:

\begin{restatable}{theorem}{mainthm}\label{thm:main-high-prob}
Let $\calX=(\calX_e:e\in E_n)$ be  distributions of edge weights such that 
each $\calX_e$ is supported on $[-1,1]$ and has its density function bounded from above by a parameter $\phi>0$.
Then with probability at least $1-o_n(1)$ over the draw of edge weights $X\sim \calX$, any implementation of the $2$-FLIP algorithm 
takes at most $\phi n^{O(\log^{10} n)}$ steps to terminate.
\end{restatable}

A more general class of problems that is related to Max-Cut is the class of {\em Maximum Binary Constraint Satisfaction Problems} 
(MAX 2-CSP).
In a general Max-2CSP, the input consists of a set of Boolean variables and a set of weighted binary constraints over the variables; the problem is to find an assignment to the variables that maximizes the weight of the satisfied constraints. Max-Cut is the special case  when every constraint is a $\neq$ (XOR) constraint. Other examples are Max 2SAT and Max Directed Cut (i.e.,~the~Max Cut problem on weighted directed graphs).
More generally, in a {\em Binary Function Optimization Problem}  (BFOP), instead of binary constraints the input has a set of weighted binary functions on the variables, and the objective is to find an assignment that maximizes the sum of the weights of the functions (see Section \ref{csp} for the formal definitions).
It was shown in \cite{chen2020} that the results for $1$-FLIP for Max-Cut   generalize to all Max 2-CSP and BFOP problems.
We prove that this is the case also with 2-FLIP.

We say an instance of Max 2-CSP or BFOP is {\em complete} if it includes a constraint or function for every pair of variables.

\begin{restatable}{theorem}{cspthm}\label{thm:csp}
Let $I$ be an arbitrary complete instance of a MAX 2-CSP (or BFOP) problem with $n$ variables and $m$ constraints (functions) with independent random weights in $[-1,1]$ with density at most $\phi>0$.
Then, with probability at least $1-o_n(1)$ over the draw of the weights, any implementation of  2-FLIP 
takes at most $m \phi n^{O(\log^{10} n)}$ steps to terminate.
\end{restatable}

For all the aforementioned problems, by controlling the tail-bound of the failure probability, we can strengthen our analysis to derive the same bound for the expected number of steps needed to terminate as in  the standard smoothed analysis prototype (See Corollary~\ref{thm:main-exp}).

\subsection{Our Approach}\label{sec:approach}

Here, we give an overview of our proof approach, focusing on the analysis of the $2$-FLIP algorithm for Max-Cut (Theorem \ref{thm:main-high-prob}). Many details are omitted in this subsection, to help the reader get an overall view of some of the key ideas and the structure of the proof. Note that 2-FLIP clearly subsumes 1-FLIP, since it explores a much larger neighborhood structure. For example, a 2-FLIP algorithm could apply improving 1-flips as long as possible, and only when the partition is locally optimal with respect to the 1-flip neighborhood apply an improving 2-flip. Therefore, the complexity (whether smoothed or worst-case) of 2-FLIP is clearly at least as large as the complexity of 1-FLIP, and could potentially be much larger. Similarly, the analysis of 2-FLIP has to subsume the analysis of 1-FLIP, but it needs to address many more challenges, in view of the larger space of possible moves in each step (quadratic versus linear).  In a sense, it is analogous to the difference between a two-dimensional and a one-dimensional problem.

First, let's briefly review the approach of previous work \cite{etscheid2015smoothed} on the simpler 1-FLIP problem. Since the edge weights are in $[-1,1]$, the weight of any cut is in $[-n^2,n^2]$. For the execution of the FLIP algorithm to be long, it must have many moves where the gain in the cut weight is very small, in $(0,\epsilon]$ for some small $\epsilon>0$. It is easy to see that any single move by itself has small probability ($\phi \epsilon$) of this being the case. If different moves were uncorrelated, then the probability that a sequence increases the weight of the cut by no more than $\eps$ would go down exponentially with the length of the sequence. Of course, different moves are correlated. However, the same effect holds if the improvements of the moves are linearly independent in the following sense.  For any sequence of the FLIP algorithm, the \emph{improvement vector} of one move is the vector indexed by the edges with entries in $\{-1,0,1\}$ indicating whether each edge is added or removed from the cut
as a result of the move.
Most work along this line of research is based on the following fact (see Corollary \ref{cor:rank} for the formal statement): If the rank of the set of improvement vectors is $\rank$, then the sequence has improvement at most $\epsilon$ with probability at most $(\phi \epsilon)^{\rank}$. 
On the other hand, if all sequences with length at most $\Theta(n)$ have an improvement of at least $\epsilon$, then the number of steps of FLIP is bounded by $\Theta(n)\cdot (2n^2/\epsilon)=\text{poly}(n)/\eps$, as the total improvement cannot exceed $2n^2.$
So a natural approach is to union bound over all possible sequences of length $\Theta(n)$ and all $2^n$ possible initial configurations, which yields a probability upper bound of 
$2^nn^{\Theta(n)}(\phi\epsilon)^{\rank}$.

Getting a quasi-polynomial complexity bound using the union bound above requires the \rank \ of any sequence of length $\Theta(n)$ to be at least $\Omega(n/\log n)$. However, this is not always true (consider, e.g., a sequence in which only $n^{0.1}$ distinct nodes moved). 
One key idea of \cite{etscheid2015smoothed} is to avoid union bound over all initial configurations and only union bound over initial configurations of \emph{active} nodes (nodes that move at least once in the sequence) by looking at \emph{arcs}. 
An arc is defined to be two adjacent moves of the same node. 
By taking the sum of improvement vectors of the two moves of an arc, edges that involve inactive nodes are cancelled, so the union bound over sequences of length $\ell$ becomes 
$2^\ell n^\ell (\phi\epsilon)^{\rank_{\arcs}}$.

To lower bound the rank of arcs of sequence $S$, $\rank_{\arcs}(S)$, they proved it is at least half of the number of nodes that appear more than once in the sequence, denoted $V_2(S)$.
The essential combinatorial claim made by \cite{etscheid2015smoothed} is that for any sequence of length $\Omega(n)$, there exists a substring of length $\ell$ with $V_2(S)$ at least $\Omega(\ell/\log n)$. This can be shown by bucketing arcs by length into buckets $[2^i,2^{i+1})$ and picking the largest bucket as length of the substring. On average, a random substring would contain $\Omega(\ell/\log n)$ arcs with similar length, and therefore, $\Omega(\ell/\log n)$ arcs with distinct nodes. The similar idea is used in Case 1 of our Section~\ref{sec:general-case} to handle 1-moves (moves that flip a single node).

Now let's return to the case of the 2-FLIP algorithm. A step now can move two nodes at the same time, and this fact poses qualitatively new challenges to the proof framework.
Now we have to deal not just with sets (e.g., the set of nodes that move more than once) but instead with relations (graphs). 
Define an \emph{auxiliary graph} $H$ for the sequence of moves that contains $K_n$ as vertices and an edge for each 2-move of the sequence. 
If we still want to eliminate the influence of inactive nodes in the improvement vector by summing or subtracting two moves as in the 1-FLIP case, the moves have to contain the \emph{exact} same pair of nodes. 
This happens too sparsely in the improving sequence of 2-FLIP to provide enough rank. 
To this end, we generalize the notion of arcs to \emph{cycles}. 
A cycle is a set of 2-moves of the sequence whose corresponding edges form a cycle in $H$. But not all cycles of $H$ are useful.
We are interested only in cycles for which there is a linear combination of the improvement vectors of the moves of the cycle that cancels all edges of $K_n$ that involve an inactive node (i.e., the corresponding entry in the linear combination is 0); these are the cycles that are useful to the rank and we call them \emph{dependent cycles}.

So the goal is to find a substring $S$ of length $\ell$ where we can lower bound $\rank_{\cycles}(S)$ by $\ell/\text{polylog} (n)$. 
The ideal case 
would be the case where all nodes have $O(\text{polylog} (n))$ but at least $2$ appearances in the substring, i.e., all nodes have degree between 2 and $O(\text{polylog} (n))$ in $H$. 
In this case, we can repeat the following process to find enough cycles. Find a dependent cycle in $H$, pick an edge in $K_n$ that is non-zero in the improvement vector of the cycle (we call this the \emph{witness} of the cycle) and delete both nodes of the witness from $H$.
This way the improvement vector of cycles of $H$ we pick in the future will not contain witnesses from previous cycles, and improvement vectors of cycles we pick would form a triangular matrix that has full rank. 
Since any node has $O(\text{polylog}( n))$ degree in $H$, each iteration deletes $O(\text{polylog}( n))$ edges. 
So the process can be repeated at least $\Omega(\ell/\text{polylog}( n))$ times. 

However, it is not hard to construct sequences with polynomial length, such that any substring consists mostly of moves involving {one high-degree node (with degree even $\Omega(\ell)$) and one degree-$1$ node, so deleting the high-degree node would have a significant impact on the graph and the process can only repeat for a few rounds and lead to a few cycles.}
So the challenge is to run a similar process, but reuse high-degree nodes carefully without repeating witnesses found in previous cycles. 
Suppose we find a cycle $C$ with witness edge $(u,v)$. 
To avoid including the edge in another cycle $C'$, a sufficient condition is that: (1) $u$ is not included in $C'$. (2) For any two adjacent edges (edges in $H$, not $K_n$) of $v$ in $C'$, $u$ never moved between the two corresponding moves of the edges. 
To meet condition 1, we can delete $u$ from $H$. 
To meet condition 2, we can make multiple copies of $v$ in $H$ where each copy corresponds to moves in $S$ where $u$ doesn't move between them.
We call this operation \emph{splitting} since the new graph is generated by deleting and splitting the original $H$.
The new graph after splitting is denoted by \emph{splitted auxiliary graph}.
Our algorithm for finding a large number of linearly independent cycles can be described as repeatedly performing the following process. 
Find a cycle in the splitted auxiliary graph with witness $(u,v)$ by a tree-growing argument, delete $u$ and split the graph by creating multiple copies of $v$. We have to choose carefully the witness edges $(u,v)$ and do the splitting, so that the number of nodes does not proliferate in this process.

Compared to the original auxiliary graph, the number of edges deleted and the number of new nodes introduced is proportional to the degree of $u$, so the number of cycles for a sequence of length $\ell$ we can find in the algorithm is bounded by $\ell/deg(u)$. 
To find $\ell/poly(\log n)$ cycles, we need a window where decent amount of moves involve a node $u$ that has $deg(u)$  bounded by $poly(\log n)$. 
The existence of such window in an arbitrary sequence that is long enough can be proven via a sophisticated bucketing and counting argument.

The overall argument then for 2-FLIP is that, given a sufficiently long sequence of improving moves (specifically, of length $n \cdot poly(\log n))$), we can find a window (a substring) such that the rank of the arcs and cycles in the window is within a $poly(\log n)$ factor of the length of the window.
As a consequence, with high probability the weight of the cut improves by a nontrivial amount $\epsilon$ (1/quasi-polynomial) during this window.
This can happen at most $n^2/\epsilon$ times, hence the length of the execution sequence of 2-FLIP is at most quasi-polynomial.

\subsection{Organization of the paper.}
The rest of the paper is organized as follows. Section~\ref{sec:setup} gives basic definitions of the problems and the smoothed model, defines the central concepts of arcs and cycles, their improvement vectors, and proves a set of basic lemmas about them that are used throughout in the subsequent analysis. Section~\ref{sec:main-lemma-and-thm} states the main lemma on the existence of a nice window in the move sequence such that the arcs and cycles in the window have high rank, and shows how to derive the main theorem from this lemma. Sections \ref{sec:good-window} and \ref{sec:finding-cycles} prove the main lemma in the case that all the moves are 2-moves (this is the more challenging case). First we show in Section \ref{sec:good-window} the existence of a nice window (in fact a large number of nice windows, since this is needed in the general case) such that many moves in the window have the property that both nodes of the move appear a substantial number of times in the window (at least polylog$(n)$ times), and one of them does not appear too many times (at most a higher polylog$(n)$). In Section \ref{sec:finding-cycles} we show how to find in such a nice window a large number of cycles whose improvement vectors are linearly independent. Section~\ref{sec:general-case} extends the proof of the main lemma to the general case where the sequence of moves generated by 2-FLIP contains both 1- and 2-moves. Finally, in Section \ref{csp} we extend the results to the class of Maximum Binary Constraint Satisfaction and Function Optimization problems.

%% file: Preliminaries.tex
\section{Preliminaries}\label{sec:setup}

We write $[n]$ to denote $\{1,\ldots,n\}$.
Given two integers $a\le b$, we write $[a:b]$ to denote $\{a,\ldots,b\}$.
Given $\gamma,\gamma'\in \{\pm 1\}^n$  we use $d(\gamma,\gamma')$ to denote the Hamming distance between $\gamma$ and $\gamma'$, i.e., the number of entries $i\in [n]$ such that $\gamma_i\ne \gamma_i'$.

\subsection{Local Max-Cut and the FLIP Algorithm}
Let $K_n=(V_n,E_n)$ with $V_n=[n]$ be the complete undirected graph over $n$ nodes. 
Given 
edge weights $X=(X_e: e\in E_n)$ with $X_e\in [-1,1]$,
  the \emph{$k$-local Max-Cut problem} is to find a partition of $V_n$ into two sets $V_1$ and $V_2$
  such that the weight of the corresponding cut 
  (the sum of weights of edges with one node in $V_1$ and the other in $V_2$)
  cannot be improved by moving no more than $k$ nodes to the other set.
Formally, the objective function of our interest
  is defined as follows:
Given any \emph{configuration} $\gamma\in \{\pm 1\}^n$
  (which corresponds to a partition $V_1,V_2$ with
  $V_1=\{u\in V_n:\gamma(u)=-1\}$ and $V_2=\{u\in V_n:\gamma(u)=1\}$), the objective function is
    
\begin{equation}\label{eqn:maxcut-obj}
\textsf{obj}_{X}(\gamma):= \sum_{(u,v)\in E_n}  {X_{(u,v)}}\cdot \mathbf{1}\big\{\gamma(u) \neq \gamma(v) \big\} = \frac{1}{2}\sum_{(u,v)\in E_n} X_{(u,v)}\cdot \big(1-\gamma(u)\gamma(v)\big).
\end{equation}
Our goal is to find a configuration $\gamma\in \{\pm 1\}^n$ that is a $k$-local optimum, i.e., $\textsf{obj}_{X}(\gamma)\ge \textsf{obj}_X(\gamma')$
  for every configuration $\gamma'\in \{\pm 1\}^n$ with Hamming distance no more than $k$ from $\gamma$.

A simple local search algorithm for $k$-Local Max-Cut  is the following $k$-FLIP algorithm: 
\begin{center}
\begin{minipage}{15.75cm}
\begin{flushleft}
\emph{Start with some initial configuration $\gamma=\gamma_0\in \{\pm 1\}^n$. 
While there exists a configuration 
  $\gamma'$\\ with $d(\gamma',\gamma)\le k$ such that $\emph{\textsf{obj}}_X(\gamma')>\emph{\textsf{obj}}_X(\gamma)$, select one such configuration $\gamma'$ 
   (according\\ to some
pivoting criterion), set $\gamma=\gamma'$ and repeat, until no such configuration $\gamma'$ exists.
}\end{flushleft} 
\end{minipage}
\end{center}
The execution of $k$-FLIP on  $K_n$ with edge weights $X$ depends on both the initial configuration $\gamma_0$ and the pivoting criterion used to select the next configuration in each iteration.
The larger the value of $k$, the larger the neighborhood structure that is being explored, hence the better the quality of solutions that is expected to be generated.
However, the time complexity of each iteration grows rapidly with $k$: there are $\Theta(n^k)$ candidate moves, and with suitable data structures we can determine in  $O(n^k)$ if there is an improving move and select one. Thus, the algorithm is feasible only for small values of $k$. For $k=1$, it is the standard FLIP algorithm. Here we are interested in the case $k=2$. 
We will not make any assumption on the pivoting criterion in our results, except that we assume that the algorithm does not choose to flip in any step two nodes when flipping only one of them would produce a strictly better cut. This is a natural property satisfied by any reasonable implementation of 2-FLIP.
For example, one approach (to optimize the time of each iteration) is to first check if there is an improving 1-flip ($n$ possibilities), and only if there is none, proceed to search for an improving 2-flip ($O(n^2)$ possibilities). Clearly any implementation that follows this approach satisfies the above property.
Also, the greedy approach, that examines all $O(n^2)$ possible 1-flips and 2-flips and chooses one that yields the maximum improvement, obviously satisfies the above property.

Our results hold also for the variant of 2-FLIP that uses only 2-flips (no 1-flips). We refer to this variant as {\em Pure 2-FLIP}.

\subsection{Graph Partitioning and the SWAP Algorithm}

In the Graph Partitioning (or Bisection Width) problem, we are given a graph $G$ on $2n$ nodes with weighted edges; the problem is to find a partition of the set $V$ of nodes into two equal-sized subsets $V_1, V_2$ to minimize the weight of the cut.\footnote{Since the weights can be positive or negative, there is no difference between maximization and minimization. The Graph Partitioning problem is usually stated as a minimization problem.} As in the Max Cut problem, in this paper we will assume the graph is complete and the edge weights are in $[-1,1]$.
A simple local search algorithm is the SWAP algorithm: Starting from some initial partition $(V_1, V_2)$ with $n$ nodes in each part, while there is a pair of nodes $u \in V_1, v \in V_2$ whose swap (moving to the other part) decreases the weight of the cut, swap $u$ and $v$. We do not make any assumption on the pivoting rule, i.e. which pair is selected to swap in each iteration if there are multiple pairs whose swap improves the cut. At the end, when the algorithm terminates it produces a locally optimal balanced partition, i.e. one that cannot be improved by swapping any pair of nodes. The SWAP algorithm is clearly a restricted version of Pure 2-FLIP 
(restricted because the initial partition is balanced, and in each step the 2-flip must involve two nodes from different parts of the partition).

The SWAP algorithm is the simplest local search algorithm for the Graph Partitioning problem, but it is a rather weak one, in the sense that the quality of the locally optimal solutions produced may~not be very good. For this reason, more sophisticated local search algorithms have been proposed and are typically used, most notably the Kernighan-Lin algorithm \cite{KL70}, in which a move from a partition to a neighboring partition involves a sequence of swaps. 
If a partition has a profitable swap, then Kernighan-Lin (KL) will perform the best swap; however, if there is no profitable swap then KL  explores a sequence of $n$ greedy steps, selecting greedily in each step the best pair of nodes to swap that have not changed sides before in the current sequence, and if this sequence of swaps produces eventually a better partition, then KL moves to the best such partition generated during this sequence. A related variant, to reduce the time cost of each iteration, was proposed by Fiduccia and Matheyses \cite{FM82}. This idea of guided deep neighborhood search is a powerful method in local search that was introduced first in the \cite{KL70} paper of Kernighan and Lin on Graph Partitioning, and was applied subsequently successfully to the Traveling Salesman Problem and other problems.

\subsection{Smoothed Analysis}

We focus on the $2$-FLIP algorithm from now on.
Under the smoothed complexity model, there is a family  $\mathcal{X}=(\mathcal{X}_e: e\in E_n)$ of probability distributions, one for each edge in $K_n=(V_n,E_n)$. The edge weights $X=(X_e:e\in E_n)$ are drawn independently with $X_e\sim \mathcal{X}_e$.
We assume that each $\mathcal{X}_e$ is a distribution supported on $[-1,1]$ and its density function is bounded from above by a parameter $\phi>0$. (The assumption that the edge weights are in $[-1,1]$ is no loss of generality, since they can be always scaled to lie in that range.) 
Our goal is to bound the number of steps the
  $2$-FLIP algorithm takes to terminate when running on $K_n$ with edge weights $X\sim\mathcal{X}$, in terms of $n$ and the parameter $\phi$.


\subsection{Move Sequences}
  
  

\def\calS{\mathcal{S}}
\def\len{\textsf{len}}
\def\imp{\textsf{imprv}}
We introduce some of the key definitions that will be used in the smoothed analysis of $2$-FLIP.

A \emph{move sequence} $\calS=(\calS_1,\ldots,\calS_\ell)$ is 
  an $\ell$-tuple for some $\ell\ge 1$ such that $\calS_i$ is a subset
  of $V_n$ of size either one or two.
We will refer to the $i$-th move in $\calS$ as a \emph{$1$-move} if 
  $|\calS_i|=1$ and a \emph{$2$-move} if $|\calS_i|=2$,
and write $\len(\calS):=\ell$ to denote its length. Additionally, let $\emph{$1$-move}(\calS)$ and $\emph{$2$-move}(\calS)$ denote the corresponding subsequence of single flip or double flips correspondingly.  
We say a node $u\in V_n$ is \emph{active} in $\calS$
  if $u$ appears in $\calS_i$ for some $i$, and is \emph{inactive} otherwise.
We write $V(\calS)\subseteq V_n$ to denote the set of active nodes in $\calS$.

Given $\gamma_0\in \{\pm 1\}^n$ as the initial configuration, a move sequence $\calS=(\calS_1,\ldots,\calS_\ell)$ naturally induces a sequence of configurations $\gamma_0,\gamma_1,\ldots,\gamma_\ell\in \{\pm 1\}^n$,
  where $\gamma_{i+1}$ is obtained from $\gamma_i$
  by flipping the nodes in $\calS_{i+1}$.
We say $(\gamma_0,\calS)$ is \emph{improving}
  with respect to edge weights $X$ if 
$$
\obj_X(\gamma_{i })>\obj_X(\gamma_{i-1}),\quad\text{for all $i\in [\ell]$}
$$  
and is \emph{$\eps$-improving}
  with respect to edge weights $X$, for some $\eps>0$, if
$$
\obj_X(\gamma_{i })-\obj_X(\gamma_{i-1})\in (0,\eps],\quad\text{for all $i\in [\ell]$.}
$$
For each $i\in [ \ell ]$,
  the change $\obj_{X}(\gamma_{i })-\obj_X(\gamma_{i-1})$ from the $i$-th move $\calS_i$ can be written as follows:
\begin{enumerate}
\item When $\calS_i=\{u\}$,
\begin{equation}\label{eq:1-move-improvement}
\obj_{X}(\gamma_{i })-\obj_X(\gamma_{i-1}) = \sum_{w\in V_n:w\not=u} \gamma_{i-1}(w)\gamma_{i-1}(u)X_{(u,w)}.
\end{equation}
 \input{execution-1flip-tikz}
\item When $\calS_i=\{u,v\}$,
\begin{equation}\label{eq:2-move-improvement}
\obj_{X}(\gamma_{i })-\obj_X(\gamma_{i-1}) = \sum_{w\in V_n: w\not\in\{u,v\}} ( \gamma_{i-1}(w)\gamma_{i-1}(u)X_{(w,u)} + \gamma_{i-1}(w)\gamma_{i-1}(v)X_{(w,v)}).
\end{equation}
\input{execution-2flip-tikz}
\end{enumerate}
For each $i\in [\ell]$, we write $\imp_{\gamma_0,\calS}(i)$ to denote 
  the improvement vector in $\{0,\pm 1\}^{E_n}$ such that 
\begin{equation}\label{eq:general-improvement}
\obj_X(\gamma_{i})-\obj_X(\gamma_{i-1})=\imp_{\gamma_0,\calS}(i)\cdot X.
\end{equation}

Next, let $E(\calS)$ denote the set of edges $(u,v)\in E_n$ such that 
  both $u$ and $v$ are active in $\calS$.
We write $\imp_{\gamma_0,\calS}^* (i)\in \{0,\pm 1\}^{E(\calS)}$ to denote the 
  projection of $\imp_{\gamma_0,\calS}(i)$ on entries that correspond to edges in $E(\calS)$.
We note that $\imp_{\gamma_0,\calS}^*(i)$ 
  only depends on the initial configuration of active nodes $V(S)$ in $\gamma_0$.
Given a (partial) configuration $\tau_0\in \{\pm 1\}^{V(S)}$ 
  of $V(S)$, we let 
$$
\imp_{\tau_0,\calS}(i):=\imp^*_{\gamma_0,\calS}(i)\in \{0,\pm 1\}^{E(\calS)},
$$
where $\gamma_0\in \{\pm 1\}^n$ is an arbitrary (full) configuration that is
  an extension of $\tau_0$.
(To aid the reader we will always use $\gamma$ to denote a full configuration and $\tau$ to denote
a  partial configuration in the paper.)

Note that if $\calS$ is a sequence of moves generated by an execution of the 2-FLIP algorithm then $\calS$ must be improving, because every move must increase the weight of the cut and therefore every $1-$ or $2-$ move is improving. On the other hand, if every move in $\calS$ increases the cut weight by no more than $\epsilon$ then we can not directly guarantee that after $poly(|\calS|,n,1/\epsilon)$ steps the algorithm would  certainly terminate. From probabilistic perspective, in order to provide a smoothed  upper bound on the running time of 2-FLIP method, it suffices to show that it is exponentially small probability for every move in a long enough sequence to incur only a $o(1/poly(n))$ improvement in our objective.

 Indeed, in an idealized scenario where the improvements of different moves of a sequence were disentangled, the event for a linear-length sequence to be at most $\eps-$improving would have exponentially small probability.
Unfortunately, going back to the 2-FLIP algorithm, there could be improving steps that are strongly correlated (as an extreme situation there could be two flips with almost the same improvement vector). Thus, as one may expect the probability exponential decay holds only for linearly  independent $\imp_{\tau_0,\calS}(\cdot)$, introducing the necessity of analysis of the $\rank\left(\{\imp_{\tau_0,\calS}(i)|i\in \calS'\}\right)$, for some neatly chosen subset $\calS'$ of moves from the sequence $\calS$.

\begin{corollary}[\cite{etscheid2015smoothed}]\label{cor:rank}
Let $X_1, . . . , X_m$ be independent real random variables and let $f_i:\mathbb{R}\to[0,\phi]$ for some $\phi>0$ denote the density of $X_i$ for each $i\in[m]$. Additionally, let $\mathcal{C}$ be a collection of k  not necessarily linearly independent integer row vectors, namely $\mathcal{C}=\{ V_1,\cdots,V_k\}$. Then it holds that for any interval $I\subset \mathbb{R}$
\[
\Pr[F_\epsilon]=\Pr\left[
\bigcap_{i\in[k]}\{ V_i \cdot X \in I \}\right]\leq (\phi\len(I))^{\rank(\mathcal{C})}
\]
\end{corollary}

However, one standard issue, which typically occurs with the direct usage of improvement vectors of  sequence's moves, is their dependence also on the initial configuration $\gamma$ of inactive nodes that do not appear in the sequence $\calS$. Their number may be much larger than the rank of the active nodes, and thus considering all their possible initial values in a union-bound will overwhelm the probability  $(\phi \epsilon)^r$. For these reasons, in the literature \cite{etscheid2015smoothed,bibak2019improving,chen2020} more complex combinatorial structures  have been proposed, like pairs of (consecutive) moves of the same node. Interestingly, for the case of 2-FLIP, new challenges have to be overcome due to the 2-move case. To alleviate these harnesses, we introduce the idea of dependent cycles whose role will be revealed in the case that our sequence abounds with 2-moves.


\subsection{Arcs and Cycles}
\begin{definition}\label{def:arc}
An \emph{arc} $\alpha$  in a move sequence $\calS=(\calS_1,\ldots,\calS_\ell)$ is  an ordered pair $(i,j)$  
  with $i< j\in [\ell]$ such that $\calS_i=\calS_j=\{u\}$ for some node $u\in V_n$ and for any $i<k<j$, $\calS_k\not=\{u\}$. 
\end{definition}

Let $\tau_0\in \{\pm 1\}^{V(\calS)}$ be a  configuration
  of active nodes in $\calS$, and let $\tau_0,\tau_1,\ldots,\tau_\ell\in \{\pm 1\}^{V(\calS)}$ be the sequence of configurations
  induced by $\calS$, i.e., $\tau_{i }$ is obtained
  from $\tau_{i-1}$ by flipping nodes in $\calS_i$.
We make the following observation:
\begin{restatable}{lemma}{CancellationOfExternalEdgesArc}\label{lem:concellation-of-external-edges-arc}
For any configuration $\gamma_0\in \{\pm 1\}^n$ that is an 
  extension of $\tau_0$, letting $\gamma_0,\gamma_1,\ldots,\gamma_\ell\in \{\pm 1\}^n$ be the sequence of configurations induced by $\calS$ and letting
$
w[u,i,j]:=\gamma_i(u)\cdot \emph{\imp}_{\gamma_0,\calS}(i)-\gamma_j(u)\cdot 
  \emph{\imp}_{\gamma_0,\calS} (j)
$,
we have that
\[
(w[u,i,j]_e)_{e\in E_n}=
\begin{cases}
(\tau_i(u)\cdot \emph{\imp}_{\tau_0,\calS}(i)-\tau_j(u)\cdot 
  \emph{\imp}_{\tau_0,\calS}(j))_e& \text{ for every entry $e\in E(\calS)$,} \\
0& \text{ otherwise.}
\end{cases}
\]
for any arbitrary choice of $u\in V(\calS)$.
\end{restatable}

\noindent
Motivated by Lemma \ref{lem:concellation-of-external-edges-arc},
 we define for an arc $\alpha=(i,j)$ of a node $u$,
\begin{equation}
\label{eq:arc-improvement}
\imp_{\tau_0,\calS}(\alpha):= \tau_i(u)\cdot \imp_{\tau_0,\calS}(i)-\tau_j(u)\cdot 
  \imp_{\tau_0,\calS}(j)\in \mathbb{Z}^{E(\calS)}.
\end{equation}
Let $\arcs(\calS)$ denote the set of all arcs in $\calS$.
We will be interested in the rank of 
\begin{equation}\label{simpleeq:1}Q_\arcs:=\Big\{\imp_{\tau_0,\calS}(\alpha):\alpha\in \arcs(\calS)\Big\}
\end{equation} 
It is easy to show that the rank does not depend on the choice
  of $\tau_0$ so we will denote it by
  $\rank_{\arcs}(\calS)$.
\begin{lemma}\label{lem:independence-init-rank-arcs}
The rank of the set of vectors in (\ref{simpleeq:1})
  does not depend on the choice of $\tau_0$.
\end{lemma}

\begin{definition}\label{def:cycle}
A \emph{cycle} $C$ in a move sequence $\calS=(\calS_1,\ldots,\calS_\ell)$ is 
  an ordered tuple $C=(c_1,\ldots,c_t)$  for some $t\ge 2$ such that $c_1,\cdots c_t$ are distinct, and
  $\calS_{c_j}=\{u_j,u_{j+1}\}$ for all $j\in [t-1]$
  and $\calS_{c_t}=\{u_t,u_1\}$
  for some nodes $u_1,\ldots,u_t\in V_n$. (Every $\calS_{c_j}$ is a $2$-move. The same vertex may appear in multiple $\calS_{c_j}$'s). 
\end{definition}
\begin{definition}\label{def:dependent-cycle}
Given a configuration $\tau_0\in \{\pm 1\}^{V(\calS)}$,
we say a cycle $C=(c_1,\ldots,c_t)$ in $\calS$ is \emph{dependent} with respect to $\tau_0$ if there exists $b\in \{\pm 1\}^t$ such that 
$$ \text{For all $j\in [t-1]$ we have that }
b_j\cdot \tau_{c_j}(u_{j+1}) + b_{j+1} \cdot \tau_{c_{j+1}}(u_{j+1})=0
\text{ and } b_t \cdot \tau_{c_t}(u_1)+b_1\cdot \tau_{c_1}(u_1)=0,
$$ 
where $\tau_0,\tau_1,\ldots,\tau_\ell$ are 
  configurations induced by $\calS$ starting from $\tau_0$.
\end{definition}

We note that such a vector $b$, if it exists,  it has the form
$b=b_1\cdot\left(1,\cdots,(-1)^{k-1}{\prod_{i\in[2:k]}\tau_{c_{i-1}}(u_i)}{\tau_{c_{i}}(u_i)},\cdots\right)^\top
$
%
and hence
it is unique if we further require $b_1=1$.
After elimination of the above equations, we see the following equivalent criterion:
\begin{remark}[Dependence Criterion]\label{remark:criterion}
A cycle $C$ is \emph{dependent} $\Leftrightarrow$ 
$(-1)^{t}=\tau_{c_t}(u_1)\tau_{c_1}(u_1)\cdot {\prod_{i=2}^{t}\tau_{c_{i-1}}(u_i)}{\tau_{c_{i}}(u_i)}.$ 
\end{remark}

\noindent We will  refer to the unique vector $b\in \{\pm 1\}^t$ as the \emph{cancellation vector} of $C$.
Notice that whether a cycle $C$ in $\calS$ is dependent or
  not actually does not depend on the choice of $\tau_0$.

\begin{lemma}\label{lem:independence-init-dependent-cycles}
If a cycle $C$ of $\calS$ is dependent with respect to some $\tau_0\in \{\pm 1\}^{V(\calS)}$ using $b$ as its cancellation vector, then it is dependent with respect to every configuration $\tau_0'\in \{\pm 1\}^{V(\calS)}$ 
using the same $b$ as its cancellation vector.
\end{lemma}
  
As a result, we can refer to cycles of $\calS$ as dependent cycles without specifying a configuration $\tau_0$; the same holds for  cancellation vectors.
Next we prove a lemma that is similar to Lemma~\ref{lem:concellation-of-external-edges-arc} for arcs:
\begin{restatable}{lemma}{CancellationOfExternalEdgesCycle}\label{lem:concellation-of-external-edges-cycle}
Let $C=(c_1,\ldots,c_t)$ be a dependent cycle of $\calS$ and 
  let $b$ be its cancellation vector.
Then for any configurations $\tau_0\in \{\pm 1\}^{V(\calS)}$ 
  and $\gamma_0\in \{\pm 1\}^n$ such that $\gamma_0$ is an 
  extension of $\tau_0$, letting 
$
w[C]:=\sum_{j\in [t]} b_j\cdot \emph{\imp}_{\gamma_0,\calS}(c_j),
$
we have that
\[
(w[C]_e)_{e\in E_n}=
\begin{cases}
(\sum_{j\in [t]} b_j\cdot \emph{\imp}_{\tau_0,\calS}(c_j))_e
& \text{ for every entry $e\in E(\calS)$,} \\
0& \text{ otherwise.}
\end{cases}
\]
\end{restatable}

\def\cycles{\textsf{cycles}}
\noindent
Given $\tau_0\in \{\pm 1\}^{V(\calS)}$ and a dependent cycle $C$ of $\calS$ with
  $b$ as its cancellation vector,
  we define
\begin{equation}
\label{eq:cycle-improvement}    
\imp_{\tau_0,\calS}(C):=\sum_{j\in [t]} b_j\cdot \imp_{\tau_0,\calS}(c_j).
\end{equation}
Let $\cycles(\calS)$ denote the set of all dependent cycles in $\calS$.
We will be interested in the rank of 
\begin{equation}\label{simpleeq:2}
Q_\cycles:=\Big\{\imp_{\tau_0,\calS}(C):C\in \cycles(\calS)\Big\}
\end{equation}
Similarly we note that the rank
  does not depend on the choice of $\tau_0$ so we  denote it by 
  $\rank_{\cycles}(\calS)$.

\begin{lemma}\label{lem:independence-init-rank-cycles}
The rank of the set of vectors in (\ref{simpleeq:2})
  does not depend on the choice of $\tau_0$.
\end{lemma}

 For the sake of readability we defer the proofs of initial configuration invariance for the rank of improvement vectors of arcs and cycles to Appendix~\ref{sec:rank-invariance}. 
Having defined the sets of $\arcs(\calS)$ and $\cycles(\calS)$, we conclude this section by showing that for a fixed parameter $\eps>0$, a move sequence $\calS$ and an initial configuration $\tau_0\in \{\pm 1\}^{V(\calS)}$, if either $\rank_{\arcs}(\calS)$ or $\rank_{\cycles}(\calS)$ is high,
  then most likely (over $X\sim \calX$) $(\gamma_0,\calS)$ is not $\eps$-improving for every $\gamma_0\in \{\pm 1\}^n$ that is an extension of $\tau_0$.

\begin{lemma}\label{problemma}
Let $\eps>0$. 
With probability at least
$$
1-\left(2\len(\calS)\cdot \phi \eps\right)^{\max
\big(\emph{\rank}_{\emph{\arcs}}(\calS),\hspace{0.05cm}\emph{\rank}_{\emph{\cycles}}(\calS)\big)} 
$$
over $X\sim \calX$, we have that 
$(\gamma_0,\calS)$ is not $\eps$-improving for every $\gamma_0\in \{\pm 1\}^n$ that is an extension of $\tau_0$.
\end{lemma}
\begin{proof}
Let $\mathcal{E}_{moves}$  be the event of a given $(\gamma_0,\calS)$ being  $\eps$-improving  with respect to edge weights $X$, for some fixed $\eps>0$:
\[
\mathcal{E}_{moves}:\left\{\imp_{\gamma_0,\calS}(i)\cdot X\in (0,\eps],\quad\text{for all $i\in [\ell]$}\right\}
\]
where $\imp_{\gamma_0,\calS}(i)$ correspond to the improvement vector of $\calS_i$ move(See (\ref{eq:1-move-improvement}),(\ref{eq:2-move-improvement})). 
Now, notice that  the improvement vector of an arc (See (\ref{eq:arc-improvement})) or of a dependent cycle (See (\ref{eq:cycle-improvement})) can be written as the $\{-1,0,1\}$ sum of all the improvement vectors of either 1 or 2-moves in $\calS$.
Thus, we define the corresponding event for cycles and arcs for a given sequence $(\gamma_0,\calS)$ with respect to edge weights $X$:
\[
\mathcal{E}_{arcs/cycles}:\left\{\imp_{\gamma_0,\calS}(\beta)\cdot X\in [-\len(\calS)\eps,\len(\calS)\eps],\quad\text{for any $\beta\in \arcs(\calS)/\cycles(\calS)$}\right\}
\]
So it is easy to see that  $\mathcal{E}_{moves}$ implies $\mathcal{E}_{arcs/cycles}$, or equivalently $\Pr[\mathcal{E}_{moves}]\leq \min\{\Pr[\mathcal{E}_{arcs}], \Pr[\mathcal{E}_{cycles}]\}$. Thus, by leveraging Corollary~\ref{cor:rank} for vectors in $Q_\arcs$ and $Q_\cycles)$, we get that:
\[\Pr[\text{$(\gamma_0,\calS)$ being an $\eps$-improving sequence}] \leq \left(2\len(\calS)\cdot \phi \eps\right)^{\max
\big(\emph{\rank}_{\emph{\arcs}}(\calS),\hspace{0.05cm}\emph{\rank}_{\emph{\cycles}}(\calS)\big)}
 \]
This finishes the proof of the lemma.\end{proof}

%% file: execution-1flip-tikz
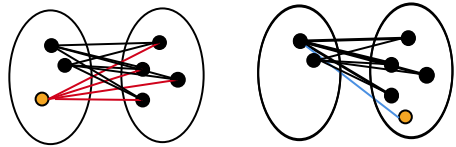
\begin{figure}[h!]
    \centering
\tikzset{every picture/.style={line width=0.75pt}} 

\begin{tikzpicture}[x=0.75pt,y=0.75pt,yscale=-0.5,xscale=0.5]

\draw   (96,117.5) .. controls (96,80.22) and (114.58,50) .. (137.5,50) .. controls (160.42,50) and (179,80.22) .. (179,117.5) .. controls (179,154.78) and (160.42,185) .. (137.5,185) .. controls (114.58,185) and (96,154.78) .. (96,117.5) -- cycle ;
\draw   (209,115.5) .. controls (209,78.22) and (227.58,48) .. (250.5,48) .. controls (273.42,48) and (292,78.22) .. (292,115.5) .. controls (292,152.78) and (273.42,183) .. (250.5,183) .. controls (227.58,183) and (209,152.78) .. (209,115.5) -- cycle ;
\draw  [fill={rgb, 255:red, 0; green, 0; blue, 0 }  ,fill opacity=1 ] (132,85.5) .. controls (132,81.91) and (134.91,79) .. (138.5,79) .. controls (142.09,79) and (145,81.91) .. (145,85.5) .. controls (145,89.09) and (142.09,92) .. (138.5,92) .. controls (134.91,92) and (132,89.09) .. (132,85.5) -- cycle ;
\draw  [fill={rgb, 255:red, 0; green, 0; blue, 0 }  ,fill opacity=1 ] (145.5,105.5) .. controls (145.5,101.91) and (148.41,99) .. (152,99) .. controls (155.59,99) and (158.5,101.91) .. (158.5,105.5) .. controls (158.5,109.09) and (155.59,112) .. (152,112) .. controls (148.41,112) and (145.5,109.09) .. (145.5,105.5) -- cycle ;
\draw  [fill={rgb, 255:red, 0; green, 0; blue, 0 }  ,fill opacity=1 ] (224,140.5) .. controls (224,136.91) and (226.91,134) .. (230.5,134) .. controls (234.09,134) and (237,136.91) .. (237,140.5) .. controls (237,144.09) and (234.09,147) .. (230.5,147) .. controls (226.91,147) and (224,144.09) .. (224,140.5) -- cycle ;
\draw  [fill={rgb, 255:red, 0; green, 0; blue, 0 }  ,fill opacity=1 ] (241,82.5) .. controls (241,78.91) and (243.91,76) .. (247.5,76) .. controls (251.09,76) and (254,78.91) .. (254,82.5) .. controls (254,86.09) and (251.09,89) .. (247.5,89) .. controls (243.91,89) and (241,86.09) .. (241,82.5) -- cycle ;
\draw  [fill={rgb, 255:red, 245; green, 166; blue, 35 }  ,fill opacity=1 ] (122.5,139.5) .. controls (122.5,135.91) and (125.41,133) .. (129,133) .. controls (132.59,133) and (135.5,135.91) .. (135.5,139.5) .. controls (135.5,143.09) and (132.59,146) .. (129,146) .. controls (125.41,146) and (122.5,143.09) .. (122.5,139.5) -- cycle ;
\draw  [fill={rgb, 255:red, 0; green, 0; blue, 0 }  ,fill opacity=1 ] (259,120) .. controls (259,116.13) and (262.13,113) .. (266,113) .. controls (269.87,113) and (273,116.13) .. (273,120) .. controls (273,123.87) and (269.87,127) .. (266,127) .. controls (262.13,127) and (259,123.87) .. (259,120) -- cycle ;
\draw  [fill={rgb, 255:red, 0; green, 0; blue, 0 }  ,fill opacity=1 ] (224,109.5) .. controls (224,105.91) and (226.91,103) .. (230.5,103) .. controls (234.09,103) and (237,105.91) .. (237,109.5) .. controls (237,113.09) and (234.09,116) .. (230.5,116) .. controls (226.91,116) and (224,113.09) .. (224,109.5) -- cycle ;
\draw [color={rgb, 255:red, 208; green, 2; blue, 27 }  ,draw opacity=1 ]   (142,139.5) -- (230.5,140.5) ;
\draw [color={rgb, 255:red, 0; green, 0; blue, 0 }  ,draw opacity=1 ]   (158.5,105.5) -- (230.5,109.5) ;
\draw [color={rgb, 255:red, 208; green, 2; blue, 27 }  ,draw opacity=1 ]   (135.5,139.5) -- (230.5,109.5) ;
\draw [color={rgb, 255:red, 208; green, 2; blue, 27 }  ,draw opacity=1 ]   (135.5,139.5) -- (247.5,82.5) ;
\draw [color={rgb, 255:red, 208; green, 2; blue, 27 }  ,draw opacity=1 ]   (135.5,139.5) -- (266,120) ;
\draw [color={rgb, 255:red, 0; green, 0; blue, 0 }  ,draw opacity=1 ]   (158.5,105.5) -- (230.5,140.5) ;
\draw [color={rgb, 255:red, 0; green, 0; blue, 0 }  ,draw opacity=1 ]   (158.5,105.5) -- (266,120) ;
\draw [color={rgb, 255:red, 0; green, 0; blue, 0 }  ,draw opacity=1 ]   (158.5,104.5) -- (247.5,81.5) ;
\draw    (138.5,85.5) -- (247.5,82.5) ;
\draw    (138.5,85.5) -- (237,109.5) ;
\draw [color={rgb, 255:red, 0; green, 0; blue, 0 }  ,draw opacity=1 ]   (138.5,85.5) -- (266,120) ;
\draw    (138.5,85.5) -- (230.5,140.5) ;
\draw   (347,113.5) .. controls (347,76.22) and (365.58,46) .. (388.5,46) .. controls (411.42,46) and (430,76.22) .. (430,113.5) .. controls (430,150.78) and (411.42,181) .. (388.5,181) .. controls (365.58,181) and (347,150.78) .. (347,113.5) -- cycle ;
\draw   (460,111.5) .. controls (460,74.22) and (478.58,44) .. (501.5,44) .. controls (524.42,44) and (543,74.22) .. (543,111.5) .. controls (543,148.78) and (524.42,179) .. (501.5,179) .. controls (478.58,179) and (460,148.78) .. (460,111.5) -- cycle ;
\draw  [fill={rgb, 255:red, 0; green, 0; blue, 0 }  ,fill opacity=1 ] (383,81.5) .. controls (383,77.91) and (385.91,75) .. (389.5,75) .. controls (393.09,75) and (396,77.91) .. (396,81.5) .. controls (396,85.09) and (393.09,88) .. (389.5,88) .. controls (385.91,88) and (383,85.09) .. (383,81.5) -- cycle ;
\draw  [fill={rgb, 255:red, 0; green, 0; blue, 0 }  ,fill opacity=1 ] (475,136.5) .. controls (475,132.91) and (477.91,130) .. (481.5,130) .. controls (485.09,130) and (488,132.91) .. (488,136.5) .. controls (488,140.09) and (485.09,143) .. (481.5,143) .. controls (477.91,143) and (475,140.09) .. (475,136.5) -- cycle ;
\draw  [fill={rgb, 255:red, 0; green, 0; blue, 0 }  ,fill opacity=1 ] (492,78.5) .. controls (492,74.91) and (494.91,72) .. (498.5,72) .. controls (502.09,72) and (505,74.91) .. (505,78.5) .. controls (505,82.09) and (502.09,85) .. (498.5,85) .. controls (494.91,85) and (492,82.09) .. (492,78.5) -- cycle ;
\draw  [fill={rgb, 255:red, 0; green, 0; blue, 0 }  ,fill opacity=1 ] (510,116) .. controls (510,112.13) and (513.13,109) .. (517,109) .. controls (520.87,109) and (524,112.13) .. (524,116) .. controls (524,119.87) and (520.87,123) .. (517,123) .. controls (513.13,123) and (510,119.87) .. (510,116) -- cycle ;
\draw  [fill={rgb, 255:red, 0; green, 0; blue, 0 }  ,fill opacity=1 ] (475,105.5) .. controls (475,101.91) and (477.91,99) .. (481.5,99) .. controls (485.09,99) and (488,101.91) .. (488,105.5) .. controls (488,109.09) and (485.09,112) .. (481.5,112) .. controls (477.91,112) and (475,109.09) .. (475,105.5) -- cycle ;
\draw    (389.5,81.5) -- (498.5,78.5) ;
\draw    (389.5,81.5) -- (488,105.5) ;
\draw    (389.5,81.5) -- (517,116) ;
\draw    (389.5,81.5) -- (481.5,136.5) ;
\draw  [fill={rgb, 255:red, 245; green, 166; blue, 35 }  ,fill opacity=1 ] (489,157.5) .. controls (489,153.91) and (491.91,151) .. (495.5,151) .. controls (499.09,151) and (502,153.91) .. (502,157.5) .. controls (502,161.09) and (499.09,164) .. (495.5,164) .. controls (491.91,164) and (489,161.09) .. (489,157.5) -- cycle ;
\draw [color={rgb, 255:red, 74; green, 144; blue, 226 }  ,draw opacity=1 ]   (389.5,81.5) -- (489,158.5) ;
\draw   (347,112.5) .. controls (347,75.22) and (365.58,45) .. (388.5,45) .. controls (411.42,45) and (430,75.22) .. (430,112.5) .. controls (430,149.78) and (411.42,180) .. (388.5,180) .. controls (365.58,180) and (347,149.78) .. (347,112.5) -- cycle ;
\draw   (460,110.5) .. controls (460,73.22) and (478.58,43) .. (501.5,43) .. controls (524.42,43) and (543,73.22) .. (543,110.5) .. controls (543,147.78) and (524.42,178) .. (501.5,178) .. controls (478.58,178) and (460,147.78) .. (460,110.5) -- cycle ;
\draw  [fill={rgb, 255:red, 0; green, 0; blue, 0 }  ,fill opacity=1 ] (383,80.5) .. controls (383,76.91) and (385.91,74) .. (389.5,74) .. controls (393.09,74) and (396,76.91) .. (396,80.5) .. controls (396,84.09) and (393.09,87) .. (389.5,87) .. controls (385.91,87) and (383,84.09) .. (383,80.5) -- cycle ;
\draw  [fill={rgb, 255:red, 0; green, 0; blue, 0 }  ,fill opacity=1 ] (396.5,100.5) .. controls (396.5,96.91) and (399.41,94) .. (403,94) .. controls (406.59,94) and (409.5,96.91) .. (409.5,100.5) .. controls (409.5,104.09) and (406.59,107) .. (403,107) .. controls (399.41,107) and (396.5,104.09) .. (396.5,100.5) -- cycle ;
\draw  [fill={rgb, 255:red, 0; green, 0; blue, 0 }  ,fill opacity=1 ] (475,135.5) .. controls (475,131.91) and (477.91,129) .. (481.5,129) .. controls (485.09,129) and (488,131.91) .. (488,135.5) .. controls (488,139.09) and (485.09,142) .. (481.5,142) .. controls (477.91,142) and (475,139.09) .. (475,135.5) -- cycle ;
\draw  [fill={rgb, 255:red, 0; green, 0; blue, 0 }  ,fill opacity=1 ] (492,77.5) .. controls (492,73.91) and (494.91,71) .. (498.5,71) .. controls (502.09,71) and (505,73.91) .. (505,77.5) .. controls (505,81.09) and (502.09,84) .. (498.5,84) .. controls (494.91,84) and (492,81.09) .. (492,77.5) -- cycle ;
\draw  [fill={rgb, 255:red, 0; green, 0; blue, 0 }  ,fill opacity=1 ] (510,115) .. controls (510,111.13) and (513.13,108) .. (517,108) .. controls (520.87,108) and (524,111.13) .. (524,115) .. controls (524,118.87) and (520.87,122) .. (517,122) .. controls (513.13,122) and (510,118.87) .. (510,115) -- cycle ;
\draw  [fill={rgb, 255:red, 0; green, 0; blue, 0 }  ,fill opacity=1 ] (475,104.5) .. controls (475,100.91) and (477.91,98) .. (481.5,98) .. controls (485.09,98) and (488,100.91) .. (488,104.5) .. controls (488,108.09) and (485.09,111) .. (481.5,111) .. controls (477.91,111) and (475,108.09) .. (475,104.5) -- cycle ;
\draw [color={rgb, 255:red, 0; green, 0; blue, 0 }  ,draw opacity=1 ]   (409.5,100.5) -- (481.5,104.5) ;
\draw [color={rgb, 255:red, 0; green, 0; blue, 0 }  ,draw opacity=1 ]   (409.5,100.5) -- (481.5,135.5) ;
\draw [color={rgb, 255:red, 0; green, 0; blue, 0 }  ,draw opacity=1 ]   (409.5,100.5) -- (517,115) ;
\draw [color={rgb, 255:red, 0; green, 0; blue, 0 }  ,draw opacity=1 ]   (409.5,99.5) -- (498.5,76.5) ;
\draw    (389.5,80.5) -- (498.5,77.5) ;
\draw    (389.5,80.5) -- (488,104.5) ;
\draw [color={rgb, 255:red, 0; green, 0; blue, 0 }  ,draw opacity=1 ]   (389.5,80.5) -- (517,115) ;
\draw    (389.5,80.5) -- (481.5,135.5) ;

\end{tikzpicture}

\caption{Example of a \emph{1-move}, showing edges in the cut only.}

\label{fig:flip}

\end{figure}

%% file: execution-2flip-tikz
\begin{figure}[h!]
    \centering
\tikzset{every picture/.style={line width=0.75pt}} 

\begin{tikzpicture}[x=0.75pt,y=0.75pt,yscale=-0.5,xscale=0.5]

\draw   (96,117.5) .. controls (96,80.22) and (114.58,50) .. (137.5,50) .. controls (160.42,50) and (179,80.22) .. (179,117.5) .. controls (179,154.78) and (160.42,185) .. (137.5,185) .. controls (114.58,185) and (96,154.78) .. (96,117.5) -- cycle ;
\draw   (209,115.5) .. controls (209,78.22) and (227.58,48) .. (250.5,48) .. controls (273.42,48) and (292,78.22) .. (292,115.5) .. controls (292,152.78) and (273.42,183) .. (250.5,183) .. controls (227.58,183) and (209,152.78) .. (209,115.5) -- cycle ;
\draw  [fill={rgb, 255:red, 0; green, 0; blue, 0 }  ,fill opacity=1 ] (132,85.5) .. controls (132,81.91) and (134.91,79) .. (138.5,79) .. controls (142.09,79) and (145,81.91) .. (145,85.5) .. controls (145,89.09) and (142.09,92) .. (138.5,92) .. controls (134.91,92) and (132,89.09) .. (132,85.5) -- cycle ;
\draw  [fill={rgb, 255:red, 245; green, 166; blue, 35 }  ,fill opacity=1 ] (145.5,105.5) .. controls (145.5,101.91) and (148.41,99) .. (152,99) .. controls (155.59,99) and (158.5,101.91) .. (158.5,105.5) .. controls (158.5,109.09) and (155.59,112) .. (152,112) .. controls (148.41,112) and (145.5,109.09) .. (145.5,105.5) -- cycle ;
\draw  [fill={rgb, 255:red, 0; green, 0; blue, 0 }  ,fill opacity=1 ] (224,140.5) .. controls (224,136.91) and (226.91,134) .. (230.5,134) .. controls (234.09,134) and (237,136.91) .. (237,140.5) .. controls (237,144.09) and (234.09,147) .. (230.5,147) .. controls (226.91,147) and (224,144.09) .. (224,140.5) -- cycle ;
\draw  [fill={rgb, 255:red, 0; green, 0; blue, 0 }  ,fill opacity=1 ] (241,82.5) .. controls (241,78.91) and (243.91,76) .. (247.5,76) .. controls (251.09,76) and (254,78.91) .. (254,82.5) .. controls (254,86.09) and (251.09,89) .. (247.5,89) .. controls (243.91,89) and (241,86.09) .. (241,82.5) -- cycle ;
\draw  [fill={rgb, 255:red, 245; green, 166; blue, 35 }  ,fill opacity=1 ] (122.5,139.5) .. controls (122.5,135.91) and (125.41,133) .. (129,133) .. controls (132.59,133) and (135.5,135.91) .. (135.5,139.5) .. controls (135.5,143.09) and (132.59,146) .. (129,146) .. controls (125.41,146) and (122.5,143.09) .. (122.5,139.5) -- cycle ;
\draw  [fill={rgb, 255:red, 0; green, 0; blue, 0 }  ,fill opacity=1 ] (259,120) .. controls (259,116.13) and (262.13,113) .. (266,113) .. controls (269.87,113) and (273,116.13) .. (273,120) .. controls (273,123.87) and (269.87,127) .. (266,127) .. controls (262.13,127) and (259,123.87) .. (259,120) -- cycle ;
\draw  [fill={rgb, 255:red, 0; green, 0; blue, 0 }  ,fill opacity=1 ] (224,109.5) .. controls (224,105.91) and (226.91,103) .. (230.5,103) .. controls (234.09,103) and (237,105.91) .. (237,109.5) .. controls (237,113.09) and (234.09,116) .. (230.5,116) .. controls (226.91,116) and (224,113.09) .. (224,109.5) -- cycle ;
\draw [color={rgb, 255:red, 208; green, 2; blue, 27 }  ,draw opacity=1 ]   (142,139.5) -- (230.5,140.5) ;
\draw [color={rgb, 255:red, 208; green, 2; blue, 27 }  ,draw opacity=1 ]   (158.5,105.5) -- (230.5,109.5) ;
\draw [color={rgb, 255:red, 208; green, 2; blue, 27 }  ,draw opacity=1 ]   (135.5,139.5) -- (230.5,109.5) ;
\draw [color={rgb, 255:red, 208; green, 2; blue, 27 }  ,draw opacity=1 ]   (135.5,139.5) -- (247.5,82.5) ;
\draw [color={rgb, 255:red, 208; green, 2; blue, 27 }  ,draw opacity=1 ]   (135.5,139.5) -- (266,120) ;
\draw [color={rgb, 255:red, 208; green, 2; blue, 27 }  ,draw opacity=1 ]   (158.5,105.5) -- (230.5,140.5) ;
\draw [color={rgb, 255:red, 208; green, 2; blue, 27 }  ,draw opacity=1 ]   (158.5,105.5) -- (266,120) ;
\draw [color={rgb, 255:red, 208; green, 2; blue, 27 }  ,draw opacity=1 ]   (158.5,104.5) -- (247.5,81.5) ;
\draw    (138.5,85.5) -- (247.5,82.5) ;
\draw    (138.5,85.5) -- (237,109.5) ;
\draw    (138.5,85.5) -- (266,120) ;
\draw    (138.5,85.5) -- (230.5,140.5) ;
\draw   (347,113.5) .. controls (347,76.22) and (365.58,46) .. (388.5,46) .. controls (411.42,46) and (430,76.22) .. (430,113.5) .. controls (430,150.78) and (411.42,181) .. (388.5,181) .. controls (365.58,181) and (347,150.78) .. (347,113.5) -- cycle ;
\draw   (460,111.5) .. controls (460,74.22) and (478.58,44) .. (501.5,44) .. controls (524.42,44) and (543,74.22) .. (543,111.5) .. controls (543,148.78) and (524.42,179) .. (501.5,179) .. controls (478.58,179) and (460,148.78) .. (460,111.5) -- cycle ;
\draw  [fill={rgb, 255:red, 0; green, 0; blue, 0 }  ,fill opacity=1 ] (383,81.5) .. controls (383,77.91) and (385.91,75) .. (389.5,75) .. controls (393.09,75) and (396,77.91) .. (396,81.5) .. controls (396,85.09) and (393.09,88) .. (389.5,88) .. controls (385.91,88) and (383,85.09) .. (383,81.5) -- cycle ;
\draw  [fill={rgb, 255:red, 0; green, 0; blue, 0 }  ,fill opacity=1 ] (475,136.5) .. controls (475,132.91) and (477.91,130) .. (481.5,130) .. controls (485.09,130) and (488,132.91) .. (488,136.5) .. controls (488,140.09) and (485.09,143) .. (481.5,143) .. controls (477.91,143) and (475,140.09) .. (475,136.5) -- cycle ;
\draw  [fill={rgb, 255:red, 0; green, 0; blue, 0 }  ,fill opacity=1 ] (492,78.5) .. controls (492,74.91) and (494.91,72) .. (498.5,72) .. controls (502.09,72) and (505,74.91) .. (505,78.5) .. controls (505,82.09) and (502.09,85) .. (498.5,85) .. controls (494.91,85) and (492,82.09) .. (492,78.5) -- cycle ;
\draw  [fill={rgb, 255:red, 0; green, 0; blue, 0 }  ,fill opacity=1 ] (510,116) .. controls (510,112.13) and (513.13,109) .. (517,109) .. controls (520.87,109) and (524,112.13) .. (524,116) .. controls (524,119.87) and (520.87,123) .. (517,123) .. controls (513.13,123) and (510,119.87) .. (510,116) -- cycle ;
\draw  [fill={rgb, 255:red, 0; green, 0; blue, 0 }  ,fill opacity=1 ] (475,105.5) .. controls (475,101.91) and (477.91,99) .. (481.5,99) .. controls (485.09,99) and (488,101.91) .. (488,105.5) .. controls (488,109.09) and (485.09,112) .. (481.5,112) .. controls (477.91,112) and (475,109.09) .. (475,105.5) -- cycle ;
\draw    (389.5,81.5) -- (498.5,78.5) ;
\draw    (389.5,81.5) -- (488,105.5) ;
\draw    (389.5,81.5) -- (517,116) ;
\draw    (389.5,81.5) -- (481.5,136.5) ;
\draw  [fill={rgb, 255:red, 245; green, 166; blue, 35 }  ,fill opacity=1 ] (489,157.5) .. controls (489,153.91) and (491.91,151) .. (495.5,151) .. controls (499.09,151) and (502,153.91) .. (502,157.5) .. controls (502,161.09) and (499.09,164) .. (495.5,164) .. controls (491.91,164) and (489,161.09) .. (489,157.5) -- cycle ;
\draw  [fill={rgb, 255:red, 245; green, 166; blue, 35 }  ,fill opacity=1 ] (510,142.5) .. controls (510,138.91) and (512.91,136) .. (516.5,136) .. controls (520.09,136) and (523,138.91) .. (523,142.5) .. controls (523,146.09) and (520.09,149) .. (516.5,149) .. controls (512.91,149) and (510,146.09) .. (510,142.5) -- cycle ;
\draw [color={rgb, 255:red, 74; green, 144; blue, 226 }  ,draw opacity=1 ]   (389.5,81.5) -- (510,143.5) ;
\draw [color={rgb, 255:red, 74; green, 144; blue, 226 }  ,draw opacity=1 ]   (389.5,81.5) -- (489,158.5) ;

\end{tikzpicture}
\caption{Example of a \emph{2-move} , showing edges in the cut only.}

\label{fig:flip}

\end{figure}
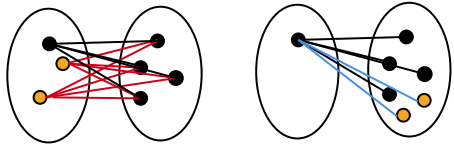

%% file: Main-lemma-and-proof-of-theorem.tex
\section{Main Lemma and the Proof of Theorem \ref{thm:main-high-prob} and Theorem \ref{thm:swap}} \label{sec:main-lemma-and-thm}

We start with the  definition of \emph{\valid} move sequences:

\begin{definition}\label{def:goodsequence}
We say a move sequence $\calS=(\calS_1,\ldots,\calS_\ell)$ is \emph{\valid} if it satisfies the following property: 
For every $i<j\in [\ell]$,  at least one node $w\notin \calS_i$  appears an odd number of times in $\calS_{i},\ldots,\calS_{j}$.
\end{definition}

\begin{lemma}\label{lem:goodseq}
The move sequence generated by 2-FLIP (or by pure 2-FLIP), for any pivoting rule and any instance, is \valid.
\end{lemma}
\begin{proof}
Let $\calS$ be a move sequence generated by 2-FLIP (or pure 2-FLIP). If there are two moves $\calS_i, \calS_j$, $i \leq j$, such that no node appears an odd number of times in $\calS_i, \dots , \calS_j$, then the configurations before $\calS_i$ and after $\calS_j$ are the same, contradicting the fact that all the moves increase the weight of the cut. Therefore, 
the set $O$ of nodes that appear an odd number of times in $\calS_i, \dots , \calS_j$ is nonempty.
Suppose that $i<j$ and $O \subseteq \calS_i$.
If $O = \calS_i$, then the set of nodes that appear an odd number of times in $\calS_{i+1}, \dots , \calS_j$  would be empty,
a contradiction to the above property.
Therefore, $O \neq \calS_i$.

In the case of pure 2-FLIP, since all moves are 2-flips, $O$ has even size, and hence
$O \neq \emptyset$ and $O \neq \calS_i$ imply the claim.

In the case of 2-FLIP, $O \neq \emptyset$, $O \neq \calS_i$ and $O \subseteq \calS_i$ imply that $\calS_i$ has size 2, say $\calS_i = \{ u,v \}$ and $O = \{u\}$ or $O = \{v\}$.
If $O = \{u\}$ then the configuration $\gamma_j$ differs from $\gamma_{i-1}$ only in that node $u$ is flipped. Thus, at configuration $\gamma_{i-1}$, flipping node $u$ results in configuration $\gamma_j$ which has strictly greater cut than the configuration $\gamma_i$ that results by flipping the pair $\{u,v\}$, contradicting our assumption about 2-FLIP.
A similar argument holds if $O = \{v\}$.
In either case we have a contradiction to $O \subseteq \calS_i$. The claim follows.
\end{proof}

Given a move sequence $\calS$, a \emph{window} $W$ of $\calS$ is a substring of $\calS$, i.e., $W=(\calS_i,\ldots,\calS_j)$ for some $i<j\in [\ell]$ (so $W$ itself is also a move sequence).
Our main technical lemma below shows that every long enough \valid move sequence has a window $W$ such that either $\rank_{\arcs}(W)$ or $\rank_{\cycles}(W)$ is large relative to $\len(W)$.

\begin{lemma}\label{maintech}
Let $\calS$ be a \valid move sequence with $\emph{\len}(\calS)\ge n\log^{10} n$.
Then $\calS$ has a window $W$ such that 
\begin{equation}\label{hiddenconst}
\max\Big(\emph{\rank}_{\emph{\arcs}}(W),\emph{\rank}_{\emph{\cycles}}(W)\Big)\ge \Omega\left(\frac{\emph{\len}(W)}{\log^{10} n}\right).
\end{equation}
\end{lemma}

We prove Lemma \ref{maintech} when $\calS$ consists of $2$-moves only in Section \ref{sec:good-window} and \ref{sec:finding-cycles}, and then generalize the proof to work with general move sequences in Section \ref{sec:general-case}.
Assuming Lemma \ref{maintech}, we use it to establish our main theorem, restated below:

\def\calF{\mathcal{F}}

\mainthm*

\begin{proof}
Let $\eps>0$ be specified as follows:
$$
\eps:=\frac{1}{\phi n^{c_1\log^{10} n}}
$$
for some large enough constant $c_1>0$ to be specified later.
We write $\calF$ to denote the following event on the draw of edge weights $X\sim \calX$: 
\begin{flushleft}\begin{quote}
Event $\calF$: 
  For every move sequence $W$ of length at most $n\log^{10} n$ such that (letting $a>0$ be the constant hidden in (\ref{hiddenconst}))
\begin{equation}\label{eqhehe}
\max\Big(\rank_{\arcs}(W),
\rank_{\cycles}(W)\Big)\ge \frac{a}{\log^{10} n}\cdot \len(W).
\end{equation}
and every configuration $\gamma_0\in \{\pm 1\}^n$,
  $(\gamma_0,W)$ is not $\eps$-improving with respect to $X$.
\end{quote}\end{flushleft}
 We break the proof of the theorem into two steps. First we show that $\calF$ occurs with probability at least $1-o_n(1)$.
Next we show that when $\calF$ occurs, any implementation
  of $2$-FLIP must terminate in at most $\phi n^{O(\log^{10} n)}$ many steps.

For the first step, we apply Lemma \ref{problemma} on every 
  move sequence $W$ of length at most $n\log^{10} n$ that satisfies 
  (\ref{eqhehe}) and every configuration $\tau_0\in \{\pm 1\}^{V(W)}$.
It then follows from a union bound that $\calF$ occurs with probability at least
$$
1-\sum_{\ell\in [n\log^{10} n]} n^{2\ell}\cdot 2^{2\ell}\cdot 
  (\ell \phi\eps)^{\frac{a \ell}{\log^{10} n}}
=1- \sum_{\ell\in [n\log^{10} n]} \left((2n)^{\frac{2\log^{10} n}{a}}
  \cdot \ell \phi \eps\right)^{\frac{a\ell}{\log^{10} n}}=1-o_n(1),
$$
where the factor $n^{2\ell}$ is an upper bound for the number of
  $W$ of length $\ell$ and $2^{2\ell}$ is an upper bound for the number
  of configurations $\tau_0$ since $|V(W)|\le 2\ell$.
The last equation follows by setting the constant $c_1$ sufficiently large.

For the second step, we assume that the event $\calF$ occurs, and
  let $\gamma_0,\ldots,\gamma_N\in \{\pm 1\}^n$ be a sequence of $N$ configurations that is the result of the execution of some implementation of $2$-FLIP under $X$.
Let $\calS=(\calS_1,\ldots,\calS_N)$ denote the move sequence induced by
  $\gamma_0,\ldots,\gamma_N$.  
So $(\gamma_0,\calS)$ is improving with respect to edge weights $X$.
By Lemma \ref{lem:goodseq}, $\calS$ is a \valid sequence.

We use the event $\calF$ to bound the length of $\calS$.
Because of $\calF$ and that $\calS$ is a \valid move sequence, 
  it follows from Lemma \ref{maintech} that the objective function
  gets improved by at least $\eps$ for every $n\log^{10} n$ consecutive moves in $\calS$.
Given that the objective function lies between $[-n^2,n^2]$, we have 
\begin{equation}
 \len(\calS)\le  n\log^{10} n\cdot \frac{2n^2}{\eps}\le \phi n^{O(\log^{10} n)}.
\label{eq:bound-of-length}
\end{equation}
 
This finishes the proof of the theorem.
\end{proof}

\begin{corollary}
Under the same setting of Theorem~\ref{thm:main-high-prob},  the same result holds for Pure 2-FLIP. 
\end{corollary}
\begin{proof}
The only property of the sequence of moves used in the proof of Theorem~\ref{thm:main-high-prob} is that it is a valid sequence, and this property holds for Pure 2-FLIP as well.  
\end{proof}

Notice that by twining the constant $c_1$ in the exponent,  we can control the tail-bound of the failure probability.
Thus, we can strengthen our proof to get the same bound for the expected number of steps needed to terminate as in  the standard smoothed analysis prototype :
{
\begin{corollary}\label{thm:main-exp}
Under the same setting of Theorem~\ref{thm:main-high-prob}, any implementation of the 2-FLIP algorithm (or Pure 2-FLIP) takes at most $\phi n^{O\left( \log^{10} n\right)}$ many steps to terminate on expectation.
\end{corollary}
\begin{proof}
We let $\calF_\eps$ denote the event $\calF$ in the proof of Theorem~\ref{thm:main-high-prob} with a specified $\eps > 0$. Let
$\eps_0=1/\left(\phi\cdot n^{c_1\log^{10} n} \right)$, where $c_1>0$ is 
a constant to be fixed shortly. For any $\eps<\eps_0$, we have that
\begin{align*}\Pr[\lnot \calF_{\eps}]&\leq \sum_{\ell\in [n\log^{10} n]} n^{2\ell}\cdot 2^{2\ell}\cdot 
(\ell \phi\eps)^{ \max\Big(\emph{\rank}_{\emph{\arcs}}(W),\emph{\rank}_{\emph{\cycles}}(W)\Big)}\leq \sum_{\ell\in [n\log^{10} n]} n^{2\ell}\cdot 2^{2\ell}\cdot 
(\ell \phi\eps \cdot \frac{\eps_0}{\eps_0})^{\lceil \frac{a \ell}{\log^{10} n}\rceil}\\
&\leq \sum_{\ell\in [n\log^{10} n]} n^{2\ell}\cdot 2^{2\ell}\cdot 
\left(\ell \phi\eps_0\right)^{\lceil \frac{a \ell}{\log^{10} n}\rceil} \cdot  (\frac{\eps}{\eps_0})^{\lceil \frac{a \ell}{\log^{10} n}\rceil}\leq \sum_{\ell\in [n\log^{10} n]} \frac{1}{n^{3}} \left(\frac{\eps}{\eps_0}\right)^{\lceil \frac{a \ell}{\log^{10} n}\rceil}\leq\frac{\eps}{c_2n\eps_0}
\end{align*}
where $c_1=2+6/a$ (letting $a>0$ be the constant hidden in (\ref{hiddenconst})) and $c_2=10^8$.
From the proof of Theorem~\ref{thm:main-high-prob}, conditionally to the event $\calF_\eps$ for any $\eps\leq \eps_0$, $\len(\calS)\leq L(\eps):= \tfrac{2n^3\log^{10}n}{\eps}$. Notice that for any $\eps(\rho):=\eps_0/\rho$,  $L(\eps(\rho))= \rho L(\eps_0)$.
Thus, the probability that $\len(\calS)$ is larger than $cL(\eps)$ for any $\rho\ge 1$ is
$$\Pr[\len(\calS)> \rho L(\eps)]\le \Pr[\lnot F_{\eps_0/\rho}]\le   \frac{1/\rho}{  c_2 \cdot n}.$$
Note that $L$ is always trivially bounded by the total number of configurations, $2^n$. Therefore, we have\vspace{0.1cm}
\begin{align*}
\mathbb{E}[\len(\calS)]&=\sum_{s=1}^{2^n} \Pr[\len(\calS)\ge s]  \le 
 \sum_{s=1}^{\lceil L(\eps_0)\rceil} \Pr[L\ge s] 
 +\sum_{s=\lceil L(\eps_0)\rceil}^{2^n} \Pr[L\ge s] \leq  L(\eps_0) +\sum_{s=\lceil L(\eps_0)\rceil}^{2^n} \Pr[L\ge s]
 \\&\leq  L(\eps_0) + \sum_{s=\lceil L(\eps_0)\rceil}^{2^n}\Pr[L\ge  \tfrac{s}{ L(\eps_0)}\cdot L(\eps_0)] \leq L(\eps_0) + \sum_{s=\lceil L(\eps_0)\rceil}^{2^n}\frac{L(\eps_0)/s}{ c_2n}
= O(n)\cdot L(\eps_0)=\phi n^{O(\log^{10} n)}.\vspace{0.1cm} 
\end{align*} 
This finishes the proof of Corollary~\ref{thm:main-exp}.
\end{proof}
}
The same results hold for the Graph Partitioning problem and the SWAP neighborhood.

\swapthm*

\begin{proof}
Every move sequence $\calS$ generated by SWAP (for any pivoting rule, any weights, and any initial balanced partition) is also a legal move sequence for Pure 2-FLIP on the same instance, except that the sequence may be incomplete for Pure 2-FLIP, that is, the final partition may not be locally optimal for Pure 2-FLIP, since there may be a  2-move 
(but not a swap) that improves the weight of the cut (the resulting partition would not be balanced), and Pure 2-FLIP would continue and produce a longer sequence.
Hence, the number of steps of SWAP is upper bounded by the number of steps of Pure 2-FLIP, and thus it is at most $\phi n^{O\left( \log^{10} n\right)}$ with probability $1-o_n(1)$,
as well as in expectation.
\end{proof}

%% file: Lemma-for-finding-good-windows.tex
{
\section{Windows in a Valid Sequence of 
$2$-Moves}
\label{sec:good-window}

\def\order{\textsf{order}}

We will start with the proof of Lemma \ref{maintech} for the case when $\calS$ consists of $2$-moves only in Section \ref{sec:good-window} and \ref{sec:finding-cycles}, and generalize it to deal with general move sequences in Section \ref{sec:general-case}.

We start with a combinatorial argument about sets and subsequences of $[N]$, where $N=poly(n)$ for any polynomial at $n$.
Let $I$ be a subset of $[N]$ with $|I|\ge \log^{10} n$.
  Intuitively, later in this section $I$ will be chosen to be the set $I_u$ representing the appearances of some frequently appeared active node $u\in V(\calS)$ in a move sequence $\calS$.
  %
  %
  We will write $\order(i)$ to denote the order of $i\in I$. In other words, the smallest index in $I$ has order $1$ and the largest index in $I$ has order $|I|$. To give an example if $I=\{2,5,9,11\}$ then $\order(2)=1, \order(5)=2$ and so on. Let $\delta=0.01$. We start by quantifying how much large windows centered around an index $i\in I$ should be to cover the majority of a set $I$. Afterwards, we present the combinatorial lemmas about subset $I$.



\begin{definition}Let $I\subseteq [N]$.
We say an index $i\in I$ is \emph{$\ell$-good} for some positive integer $\ell$ if
$$\Big[i-\lceil(1+2\delta)L'\rceil:i+\lceil(1+2\delta)L'\rceil\Big]\subseteq [N],\quad \text{where $ L'=\lceil(1+\delta)^{\ell-1}\rceil$}$$
and $I$ satisfies  
\begin{equation}
\Big|I\cap [i-L':i+L']\Big|\ge \log^3 n\quad\text{and}\quad
\Big|I\cap \Big[i-\lceil(1+2\delta)L'\rceil:i+\lceil(1+2\delta)L'\rceil\Big]\Big|\le \log^7 n.
\end{equation}
If there exists no such constant $\ell$, we call the corresponding index \emph{bad}.
\label{def:good-index}
\end{definition}
\begin{figure}[h!]
\centering
\input{L-intuition-tikz}
\caption{Two examples of windows whose intersection in $I$ is between $[\log^3n, \log^7n]$. }
\end{figure}
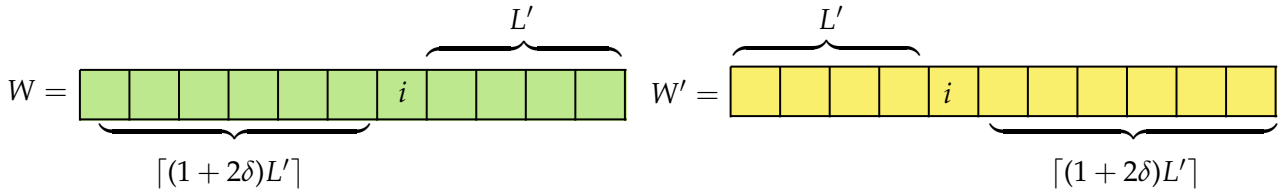

\begin{remark}\label{rem:windows}
Some motivation behind the definition \ref{def:good-index}:
When $i\in I$ is $\ell$-good (letting $L=L'+\lceil(1+2\delta)L'\rceil+1$ and $L'=\lceil(1+\delta)^{\ell-1}\rceil$), it implies that all the $\lceil(1+2\delta)L'\rceil-L'\ge 2\delta L'=\Omega(L)$\footnote{Indeed, by definition $L=L'+1+\lceil(1+2\delta)L'\rceil\leq 2((1+\delta)L'+1)\le 2((1+\delta)L'+1)\le 2(2+\delta)L'$, since $L'\ge 1$} windows $W$ of length $L$,
  i.e., those start at $i-\lceil(1+2\delta)L'\rceil,\ldots,i-L'$, satisfy
$$
i\in W,\quad
\big|I\cap W\big|\ge \log^3 n\quad\text{and}\quad
\big|I\cap W\big|\le \log^7 n.
$$
\end{remark}

\begin{remark}
By definition~\ref{def:good-index}, $\ell$ can get at most $\log_{1+\delta}N=\Theta(\log n)$, for any $N=poly(n)$.
\end{remark}
\begin{lemma}\label{lem:lgood}
Suppose $I$ is a subset of $[N]$ with $|I|\ge {\log^{8} n}$.
Then at least a $(1-O(1/\log n))$-fraction of $i\in I$ is $\ell_i$-good for some nonnegative integer $\ell_i$.
\end{lemma}
\begin{proof}
We start by defining an $\ell_i$ for each $i\in I$ (except for the smallest $\lceil\log^7 n\rceil$ indices and  the largest $\lceil\log^7 n\rceil$ indices in $I$, which are negligible since $|I|\ge \log^{8} n$)
  and then show that most $i\in I$ is $\ell_i$-good.
 Let $I'$ be the subset of $I$ after removing the smallest $\lceil\log^7 n\rceil$ indices and  the largest $\lceil\log^7 n\rceil$ indices in $I$. 
For each $i\in I'$, let
\begin{itemize}
    \item 
 $j\in I'$ be the index in $I'$ of order $\order(i)-\lfloor\log^7 n/2\rfloor+1$,
\item $k\in I'$ be the index in $I'$ of order $\order(i)+\lfloor\log^7 n/2\rfloor-1$. 
\item $\Delta$ be the minimum distance between index $i$ and indices $j,k$, $\Delta=\min(i-j,k-i)$ and
\item $\ell_i$ be the largest integer such that  $\lceil(1+2\delta)\cdot (1+\delta)^{\ell_i-1} \rceil \le \Delta-2$.
\end{itemize}
Using the fact that for any real positive number $x$, it holds that  $ 
0 \leq \lceil(1+2\delta)\cdot \lceil x\rceil\rceil -
\lceil(1+2\delta)\cdot x \rceil 
\leq 2
$, we get that:
$$
\begin{cases}
\lceil(1+2\delta)\cdot \lceil(1+\delta)^{\ell_i-1}\rceil\rceil \le \Delta\\
\lceil(1+2\delta)\cdot (1+\delta)^{\ell_i}\rceil> \Delta-2
\end{cases}.
$$
For the rest of the proof, let $L_i'=\lceil(1+\delta)^{\ell_i-1}\rceil$.
It follows from the choice of $\ell_i$ that 
\begin{align} \label{hehe3}
\big|I\cap [i-\lceil (1+2\delta)L_i'\rceil :i+\lceil (1+2\delta)L_i'\rceil ]\big|&\le 2(\lfloor\log ^7 n/2\rfloor-1)+1\le \log ^7 n \\ 
\big|I\cap [i-\lceil (1+\delta)\cdot (1+2\delta)L_i'\rceil:i+\lceil (1+\delta)\cdot (1+2\delta)L_i'\rceil ]\big|&\ge (\lfloor\log ^7 n/2\rfloor-1)+1 
-4 \label{explainmanolis}
\end{align}
For (\ref{explainmanolis}), we use the observation that left-hand side is larger than $\big|I\cap [i-(\Delta-2):i+(\Delta-2)]\big|\ge \big|I\cap [i-\Delta:i+\Delta]\big| -4  $. 
Using $(1+\delta)(1+2\delta)\le 1+4\delta$ with $\delta=0.01$, the second inequality implies
\begin{equation}\label{hehe2}
\big|I\cap [i-\lceil (1+4\delta)L_i'\rceil :i+\lceil (1+4\delta)L_i'\rceil ]\big|\ge \lfloor \log^7 n/2\rfloor -4 .
\end{equation}
On the other hand, (\ref{hehe3}) implies that $i\in I'$ is $\ell_i$-good unless
\begin{equation}\label{hehe1}
\big|I\cap [i-L_i':i+L_i']\big|\le \log^3 n.
\end{equation}

 Assume now for a contradiction that the number of $i\in I'$ that are bad is at least $|I|/\log n$. 
 Additionally, for any possible exponent $\ell\in[\log_{1+\delta}N]$, let $\mathcal{R}_\ell$ be the set of the indices $i$ that are not $\ell$-good:
 $$\mathcal{R}_\ell:=\{i\in I' \text{ s.t } i \text{ is not } \ell-\text{good} \} \text{ \& }\ell^*=\displaystyle\mathop{\operatorname{argmax}}_{\ell\in [\log_{1+\delta}N]}|\mathcal{R}_\ell|$$
Then, it holds that 
$$|\mathcal{R}_{\ell^*}|\ge \frac{|I|/\log n 
}{\log_{1+\delta} N}=\Omega(|I|/\log^2 n),$$
where we use the facts that $|I|\ge {\log^{8} n}$ and $N=poly(n)$.
We define then $L^*=\lceil (1+\delta)^{{\ell^*}-1}\rceil $
and for each $\rho\in \mathcal{R}_{\ell^*}$ we let
$$
B_\rho = \big(I\cap [\rho- \lceil (1+4\delta)L^*\rceil,\rho-L^*]\big)\cup
\big(I\cap [\rho+L^*,\rho+\lceil (1+4\delta)L^*\rceil ]\big).
$$
Note now that when an index $i$ is not $\ell$-good, we have from (\ref{hehe1}) and (\ref{hehe2}) that
\begin{align*}
    \Big|\big(I\cap [i-\lceil (1+4\delta)L_i'\rceil :i-L_i']\big)\cup
\big(I\cap [i+L_i':i+\lceil (1+4\delta)L_i'\rceil ]\big)\Big|&\ge \lfloor \log^7 n/2\rfloor-4 -\log^3n= \Omega(\log^7 n).
\end{align*}
Hence, we have that $|B_\rho|\ge \Omega(\log^7 n)$ for every $\rho\in \mathcal{R}_{\ell^*}$ and thus,
$$
\sum_{\rho\in \mathcal{R}_{\ell^*}} |B_\rho|\ge \Omega\left(\frac{|I|}{\log^2 n}\right)\cdot \Omega(\log^7 n)
=\Omega\big(|I|\log^5 n\big).
$$
On the other hand, we can prove the following claim:
\begin{claim}\label{claim:hehe}
For any $i\in I$, the number of $\rho\in  \mathcal{R}_{\ell^*}$ such that $i\in B_\rho$  is at most $O(\log ^3 n)$.
\end{claim}
It follows then from the claim that
$$ \text{ For any $i\in I$, we get }
|\{\rho\in \mathcal{R}_{\ell^*}\big|i\in B_\rho\}|=O(\log^3 n)\Rightarrow \sum_{\rho\in \mathcal{R}_{\ell^*}} |B_\rho|\le |I|\cdot O(\log^3 n),
$$ 
which leads to a contradiction.
\end{proof}
\begin{proof}[Proof of Claim \ref{claim:hehe}]
Fix any $i\in I$.
Let us assume then that $\rho$ be a $\rho\in \mathcal{R}_{\ell^*}$ such that $i\in B_\rho$. We prove that the number of $\rho\in \mathcal{R}_{\ell^*}$ with $\rho>i$ and $i\in B_\rho$ is at most $O(\log^3 n)$; the case with $\rho<i$ is symmetric. If no such $\rho$ exists then  the claim is trivially true. Hence, let's assume that such one exists with $\rho>i$.
Given that $i\in B_\rho$, we have that $i\in [\rho-\lceil (1+4\delta)L^*\rceil ,\rho-L^*]$ and we also have 
\begin{equation}\label{hehe5}
\big|I\cap [\rho-L^*,\rho+L^*]\big|\le \log^3 n.
\end{equation}
On the other hand, every other $\rho'\in \mathcal{R}_{\ell^*}$ that satisfies $\rho'>i$ and $i\in B_{\rho'}$ also has the property that $i\in [\rho'-\lceil (1+4\delta)L^*\rceil ,\rho'-L^*]$ and thus,
$
\rho'\in [i+L^*,i+\lceil (1+4\delta)\rceil L^*].
$
But combining this with $i\in [\rho-(1+4\delta)L^*,\rho-L^*]$ we have
$$
\rho'\le i+\lceil (1+4\delta)L^*\rceil \le \rho+\lceil (1+4\delta)L^*\rceil -L^*
$$
and 
$$
\rho'\ge i+L^*\ge \rho-\lceil (1+4\delta)L^*\rceil +L^*.
$$
So $\rho'\in [\rho-\lceil (1+4\delta)L^*\rceil +L^*,\rho+\lceil (1+4\delta)L^*\rceil -L^*]\subseteq [\rho-L^*,\rho+L^*]$ and by
  (\ref{hehe5}) the number of such $\rho'$ is no more than $\log^3 n$.
\end{proof}



Now we return to work on our problem and an arbitrary move sequence $\calS=(\calS_1,\ldots,\calS_N)$.
Let $W$ be a window (move sequence) of $\calS$.
For each active node $u\in V(W)$, we
write $\#_{W}(u)$ to denote the number of occurrences of $u$ in $W$. 
The main result in this section is the following lemma:



\begin{lemma}\label{lem:sequence-selection}
Let $\calS$ be a move sequence of length $N= n\log^{10} n$ that consists of $2$-moves only. There exists a positive integer $L$ such that 
$\calS$ has at least $\Omega((N-L+1)/\log n)$ many windows $W=(W_1,\ldots,W_L)$ of length $L$ such that at least $\Omega(L/\log n)$ moves $W_i=\{u,v\}$ of $W$ satisfy
\begin{equation}\label{heheeq}
\log^{3} n\le \#_W(u)\le \log^{7} n\quad\text{and}\quad
\#_W(v)\ge \log^{3} n
\end{equation}
\end{lemma}
\begin{proof}
For each node $u\in V(\calS)$ we write $I_u\subseteq [N]$ to denote the set of $i\in [N]$ with $u\in \calS_i$.
We say the $i$-th move $\calS_i=\{u,v\}$ is $\ell$-good for some positive integer $\ell$ if $i$ is $\ell_1$-good in $I_u$ and $i$ is $\ell_2$-good in $I_v$ for some positive integers $\ell_1,\ell_2$ such that $\ell=\max(\ell_1,\ell_2)$.

Let $\calS_i=\{u,v\}$. Then we consider the following cases:
\begin{flushleft}\begin{enumerate}
\item Either $|I_u|$ or $|I_v|$ is smaller than $\log^8 n$: 
Given that no more than $n\log^8n$ moves can contain a vertex that appears less than $\log^8 n$ times in the sequence, we have
the number of such $i$ is at most $$ n\log^8 n =o(N/\log n);$$
\item $|I_u|,|I_v|\ge \log^8 n$ but either $u$ is not $\ell_1$-good for any $\ell_1$ or $v$ is not $\ell_2$-good for any $\ell_2$: By Lemma \ref{lem:lgood}, the number of such $i$ is at most
(using $\sum_u |I_u|=2N$)
$$
\sum_{u: |I_u|\ge \log^8 n} \frac{|I_u|}{\log n}\le \frac{2N}{\log n}.
$$
\item Otherwise, $\calS_i$ is $\ell$-good by setting $\ell=\max(\ell_1,\ell_2)$. 
\end{enumerate}\end{flushleft}
Thus, the number of $i\in [N]$ such that the $\calS_i$ is $\ell$-good for some $\ell$ is at least $(1-3/\log n)N$.



Given that $\ell$ is at most $O(\log N)=O(\log n)$,
there exists a positive integer $\ell$ such that the number of moves in $\calS$ that are $\ell$-good is at least $\Omega(N/\log n)$. 
Let $L'=\lceil (1+\delta)^{\ell-1}\rceil$ and $$L=L'+\big\lceil(1+2\delta)L'\big\rceil+1.$$
For any move $\calS_i=\{u,v\}$ that is $\ell$-good, it is easy to verify that there are $\Omega(L)$ windows $W$ of length $L$ that contain $i$ and satisfy \eqref{heheeq} (See Remark~\ref{rem:windows}) .

Let's pick a window $W$ of $\calS$ of size $L$ uniformly at random; note that there are $N-L+1$ many such windows in total.
Let $X$ be the random variable that denotes the number of moves in $W$ that satisfy (\ref{heheeq}).
Given that the number of moves that are $\ell$-good is at least $\Omega(N/\log n)$, we have
$$\textbf{E}\big[X\big]\ge\Omega\left(\frac{N}{\log n}\right)\cdot \frac{\Omega(L)}{N}=\Omega\left(\frac{L}{\log n}\right).$$ 
Let $a$ be the constant hidden above. Given that we always have $X\le L$, we have 
\begin{equation}\label{heheeq3}
\Pr\left[
X\ge \frac{aL}{2\log n}
\right]\ge \frac{a}{2\log n}
\end{equation}
since otherwise,
$$
\textbf{E}\big[X\big]\le 
\frac{a}{2\log n}\cdot L+ \left(1-\frac{a}{2\log n}\right)\cdot \frac{aL}{2\log n}<\frac{aL}{\log n}.
$$
a contradiction.
The lemma then follows directly from (\ref{heheeq3}).
\end{proof}
}

%% file: L-intuition-tikz
\tikzset{every picture/.style={line width=0.75pt}} 

\begin{tikzpicture}[x=0.75pt,y=0.75pt,yscale=-0.95,xscale=0.95]

\draw  [draw opacity=0][fill={rgb, 255:red, 126; green, 211; blue, 33 }  ,fill opacity=0.51 ] (131.58,114.89) -- (421.79,114.89) -- (421.79,141.85) -- (131.58,141.85) -- cycle ; \draw   (131.58,114.89) -- (131.58,141.85)(157.88,114.89) -- (157.88,141.85)(184.19,114.89) -- (184.19,141.85)(210.49,114.89) -- (210.49,141.85)(236.79,114.89) -- (236.79,141.85)(263.1,114.89) -- (263.1,141.85)(289.4,114.89) -- (289.4,141.85)(315.7,114.89) -- (315.7,141.85)(342.01,114.89) -- (342.01,141.85)(368.31,114.89) -- (368.31,141.85)(394.61,114.89) -- (394.61,141.85)(420.92,114.89) -- (420.92,141.85) ; \draw   (131.58,114.89) -- (421.79,114.89)(131.58,141.19) -- (421.79,141.19) ; \draw    ;
\draw  [draw opacity=0][fill={rgb, 255:red, 248; green, 231; blue, 28 }  ,fill opacity=0.64 ] (477.03,113.15) -- (767.25,113.15) -- (767.25,140.11) -- (477.03,140.11) -- cycle ; \draw   (477.03,113.15) -- (477.03,140.11)(503.33,113.15) -- (503.33,140.11)(529.64,113.15) -- (529.64,140.11)(555.94,113.15) -- (555.94,140.11)(582.25,113.15) -- (582.25,140.11)(608.55,113.15) -- (608.55,140.11)(634.85,113.15) -- (634.85,140.11)(661.16,113.15) -- (661.16,140.11)(687.46,113.15) -- (687.46,140.11)(713.76,113.15) -- (713.76,140.11)(740.07,113.15) -- (740.07,140.11)(766.37,113.15) -- (766.37,140.11) ; \draw   (477.03,113.15) -- (767.25,113.15)(477.03,139.45) -- (767.25,139.45) ; \draw    ;

\draw (299.18,119.47) node [anchor=north west][inner sep=0.75pt]    {$i$};
\draw (357.68,79.45) node [anchor=north west][inner sep=0.75pt]    {$L'$};
\draw (91.41,118.6) node [anchor=north west][inner sep=0.75pt]    {$W=$};
\draw (140.43,140.54) node [anchor=north west][inner sep=0.75pt]    {$\underbrace{\ \ \ \ \ \ \ \ \ \ \ \ \ \ \ \ \ \ \ \ \ \ \ \ \ \ \ \ \ \ \ \ }$};
\draw (420.48,110.76) node [anchor=north west][inner sep=0.75pt]  [rotate=-180.05]  {$\underbrace{\  \ \ \ \ \ \ \ \ \ \ \ \ \ \ \ \ \ \ \ \ \ \ }$};
\draw (588.52,119.76) node [anchor=north west][inner sep=0.75pt]    {$i$};
\draw (522.01,79.45) node [anchor=north west][inner sep=0.75pt]    {$L'$};
\draw (433.11,118.6) node [anchor=north west][inner sep=0.75pt]    {$W'=$};
\draw (613.27,140.54) node [anchor=north west][inner sep=0.75pt]    {$\underbrace{\ \ \ \ \ \ \ \ \ \ \ \ \ \ \ \ \ \ \ \ \ \ \ \ \ \ \ \ \ \ \ \ \ \ }$};
\draw (578.81,110.41) node [anchor=north west][inner sep=0.75pt]  [rotate=-180.05]  {$\underbrace{\ \ \ \ \ \ \ \ \ \ \ \ \ \ \ \ \ \ \ \ \ \ }$};
\draw (169,160.4) node [anchor=north west][inner sep=0.75pt]    {$\lceil ( 1+2\delta ) L'\rceil $};
\draw (645,160.4) node [anchor=north west][inner sep=0.75pt]    {$\lceil ( 1+2\delta ) L'\rceil $};

\end{tikzpicture}

%% file: finding-cycles.tex
\section{Finding Cycles}
\label{sec:finding-cycles}

Let $\calS=(\calS_1,\ldots,\calS_N)$ be a \valid move sequence of length $N=n\log^{10}n$ that consists of $2$-moves only.
By Lemma \ref{lem:sequence-selection}, $\calS$ has a window $W=(W_1,\ldots,W_L)$ of length $L$ such that the number of moves in $W$ that satisfy  (\ref{heheeq}) is at least $\Omega(L/\log n)$.
We show in this section that such a $W$ satisfies
\begin{equation}\label{hiddenconst2}
{\rank}_{{\cycles}}(W) = \Omega\left(\frac{L}{\log^{10} n}\right).
\end{equation}
This will finish the proof of Lemma \ref{maintech} when $\calS$ consists of $2$-moves only.

\def\Vl{V_{\textbf{L}}}
\def\Vh{V_{\textbf{H}}}

To this end, let $\tau_0\in \{\pm 1\}^{V(W)}$ be the configuration with $\tau_0(u)=-1$ for all $u\in V(W)$ 
  so that we can work on vectors $\imp_{\tau_0,W}(i)$ 
  and $\imp_{\tau_0,W}(C)$ for dependent cycles of $W$
(at the same time, recall from Lemma~\ref{lem:independence-init-rank-cycles} that $\rank_{\cycles}(W)$ does not depend on the choice of $\tau_0$).
Let $\tau_0,\ldots,\tau_L$ denote the sequence of configurations induced by $W$.

Next, let us construct an auxiliary graph $H=(V(W),E)$,
  where every move $W_i=\{u,v\}$ adds an edge between $u$ and $v$ in $E$.
Note that we allow parallel edges in $H$ so $|E|=L$
  and $\#_W(u)$ is exactly the degree of $u$ in $H$.
There is also a natural one-to-one correspondence between cycles of $W$ and  cycles of $H$.
The following lemma shows the existence of a nice looking bipartite graph in $H$:

\begin{lemma}\label{hehelemma51}
There are two disjoint sets of nodes $V_1,V_2\subset V(W)$ and a subset of edges $E'\subseteq E$ such that 
\begin{enumerate}
\item Every edge in $E'$ has one node in $V_1$ and the other node in $V_2$;
\item $|V_1\cup V_2|=O(L/\log^3 n)$ and $|E'|=\Omega(L/\log n)$;
\item $\#_W(u)\le \log^7 n$ for every node $u\in V_1$.
\end{enumerate}
\end{lemma}
 \begin{proof}
 Let $V$ be the set of vertices $v$ such that $\#_W(v)\ge \log^3n$. 
We start our proof with the size of $V$: $$|V|\log^3 n\leq |V|\min_{v\in V}\#_W(u)= |V|\min_{v\in V} \mathop{\operatorname{deg}_H}(v)\leq \sum_{v\in V} \mathop{\operatorname{deg}_H}(v)=2|E(H)|\leq 2L.$$
We further partition $V$ into $V_\ell$ and $V_h$ such that $V_\ell$ contains those in $V$
with $\#_W(v)\le \log^7 n$  and $V_h$ contains those with $\#_W(v)>\log^7 n$. 
By Lemma~\ref{lem:sequence-selection} we can assume the number of edges incident to at least one vertex in $V_\ell$ (that is edges in $V_\ell\times V_\ell\cup V_\ell\times V_h$) is at least $\Omega(L/\log n)$. 
Suppose we construct $V_1$ and $V_2$ by randomly put each node in $V_\ell$ in $V_1$ or $V_2$ and put all nodes in $V_h$ in $V_2$. Any edge in $V_\ell\times V_\ell$ or $V_\ell\times V_h$ is between $V_1$ and $V_2$ with $1/2$ probability. Thus $\mathbb{E}[\text{EdgesInCut}(V_1,V_2)]=|E|/2=\Omega(L/\log n)$. Thus, by standard probabilistic argument,
there exist at least one assignment of $V_\ell$ to $V_1$ and $V_2$ such that at least half of the edges in $V_\ell\times V_\ell\cup V_\ell\times V_h$ are included. Hence, we get a bipartite graph between $V_1$ and $V_2$ with at least $\Omega(L/\log n)$ edges, and any node $v$ in $V_1$ satisfies $\log^3 n\le \#_W(v)\le \log^7 n$. Notice that since $V=|V_1\cup V_2|=|V_\ell\cup V_h|$, we get that $|V_1\cup V_2|=O(L/\log^3 n)$.
 \end{proof}

Recall the definition of dependent cycles of $W$ (and their cancellation vectors) from Section \ref{sec:setup}.
Since we only care about the rank of vectors induced by dependent cycles, 
  we give the following definition which classify each edge of $H$ into two types and then use it to give a sufficient condition for a cycle of $W$ to be dependent:

  


\begin{definition}
We say the $i$-th move 
$W_i=\{u,v\}$ of $W$ is of the same sign if $\tau_i(u)=\tau_i(v)$, and is of different signs if $\tau_i(u)\ne \tau_i(v)$.
\end{definition}


\begin{lemma}
Let $C=(c_1,\ldots,c_t)$ be a cycle of $W$ and assume that $t$ is even. If all of $W_{c_1},\ldots,W_{c_t}$ are of different signs, then $C$ is a dependent cycle;
If all of $W_{c_1},\ldots,W_{c_t}$ are of the same sign, then $C$ is a dependent cycle of $W$.
\end{lemma}
\begin{proof}
Recall the Dependence Criterion (Remark~\ref{remark:criterion})
\[\text{
$C$ is a dependent cycle of $W$ } \Leftrightarrow
(-1)^{t}=\tau_{c_t}(u_1)\tau_{c_1}(u_1)\cdot {\prod_{i=2}^{t}\tau_{c_{i-1}}(u_i)}{\tau_{c_{i}}(u_i)}
\]


If all of $W_{c_1},\ldots,W_{c_t}$ are of different signs, then $\tau_{c_j}(u_j)\tau_{c_j}(u_{j+1})=\tau_{c_t}(u_1)\tau_{c_t}(u_t)=-1$, and the above expression equals to $(-1)^{t}=(-1)^{t}$. If all of $W_{c_1},\ldots,W_{c_t}$ are of the same sign, then $\tau_{c_j}(u_j)\tau_{c_j}(u_{j+1})=\tau_{c_t}(u_1)\tau_{c_t}(u_t)=1$, the above expression is also $1=(-1)^t$, which holds since $t$ is even.
\end{proof}

We assume in the rest of the proof that at least half of edges in $E'$ are of the same sign; the case when at least half of $E'$ are of different signs can be handled similarly. 
Let $E''$ be the subset of $E'$ that consists of edges of the same sign, with $|E''|\ge |E'|/2$. In the following discussion, cycles in $E''$ always refer to cycles that do not use the same edge twice (parallel edges are counted as different edges, since they correspond to different moves in the window $W$).

The aforementioned discussion leads to the following corollary which reduces the existence of dependent cycle of $W$ to a simple cycle in auxiliary graph $H$:
\begin{corollary}\label{cor:simple-to-dependent}
Since every cycle in a bipartite graph has even length,
  every cycle in $E''$ corresponds to a dependent cycle of $W$.
For convenience, given any cycle $C$ of $E''$ we will write $\imp_{\tau_0,W}(C)$ to denote the vector of its corresponding dependent cycle of $W$.
\end{corollary}

We first deal with the case when   $E''$ contains many parallel edges:

\begin{lemma}\label{lemma54}
Let $D$ be the subset of nodes in $V_1$ that have parallel edges in $E''$.
Then $\rank_{\cycles}(W)\ge |D|/2.$
\end{lemma}
\begin{proof}
We prove the lemma even if the sequence contains both 1-moves and 2-moves so that we can use it also in the general case in the next section. We note first that if $S_i=\{u,v\}$, $S_j=\{u,v\}, i<j$ are two moves that involve the same two nodes, then there is at least one node $z \neq u,v$ that appears an odd number of times between the two moves. This follows from the definition of a \valid move sequence. 

We will construct a set $Q$ of at least $|D|/2$ 2-cycles, where each 2-cycle consists of two parallel edges in $E''$. We use the following procedure.\\
1. While there is a 2-cycle $(u,v)$ with $u \in D$, $v \in V_2$, such that some node $z \neq u$ of $V_1$ moves an odd number of times between the two $\{u,v\}$ moves of the 2-cycle, pick any such 2-cycle $(u,v)$ and add it to our set $Q$, pick any such node $z \neq u$ that moves an odd number of times between the two  $\{u,v\}$ moves, and delete $u$ and $z$ from $D$ (if $z$ is in $D$). \\
2. Suppose now that there are no more 2-cycles as in step 1. 
While $D$ is not empty, let $u$ be any remaining node in $D$, take any two incident parallel edges $\{u,v\}$ in $E''$, add the corresponding 2-cycle to $Q$, and delete $u$ from $D$.

Firstly, notice that for every new entry at $Q$ in the procedure, we delete at most 2 nodes from $D$. Hence,
this procedure will generate clearly a set $Q$ of at least $|D|/2$ 2-cycles. 
Let $(u_1,v_1), (u_2,v_2), \ldots, (u_k,v_k)$ be the sequence of 2-cycles selected, where the first $d$ were selected in step 1, and the rest in step 2. The nodes $u_i$ are distinct, while the nodes $v_i$ may not be distinct. 
For each $i=1,\ldots, d$, let $z_i$ be the node in $V_1$ that appears an odd number of times between the two $\{u_i,v_i\}$ moves that was selected by the algorithm. Note that node $z_i \neq u_j$ for all $j \geq i$, since $z_i$ was deleted from $D$ when $u_j$ was selected. 
For each $i = d+1, \ldots, k$, let $z_i$ be any node, other than $u_i, v_i$, that appears an odd number of times between the two $\{u_i,v_i\}$ moves. Then $z_i$ is not in $V_1$ because in step 2 there are no odd nodes in $V_1$.
For each $i =1,\ldots,k$, we view the edge $\{u_i,z_i\}$ of the complete graph as a witness for the 2-cycle $(u_i,v_i)$.

Consider the matrix with columns corresponding to the selected 2-cycles $(u_i,v_i)$, $i=1,\ldots,k$, and rows corresponding to the witness edges $\{u_i,z_i\}$. The entry for the corresponding witness edge $\{u_i,z_i\}$ is nonzero. Indeed, by definition \ref{def:dependent-cycle}, $\imp_{\tau_0,W}(C=(u_i,v_i))_{\{u_i,z_i\}}=-b_1(\tau_{c_1}(u_i)\tau_{c_1}(z_i))-b_2(\tau_{c_2}(u_i)\tau_{c_2}(z_i))$ 
\begin{equation*} \text{ and }
    \begin{cases}
    \tau_{c_1}(z_i)=-\tau_{c_2}(z_i)\\
b_1=1\ \& \ b_2=-\tau_{c_1}(v_i)\tau_{c_2}(v_i)
\\ \tau_{c_m}(v_i)=\tau_{c_m}(u_i), \text{ for } m\in\{1,2\}  \
\end{cases}\text{ which yields } 
\imp_{\tau_0,W}(C=(u_i,v_i))_{\{u_i,z_i\}}=2\cdot\tau_{c_1}(z)\neq 0.
\end{equation*}

Consider the column for a 2-cycle $(u_i,v_i)$ selected in step 1. The entry for any other witness edge $\{u_j,z_j\}$ with $j<i$ is 0 because $u_j, z_j \neq u_i, v_i$. (The entries for witness edges $\{u_j,z_j\}$ with $j>i$ could be nonzero.)

 Consider the column for a 2-cycle $(u_i,v_i)$ selected in step 2. 
The entry for any witness edge $\{u_j,z_j\}$ from step 1 (i.e. with $j \leq d$) is 0 because $u_j, z_j \neq u_i, v_i$. The entry for any witness edge $\{u_j,z_j\}$ from step 2 (i.e. with $j > d$) is also 0 because (1) $u_j \neq u_i, v_i$, (2) $z_j \notin V_1$ hence $z_j \neq u_i$, and (3), even if $z_j = v_i$, all nodes of $V_1$ \textendash hence also $u_j$ \textendash occur an even number of times between the two $\{u_i,v_i\}$ moves, therefore the entry for $\{u_j,v_i\}$ is 0.
\[\footnotesize
\begin{cases}
\calM^{\text{[step 1]}}_{1\to d}=\begin{bmatrix}
\imp_{\tau_{0} ,W}( C_{1})_{\{x_{1} ,y_{1}\}} \neq 0 & 0 & 0\\
* & \ddots  & \vdots \\
* & * & \imp_{\tau_{0} ,W}( C_{d})_{\{x_{d} ,y_{d}\}} \neq 0
\end{bmatrix}\\\\
\footnotesize
\calM^{\text{[step 2]}}_{d+1\to k}=\operatorname{\mathop{diag}_{t\in (d+1)\to k}}(\footnotesize 
\imp_{\tau_{0} ,W}( C_{t})_{\{x_{t} ,y_{t}\}} \neq 0)
\end{cases}
\Rightarrow
\footnotesize 
\calM=\begin{bmatrix}
\calM^{\text{[step 1]}}_{1\to d} & \mbox{\bf 0}\\
*  & \calM^{\text{[step 2]}}_{d+1\to k}
\end{bmatrix}
\]
Thus, the matrix with columns corresponding to the selected 2-cycles $(u_i,v_i)$ and rows corresponding to their witness edges $\{u_i,z_i\}$ is a lower triangular matrix with non-zero diagonal entries.
It follows that the columns are linearly independent.
\end{proof}

As a result, it suffices to deal with the case when $|D|$ is $o(L/\log^{8} n)$.
Let $E^*$ denote the subset of edges obtained from $E''$ after deleting all nodes of $D$ and their incident edges. The remaining bipartite graph has no parallel edges.
Then we have 
$$
|E^*|\ge |E''|-|D|\cdot \log^7n
  =\Omega(L/\log n).
$$
We list all properties of the bipartite graph $H^*=(V_1\cup V_2,E^*)$ we need as follows:
\begin{enumerate}
\item $H^*$ is a bipartite graph with no parallel edges;
\item $|V_1\cup V_2|\le O(L/\log^3 n)$ and $|E^*|\ge \Omega(L/\log n)$; and
\item $\#_W(u)\le \log^7 n $ for every node $u\in V_1$.
\item Every edge $e=\{u,v\}\in E^*$ corresponds to a move $W_i=\{u,v\}$ which is of the same sign.
\end{enumerate}

Recall that $E(W)$ denotes the set of edges in $K_n$ which have both nodes in $V(W)$.
These edges are indices of $\imp_{\tau_0,\calS}(i)$ and $\imp_{\tau_0,\calS}(C)$ for a given dependent cycle $C$ of $W$. Our main lemma is the following:
\begin{lemma}\label{main}
Fix an arbitrary $s\in [0:L/\log^{10}(n)]$. Assume additionally that 
there exists a set of edges $\mathcal{E}_{s}=\{\{x_1,y_1\},\cdots \{x_s,y_s\}\}\subseteq E(W)$ 
such that $x_i\in V_1$ for all $i\in [s]$.
Then there exists a cycle $C$ in $H^*$ and an edge $\{u,v\}\in E(W)$ with $u\in V_1\setminus \{x_1,y_1,\ldots,x_s,y_s\}$ such that 
\begin{equation}\label{eq:main}
\Big(\emph{\imp}_{\tau_0,W}(C)\Big)_{\{u,v\}}\ne 0\quad\text{and}\quad
\Big(\emph{\imp}_{\tau_0,W}(C)\Big)_{\{x_i,y_i\}}=0,\quad\text{for all $i\in [s]$.}
\end{equation}
\end{lemma}

\begin{proof}[Proof of (\ref{hiddenconst2}) Assuming Lemma \ref{main}]

Start with $\mathcal{E}_{s=0}=\emptyset$, For integer $s$ going from $0$ to $\lfloor  L/\log^{10} n\rfloor $, using Lemma \ref{main}, find  cycle $C_{s+1}$ and an edge $\{u,v\}$ satisfying \eqref{eq:main}, 
let $\mathcal{E}_{s+1}=\mathcal{E}_{s}\cup\{x_{s+1},y_{s+1}\}=\{u,v\}$ and repeat the above process.

In the end, we get a set of cycles $C_1,\cdots, C_k$ where {$k= \lfloor  L/\log^{10} n\rfloor $}. And for any $j\in [k]$, we have
\[
\Big(\emph{\imp}_{\tau_0,W}(C_j)\Big)_{\{x_j,y_j\}}\ne 0\quad\text{and}\quad
\Big(\emph{\imp}_{\tau_0,W}(C_j)\Big)_{\{x_i,y_i\}}=0,\quad\text{for all $i\in [j-1]$.}
\]
Let $\calM$ be the $k\times k$ square matrix where $\calM_{ij}=\Big(\emph{\imp}_{\tau_0,W}(C_j)\Big)_{\{x_i,y_i\}}$. \[\footnotesize\calM=
\begin{bmatrix}
\imp_{\tau_{0} ,W}( C_{1})_{\{x_{1} ,y_{1}\}} \neq 0 & 0 & 0 & \cdots  & 0\\
* & \imp_{\tau_{0} ,W}( C_{2})_{\{x_{2} ,y_{2}\}} \neq 0 & 0 & \cdots  & 0\\
\vdots  & * & \imp_{\tau_{0} ,W}( C_{3})_{\{x_{3} ,y_{3}\}} \neq 0 & \cdots  & 0\\
\vdots  & \vdots  & * & \ddots  & 0\\
* & * & * & * & \imp_{\tau_{0} ,W}( C_{k})_{\{x_{k} ,y_{k}\}} \neq 0
\end{bmatrix}\]
As we can see, the matrix is lower triangular with non-zero diagonal entries, so it has full rank $k$. Note that $\calM$ is a submatrix of the matrix formed by taking $\emph{\imp}_{\tau_0,W}(C_j)$ as column vectors, therefore we have $\rank_{\cycles}(W)\ge k\ge  L/\log^{10} n.$
\end{proof}


\subsection{Proof of Lemma~\ref{main}}
Given a cycle $C$ in $H^*$, we say $\{u,v\}\in E(W)$ is a \emph{witness} of $C$ if 
$$
\left(\imp_{\tau_0,W}(C)\right)_{\{u,v\}}\ne 0.
$$
So the goal of Lemma \ref{main} is to find a  cycle $C$ of $H^*$ such that none of $(u_i,v_i)\in \mathcal{E}_s$ are witnesses of $C$ and at the same time, $C$ has a witness edge $\{u,v\}$ with $u$ being a new node in $V_1$ not seen in $\mathcal{E}_s$ before.
The proof consists of two steps. First we introduce a so-called  \emph{split auxiliary graph} $G$ using $H^*$ and $\mathcal{E}_s$, by deleting certain nodes and creating 
extra copies of certain nodes in $H^*$.
We show in Lemma \ref{mainlemma1} that certain simple cycles in $G$ correspond to cycles in $H^*$ that don't have any edge in $\mathcal{E}_s$ as witnesses. 
Next we show in Lemma \ref{lemma58} how to find such a simple cycle in $G$ that has a new witness $(u,v)$ such that $u\in V_1$ and does not appear in $\mathcal{E}_s$.


Let $\wit_1(\mathcal{E}_s)$ be the set of $u\in V_1$ that appear in $\mathcal{E}_s$ and let $\wit_2(\mathcal{E}_s)$ be the set of $v\in V_2$ that appear in $\mathcal{E}_s$. For each $v\in \wit_2(\mathcal{E}_s)$, we write $\wit_1(v)\ne \emptyset$ to denote the set of nodes $u\in \wit_1(\mathcal{E}_s)$ such that $(u,v)\in \mathcal{E}_s$, and let $k_v$ denote the number of moves in $W$ that involve at least one node in $\wit_1(v)$.
We have $k_v\le |\wit_1(v)|\cdot \log^7 n$ since $\#_W(u)\le \log^7 n$ for all $u\in V_1$.
Below, we give an example of such an auxiliary graph:
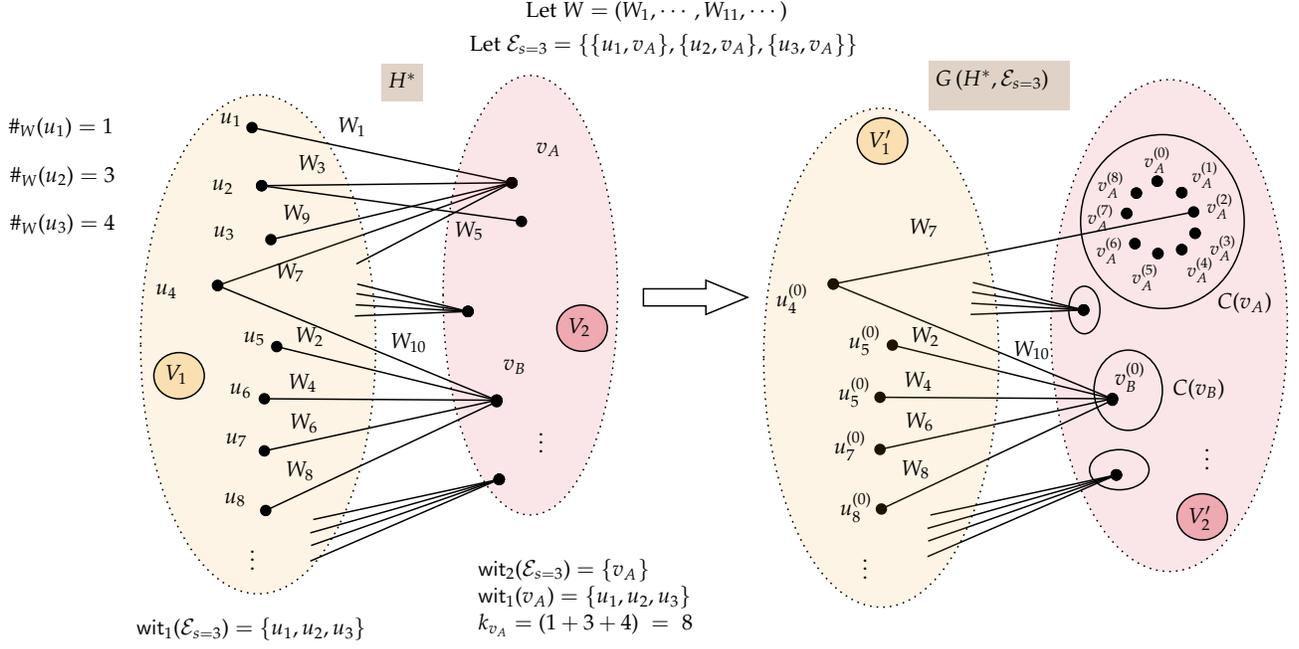
\begin{figure}
    \centering
    \scalebox{0.75}{\input{initial-auxiliary-tikz}}
    \caption{An exemplifying case of an auxiliary graph $H^*$ and splitting graph $G(H^*,\mathcal{E}_s)$}
    \label{fig:HandG}
\end{figure}

We now define our split auxiliary (bipartite) graph $G$. We start with its set of nodes $V_1'\cup V_2'$:
\begin{enumerate}
\item $V_1'=V_1\setminus \wit_1(\mathcal{E}_s)$; and 
\item 
$
V_2'=
\cup_{v\in V_2}C(v),
$
where $C(v)=\{v^{(0)}\}$ if $v\notin \wit_2(\mathcal{E}_s)$ and $C(v)=\{v^{(0)},v^{(1)}, \ldots,v^{(k_v)}\}$ if  $v\in \wit_2(\mathcal{E}_s)$.
\end{enumerate}
So we deleted nodes $\wit_1(\mathcal{E}_s)$ from $V_1$ and replaced each node $v\in \wit_2(\mathcal{E}_s)$ by $k_v+1$ new nodes. 
Next we define the edge set $E(G)$ of $G$. 
Every move  $W_i=\{u,v\}$ in $W$ that corresponds to an edge $(u,v)$ in $H^*$ with $u\in V_1\setminus \wit_1(\mathcal{E}_s)$ and $v\in V_2$ will add an edge in $G$ as follows:
\begin{flushleft}\begin{enumerate}
\item If $v\notin \wit_2(\mathcal{E}_s)$, then we add $(u,v^{(0)})$ to $G$; and\item Otherwise $(v\in \wit_2(\mathcal{E}_s)$), letting $\mu_i\in [0:k_v]$ be the number of moves before $W_i$ that contain at least one node in $\wit_1(v)$ (note that $W_i$ does not contain $\wit_1(v)$; actually $W_i$ cannot contain $\wit_1(\mathcal{E}_s)$),
we add $(u,v^{(\mu_i)})$ to $G$.
\end{enumerate}\end{flushleft}
Therefore, 
every edge in $G$ corresponds to a move in $W$
  which corresponds to an edge in $H^*$ that does not contain a node in $\wit_1(\mathcal{E}_s)$.
It is clear that each simple cycle of $G$ corresponds to a cycle of $H$, which in turn corresponds to a dependent cycle of $W$ (Since we assume w.l.o.g that all edges of auxiliary graph, and its split one, correspond to moves of the same sign  ( See Corollary~\ref{cor:simple-to-dependent}) ).
So $\imp_{\tau_0,W}(C)$ is well defined for  
 simple cycles $C$ of $G$.
Our motivation for constructing and working on $G$ is because of the following lemma:
\begin{lemma}\label{mainlemma1}
Let $C$ be a simple cycle of $G$.
Then none of the edges in $\mathcal{E}_s$ is a witness of $C$.
\end{lemma}

\begin{proof}
Let $e=\{u,v\}\in \mathcal{E}_s$ and $u\in \wit_1(\mathcal{E}_s)$. By the definition of $G$, $u$ has no copy in $G$. So $u$ does not appear on the cycle. 
Let  $C$ be a cycle in $G$ with nodes $(w_1^{(i_1)}, w_2^{(i_2)}, \cdots , w_t^{(i_t)}, w_1^{(i_1)})$ (we use $w_j^{(0)}$ to denote $w_j$ for $w_j\in V_1'$ ). If the cycle does not contain any vertex in $C(v)$, then $\imp_{\tau_0,W}(C)_{e}=0$. 
Now suppose the cycle contains nodes in $C(v)$, specifically, $w_{j_1}=\cdots =w_{j_m}=v$.
Let the corresponding  cycle $C$ on $W$ be $c_1,\cdots, c_t$ where $W_{c_i} = \{w_i,w_{i+1}\}$ if $i<t$ and $W_{c_t}=\{w_t,w_1\}$, and let $b$ be the cancellation vector of $C$. 
We can write down the value of the improvement vector on edge $e$.
\begin{align}\label{eq:imp-splitted-edge}
\begin{split}
    \imp_{\tau_0,W}(C)_{e} &= \sum_{k=1}^m \left(b_{j_k-1}\imp_{\tau_0,W}(c_{j_k-1})_e+b_{j_k} \imp_{\tau_0,W}(c_{j_k})_e \right)\\
    &= -\sum_{k=1}^m\left( b_{j_k-1} \tau_{c_{j_k-1}}(v)\tau_{c_{j_k-1}-1}(u) +  b_{j_k} \tau_{c_{j_k}}(v)\tau_{c_{j_k}-1}(u)\right)
\end{split}
\end{align}
By the construction of $G$, $u$ doesn't appear in any move between $c_{j_k}$ and $c_{j_k-1}$ (otherwise, in $G$, the edge corresponding to $W_{c_{j_k}}$ and edge corresponding to $W_{c_{j_k}-1}$ wouldn't be connected to  $w_{j_k}^{(i_{j_k})}$ with the same $i_{j_k}$.(For an illustative explanation, see Figure~\ref{fig:non-cycle} \& \ref{fig:cycle}) ) So $\tau_{c_{j_k-1}-1}(u)=\tau_{c_{j_k}-1}(u)$. By definition of cancellation vector, 
\[b_{j_k-1} \tau_{c_{j_k-1}}(v) +  b_{j_k} \tau_{c_{j_k}}(v)=0.\]
So each term in \eqref{eq:imp-splitted-edge} is 0, and $\imp_{\tau_0,W}(C)_{e}=0$.

\begin{minipage}{0.45\linewidth}
\centering
\scalebox{0.75}{\input{non-cycle-split-auxiliary-tikz}}

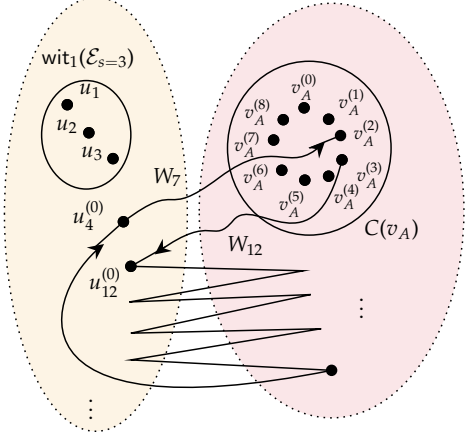
\captionof{figure}{Following the example of Figure~\ref{fig:HandG}, $u_4^{(0)}$ and $u_{11}^{(0)}$ can not close cycle, since $u_{3}\in W_9$ and $u_3\in\mathsf{wit}_1(v_A)$ and appears between $W_7$ and $W_{12}$. Thus, $u_4^{0}$ and $u_{11}^{0}$ fail to make a cycle due to the intervention of a node in $u_3\in\mathsf{wit}_1(v_A)$.}
                  \label{fig:non-cycle}
\end{minipage}\quad
\begin{minipage}{0.45\linewidth}
\centering
\scalebox{0.75}{\input{cycle-split-auxiliary-tikz}}

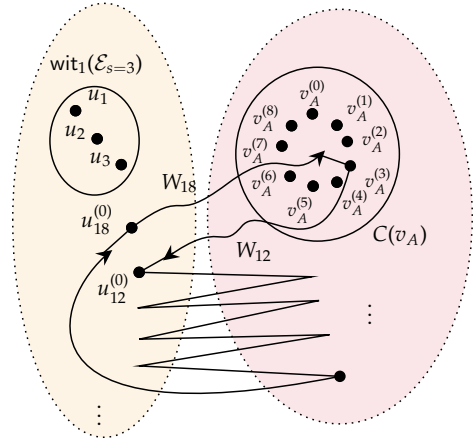
\captionof{figure}{Again by the example of Figure~\ref{fig:HandG}, since $u_{12}^{(0)}$ and $u_{18}^{(0)}$ achieved to form a cycle with moves $W_{12}=\{u_{12},v_A\}$ and $W_{18}=\{u_{18},v_A\}$, we know from our construction that none of $\mathsf{wit}_1(v_A)$ appear in the sub-window $W_{13}\cdots W_{17}$.}
                  \label{fig:cycle}
\end{minipage}

\end{proof}

To finish the proof, it suffices now to find a simple cycle $C$ of $G$ that has a witness $(u,v)\in E(W)$ with one of its vertices $u\in V_1'$.
We start by checking that all conditions for $H^*$ still  hold for $G$. 
It is clear that $G$ is a bipartite graph with no parallel edges.
By the definition of $\wit_1(v)$ for each $v\in \wit_2(\mathcal{E}_s)$, we have  $\sum_{v\in \wit_2(\mathcal{E}_s)}|\wit_1(v)|\le 2s$, also $|\wit_1(\mathcal{E}_s)|\le  2s$.
The number of nodes $|V_1'\cup V_2'|$ in $G$ is at most 
\begin{align*}
O(L/\log^3 n)+\sum_{v\in \wit_2(\mathcal{E}_s)}k_v
\le O(L/\log^3 n)+ \sum_{v\in \wit_2(\mathcal{E}_s)}|\wit_1(v)|\cdot \log^7 n
=O(L/\log^3 n).
\end{align*}
where the last equality used that $s\le L/\log^{10}n$.
The number of edges in $G$ is at least
\[\Omega(L/\log n)- |\wit_1(\mathcal{E}_s)|\cdot \log^7n=\Omega(L/\log n).\]


Let's work on another preprocessing of $G$ to simplify the proof.
Note that the average degree of nodes in $G$ is at least $\Omega\left({\footnotesize(L/\log n)/(L/\log^3 n)}\right)=\Omega(\log^2 n)$.
The following simple lemma shows that one can clean
  up $G$ to get a bipartite graph $G^*$ such that every node has degree at least $100\log n$ and the number of edges in $G^*$ remains to be $\Omega(L/\log n)$:
  
\begin{lemma}\label{lem:minimum-degree}
There is a bipartite graph $G^*=(V_1^*\cup V_2^*,E(G^*))$ with $V_1^* \subseteq V_{1}'$, $V_{2}^*\subseteq V_{2}'$
and $E(G^*)\subseteq E(G)$ such that every node in $G^*$ has degree at least $100\log n$
and $|E(G^*)|=\Omega(L/\log n)$.
\end{lemma}
\begin{proof}
Keep deleting nodes in $G$ with degree less than $100\log n$ one by one (and its adjacent edges) until no such  nodes exist.
The number of edges we delete during the whole process is no more than
$$
(|V_{1}'|+|V_{2}'|)\cdot 100\log n\le O(L/\log^2 n).
$$
So the remaining graph (which trivially has minimum degree at least $100 \log n$)
  has at least $\Omega(L/\log n)$ many edges.
\end{proof}

Let us list the properties of $G^*=(V_1^*\cup V_2^*,E(G^*))$ we will use in the rest of the proof:
\begin{flushleft}
\begin{enumerate}
    \item $V_1^*\subseteq V_1'=V_1\setminus \wit_1(\mathcal{E}_s)$ and $V_2^*\subseteq V_2'$ so each node in $V_2^*$ is in $C(v)$ for some $v\in V_2$.
    
    \item The degree of any node is at least $100\log n$; and
    \item For any $u\in V_1^*$ and $v\in V_2$, the number of neighbors of $u$ in $V_2^*\cap C(v)$ is at most one.
    \item $E(G^*)$ has no parallel edges, $|E(G^*)|\ge \Omega(L/\log n)$ and w.l.o.g.  each edge in $E(G^*)$ correspond to a move of same sign. 
\end{enumerate}
\end{flushleft}
\def\nil{\mathsf{nil}}

We prove the following lemma to finish the proof:

\begin{lemma}\label{lemma58}
Let $u\in V_1^*$ and $v\ne v'\in V_2^* $
such that $(u,v^{(j)}),(u,v'^{(j')})\in E(G^*)$ for some $j$ and $j'$, and the corresponding moves
$W_i=\{u,v\}$ and $W_{i'}=\{u,v'\}$ in $W$ are not consecutive\footnote{
It is worth mentioning, that $V_2^*$ always includes at least two vertices which are copies from different initial nodes $v,v'$. Indeed, if $G^*$ was actually a star graph around $V_2^*=\{v^*\}$, then $\Omega(L/\log n)=E(G^*)=\Theta(V(G^*))=O(L/ \log^3 n)$, which leads to a contradiction. Additionally, notice that $v^{(j)}$ and $v'^{(j')}$ correspond to different nodes in the initial graph, otherwise the initial auxiliary graph $H^*$ would have parallel edges. 
}. Then, the graph $G^*$ has a simple cycle $C$ such that $C$ has a witness $e=\{u,w\}\in E(W)$ with $w\in V_1^*$.
\end{lemma}

\begin{proof}
We begin with a simple sufficient condition for a simple cycle of $G^*$  to satisfy the above condition.

First, let $u\in V_1^*$ and $v\ne v'\in V_2^* $ such that $(u,v^{(j)}),(u,v'^{(j')})\in E(G^*)$ for some $j$ and $j'$, and the corresponding moves
$W_i=\{u,v\}$ and $W_{i'}=\{u,v'\}$ in $W$ are not consecutive. Assume that $i<i'$ without loss of generality; then $i+1< i'$.
The following claim shows that there must be a node $w\notin \{u,v,v'\}$ that moves an odd number times in $W_{i+1},\ldots,W_{i'-1}$:
  
 \begin{claim}\label{claimhehe}
There is a node $w\notin \{u,v,v'\}$ that appears in an odd number of moves in $W_{i+1},\ldots,W_{i'-1}$.
 \end{claim}
 \begin{proof}
This follows from the fact that $W$ is a \valid move sequence.
We distinguish two cases.

If $v'$ appears an even number of times in $W_{i+1},\ldots,W_{i'-1}$, then
use the condition of validity on the
subsequence $W_{i},\ldots,W_{i'-1}$: there is at least one node $w \notin W_i= \{u,v\}$ that appears an odd number of times in 
$W_{i},\ldots,W_{i'-1}$. Since $w \notin W_i$, node $w$ appears an odd number of times in $W_{i+1},\ldots,W_{i'-1}$. Hence $w \neq v'$ and thus
$w\notin \{u,v,v'\}$ and the claim follows.

If $v'$ appears an odd number of times in $W_{i+1},\ldots,W_{i'-1}$, then
$v'$ appears an even number of times in 
$W_{i},\ldots,W_{i'}$.
Use the condition of validity on the
subsequence $W_{i},\ldots,W_{i'}$: there is at least one node $w \notin W_i= \{u,v\}$ that appears an odd number of times in 
$W_{i},\ldots,W_{i'}$. Then $w \neq v'$,
and since also $w \notin W_i= \{u,v\}$,
it follows that $w$ appears an odd number of times in $W_{i+1},\ldots,W_{i'-1}$ and the claim follows again.
%
 \end{proof}

We remark that Claim \ref{claimhehe} holds even when $W$ is a mixture of $1$-moves and $2$-moves. 
This will be important when we deal with the general case in Section \ref{sec:general-case}.

We write $w^*(u,v^{(j)},v'^{(j')})\in V(W)$ to
  denote such a node $w$ promised in the above claim (if more than one exist pick one arbitrarily). 
The next claim gives us a sufficient condition for a simple cycle $C$ of $G^*$ to satisfy the condition of the lemma: 
  
\begin{claim}\label{hehelemma}
Let $$C= u_1v_1^{(j_1) }u_2v_2^{(j_2)}\cdots u_kv_k^{(j_k)}u_1$$ be a simple cycle of $G^*$ for 
some nonnegative integers $j_1,\ldots,j_k$.
Suppose for some $i\in [k]$ we have that
$w:=w^*(u_i,v_{i-1}^{(j_{i-1})},v_i^{(j_i)})\in V(W)$ does not appear in $C$ (where $v_{i-1}^{(j_{i-1})}$ denotes
  $v_k^{(j_k)}$ if $i=1$), i.e., 
  $w\notin \{u_1,\ldots,u_k,v_1,\ldots,v_k\}$, then $(u_i,w)\in E(W)$ must be a witness of $C$. 
\end{claim}
\begin{proof}
Let the corresponding cycle in $W$ be $(c_1,\cdots, c_{2k})$, edge $\{u_l,v_l^{(j_l)}\}$ corresponds to move $c_{2l-1}$ and edge $(v_l^{(j_l)},u_{l+1})$ corresponds to $c_{2l}$ (when $l=k$, $u_{l+1}$ denotes $u_1$). Let $b$ be its cancellation vector.
Recall
\[\Big(\imp_{\tau_0,W}(C)\Big)_{\{w,u_i\}} = \sum_{l=1}^{2k}b_{l}\Big(\imp_{\tau_0,W}(c_l)\Big)_{\{w,u_i\}}.\]
Since $w$ does not appear in $C$,  $\Big(\imp_{\tau_0,W}(c_l)\Big)_{\{w,u_i\}}\neq 0$ only when $u_i\in W_{c_l}$, i.e., when $l=2i-2$ or $l=2i-1$ (if $i=1$, it is $l=2k$ or $l=1$). So 
\[\Big(\imp_{\tau_0,W}(C)\Big)_{\{w,u_i\}} = -b_{2i-2}\tau_{c_{2i-2}}(u_i)\tau_{c_{2i-2}}(w)-b_{2i-1}\tau_{c_{2i-1}}(u_i)\tau_{c_{2i-1}}(w).\]
By the definition of $w^*$, $w$ moved odd number of times between move $c_{2i-2}$ and $c_{2i-1}$. So $\tau_{c_{2i-2}}(w)=-\tau_{c_{2i-1}}(w)$. Also, by property $4$ of $G^*$, Corollary~\ref{cor:simple-to-dependent} and the definition of a dependent cycle, $b_{2i-2}\tau_{c_{2i-2}}(u_i)+ b_{2i-1}\tau_{c_{2i-1}}(u_i)=0$. So 
\[\Big(\imp_{\tau_0,W}(C)\Big)_{\{w,u_i\}} = -2b_{2i-2}\tau_{c_{2i-2}}(u_i)\tau_{c_{2i-2}}(w)\ne0.\]
This finishes the proof of the claim.
\end{proof}


Finally we prove the existence of a simple cycle $C$ of $G^*$ that satisfies the condition of the above claim.
To this end, we first review a simple argument which shows that any bipartite graph
  with $n$ nodes and minimum degree at least $100\log n$ must have a simple cycle. 
Later we modify it to our needs.

The argument goes by picking an arbitrary node in the graph as the root and growing a binary tree 
  of $ \log n$ levels as follows:
\begin{flushleft}
\begin{enumerate}
    \item In the first round we just add two distinct neighbors of the root in the graph as its children.
\item Then for each round, we grow the tree by one level by going through its current leaves one by one to add two children for each leaf.
For each leaf $u$ of the current tree we just pick two of its neighbors in the graph that do not appear in ancestors of $u$ and add them as children of $u$. 
Such neighbors always exist since the tree will have no more than $ \log n$ levels
  and each node has degree at least $100 \log n$ in the graph.
\end{enumerate}
\end{flushleft}
Given that there are only $n$ nodes in the graph, there must be a node that appears more than once in the tree at the end.
Let's consider the first moment when we grow a leaf by adding one child (labelled by $u$) and $u$ already appeared in the tree.
Note that the two nodes labelled by $u$ are not related in the tree since we maintain the invariant that the label of a node does not appear in its ancestors. 
Combining paths from these two nodes to the first node at which the two paths diverge, 
  we get a simple cycle of the graph.

%
  
We now adapt the above argument to prove the existence of a simple cycle $C$ of $G^*$ that satisfies the condition of Claim \ref{hehelemma} by building a binary tree of   $2\log n$ levels as follows. 
We start with an arbitrary node $u_{\textsf{root}}\in V_1^*$ as the root of the tree and expand the tree level by level, leaf by leaf, as follows:
\begin{flushleft}\begin{enumerate}
\item Case $1$: 
The leaf we would like to grow is
  labelled a node $u\in V_1^*$.
In this case we add two children as follows.
Let $u_1 v_1^{(j_1)}\cdots u_{k-1} v_{k-1}^{(j_{k-1})}u$ be the path from the root ($u_1$) to $u$ in the tree,
  where $u_1,\ldots,u_{k-1},u\in V_1^*$ and 
  $v_1^{(j_1)},\ldots,v_{k-1}^{(j_{k-1})}\in V_2^*$.
We pick two neighbors $v^{(j)},v'^{(j')}$ of $u$ in $G^*$ with distinct $v,v'\in V_2$ as its children in the tree. 
We would like $v$ and $v'$ to satisfy the following two properties:
(1) $v$ and $v'$ do not lie in $\{v_1,v_2,\cdots,v_{k-1}\}$, (2) 
$v$ and $v'$ are different from 
  $w^*(u_i,v_i^{(j_i)},v_i'^{(j_i')})$ for every $i=1,\ldots,k-1$, where $v_i'^{(j_i')}$ denotes the other child of $v_i$ in the tree and (3) the move corresponding to $\{u,v^{(j)}\}$ in $W$ and the move corresponding to $\{u,v'^{(j')}\}$ in $W$ are not consecutive moves.
The existence of $v^{(j)}$ and $v'^{(j')}$ that satisfy (1), (2) and (3) follows trivially from the fact that every node (in particular, $u$ here) has
  degree at least $100\log n$ in $G^*$. Indeed, to satisfy (2) and (3), for each time we may reject at most 2 possible leafs.
Given that the tree will only grow for $2\log n$ levels \textendash the half of times with $V_1^*$ leafs and the rest half with $V_2^*$  \textendash, we have $k\le \log n$ and there are at most $2k\le 2\log n$ edges of $u$ that need to be avoided. 
Moreover, no two edges from $u$ go two the 
  same $C(v)$ for some $v\in V_2$ (Because we don't allow parallel edges).

\item Case $2$: The leaf we would like to grow is labelled a node $v^{(j)}\in V_2^*$. In this case we just add one neighbor $u\in V_1^*$ of $v^{(j)}$ as its only child.
Let $u_1 v_1^{(j_1)}\cdots v_{k-1}^{(j_{k-1})}u_{k} v^{(j)}$ be the path from the root to  $v^{(j)}$.
We pick a neighbor $u\in V_1^*$ of $v^{(j)}$ in $G^*$ that satisfies
(1) $u\notin \{u_1,\cdots, u_k\}$ and (2) 
$u$ is different from
  $w^*(u_i,v_i^{(j_i)},v_i'^{(j_i')})$ for every $i=1,\ldots,k-1$ and $u$ is different from $w^*(u_k,v^{(j)},v'^{(j')})$, where $v_i'^{(j_i')}$ denotes the other child of $u_i$ and $v_i'^{(j_i')}$ denotes the other child of $u_k$ in the tree. The existence of such  $u$ follows from the same argument as Case 1.
\end{enumerate}\end{flushleft}

Given that the tree has $2\log n$ levels and there are only $n$ nodes, 
there must be a node that appears more than once in the tree at the end, and let's consider the first moment when we grow a leaf by adding a child and the same node already appeared in the tree.
Similarly we trace the two paths and let $u\in V_1^*$ be the node where the 
  two paths diverge; note that given the construction of the tree, this node must be a node in $V_{1}^*$,
  given that nodes in $V_{2}^*$ only have one child in the tree.
On the one hand, the way we construct the tree makes sure that combining the two paths leads to a simple cycle $C$ of $G^*.$
On the other hand,
  let $v^{(j)},v'^{(j')}\in V_2^*$ be the two children of $u$ (which are next to $u$ on the cycle).
Then it is easy to verify that $w^*(u,v^{(j)},v'^{(j')})$ does not appear on the cycle we just found.
  
This ends the proof of the lemma. \end{proof}
\begin{figure}[h!]
    \centering
    \scalebox{0.95}{\input{tree-tikz}}
    \caption{Example of Finding-Cycle Process.}
    \label{fig:HandG}
\end{figure}
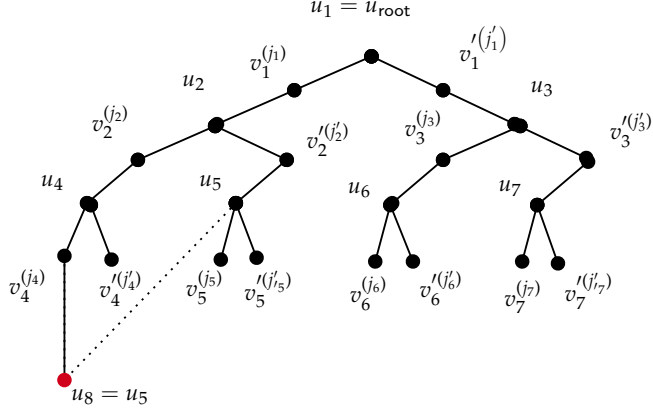

%% file: initial-auxiliary-tikz
\tikzset{every picture/.style={line width=0.75pt}} 

\begin{tikzpicture}[x=0.75pt,y=0.75pt,yscale=-1,xscale=1]

\draw  [fill={rgb, 255:red, 208; green, 2; blue, 27 }  ,fill opacity=0.1 ][dash pattern={on 0.84pt off 2.51pt}] (727,246.5) .. controls (727,155.1) and (762.59,81) .. (806.5,81) .. controls (850.41,81) and (886,155.1) .. (886,246.5) .. controls (886,337.9) and (850.41,412) .. (806.5,412) .. controls (762.59,412) and (727,337.9) .. (727,246.5) -- cycle ;
\draw    (166.67,219.53) -- (364.5,150.34) ;
\draw [shift={(364.5,150.34)}, rotate = 340.72] [color={rgb, 255:red, 0; green, 0; blue, 0 }  ][fill={rgb, 255:red, 0; green, 0; blue, 0 }  ][line width=0.75]      (0, 0) circle [x radius= 3.35, y radius= 3.35]   ;
\draw [shift={(166.67,219.53)}, rotate = 340.72] [color={rgb, 255:red, 0; green, 0; blue, 0 }  ][fill={rgb, 255:red, 0; green, 0; blue, 0 }  ][line width=0.75]      (0, 0) circle [x radius= 3.35, y radius= 3.35]   ;
\draw  [fill={rgb, 255:red, 245; green, 166; blue, 35 }  ,fill opacity=0.12 ][dash pattern={on 0.84pt off 2.51pt}] (114.9,259.09) .. controls (114.9,166.36) and (150.29,91.19) .. (193.94,91.19) .. controls (237.59,91.19) and (272.98,166.36) .. (272.98,259.09) .. controls (272.98,351.83) and (237.59,427) .. (193.94,427) .. controls (150.29,427) and (114.9,351.83) .. (114.9,259.09) -- cycle ;
\draw    (202.5,188.64) -- (364.5,150.34) ;
\draw [shift={(364.5,150.34)}, rotate = 346.7] [color={rgb, 255:red, 0; green, 0; blue, 0 }  ][fill={rgb, 255:red, 0; green, 0; blue, 0 }  ][line width=0.75]      (0, 0) circle [x radius= 3.35, y radius= 3.35]   ;
\draw [shift={(202.5,188.64)}, rotate = 346.7] [color={rgb, 255:red, 0; green, 0; blue, 0 }  ][fill={rgb, 255:red, 0; green, 0; blue, 0 }  ][line width=0.75]      (0, 0) circle [x radius= 3.35, y radius= 3.35]   ;
\draw    (196.25,152.35) -- (364.5,150.34) ;
\draw [shift={(364.5,150.34)}, rotate = 359.32] [color={rgb, 255:red, 0; green, 0; blue, 0 }  ][fill={rgb, 255:red, 0; green, 0; blue, 0 }  ][line width=0.75]      (0, 0) circle [x radius= 3.35, y radius= 3.35]   ;
\draw [shift={(196.25,152.35)}, rotate = 359.32] [color={rgb, 255:red, 0; green, 0; blue, 0 }  ][fill={rgb, 255:red, 0; green, 0; blue, 0 }  ][line width=0.75]      (0, 0) circle [x radius= 3.35, y radius= 3.35]   ;
\draw    (189.78,113.24) -- (312.71,139.34) -- (364.5,150.34) ;
\draw [shift={(364.5,150.34)}, rotate = 11.99] [color={rgb, 255:red, 0; green, 0; blue, 0 }  ][fill={rgb, 255:red, 0; green, 0; blue, 0 }  ][line width=0.75]      (0, 0) circle [x radius= 3.35, y radius= 3.35]   ;
\draw [shift={(189.78,113.24)}, rotate = 11.99] [color={rgb, 255:red, 0; green, 0; blue, 0 }  ][fill={rgb, 255:red, 0; green, 0; blue, 0 }  ][line width=0.75]      (0, 0) circle [x radius= 3.35, y radius= 3.35]   ;
\draw  [fill={rgb, 255:red, 208; green, 2; blue, 27 }  ,fill opacity=0.1 ][dash pattern={on 0.84pt off 2.51pt}] (318.28,226.05) .. controls (318.28,144.37) and (344.56,78.15) .. (376.98,78.15) .. controls (409.4,78.15) and (435.68,144.37) .. (435.68,226.05) .. controls (435.68,307.74) and (409.4,373.96) .. (376.98,373.96) .. controls (344.56,373.96) and (318.28,307.74) .. (318.28,226.05) -- cycle ;
\draw    (196.25,152.35) -- (370.97,176.41) ;
\draw [shift={(370.97,176.41)}, rotate = 7.84] [color={rgb, 255:red, 0; green, 0; blue, 0 }  ][fill={rgb, 255:red, 0; green, 0; blue, 0 }  ][line width=0.75]      (0, 0) circle [x radius= 3.35, y radius= 3.35]   ;
\draw [shift={(196.25,152.35)}, rotate = 7.84] [color={rgb, 255:red, 0; green, 0; blue, 0 }  ][fill={rgb, 255:red, 0; green, 0; blue, 0 }  ][line width=0.75]      (0, 0) circle [x radius= 3.35, y radius= 3.35]   ;
\draw    (354.33,296.74) -- (166.67,219.53) ;
\draw [shift={(166.67,219.53)}, rotate = 202.36] [color={rgb, 255:red, 0; green, 0; blue, 0 }  ][fill={rgb, 255:red, 0; green, 0; blue, 0 }  ][line width=0.75]      (0, 0) circle [x radius= 3.35, y radius= 3.35]   ;
\draw [shift={(354.33,296.74)}, rotate = 202.36] [color={rgb, 255:red, 0; green, 0; blue, 0 }  ][fill={rgb, 255:red, 0; green, 0; blue, 0 }  ][line width=0.75]      (0, 0) circle [x radius= 3.35, y radius= 3.35]   ;
\draw    (260.5,218.5) -- (335,237) ;
\draw [shift={(335,237)}, rotate = 13.95] [color={rgb, 255:red, 0; green, 0; blue, 0 }  ][fill={rgb, 255:red, 0; green, 0; blue, 0 }  ][line width=0.75]      (0, 0) circle [x radius= 3.35, y radius= 3.35]   ;
\draw    (206.42,260.64) -- (354.33,296.74) ;
\draw [shift={(354.33,296.74)}, rotate = 13.72] [color={rgb, 255:red, 0; green, 0; blue, 0 }  ][fill={rgb, 255:red, 0; green, 0; blue, 0 }  ][line width=0.75]      (0, 0) circle [x radius= 3.35, y radius= 3.35]   ;
\draw [shift={(206.42,260.64)}, rotate = 13.72] [color={rgb, 255:red, 0; green, 0; blue, 0 }  ][fill={rgb, 255:red, 0; green, 0; blue, 0 }  ][line width=0.75]      (0, 0) circle [x radius= 3.35, y radius= 3.35]   ;
\draw    (199.02,370.94) -- (354.33,297.74) ;
\draw [shift={(354.33,297.74)}, rotate = 334.76] [color={rgb, 255:red, 0; green, 0; blue, 0 }  ][fill={rgb, 255:red, 0; green, 0; blue, 0 }  ][line width=0.75]      (0, 0) circle [x radius= 3.35, y radius= 3.35]   ;
\draw [shift={(199.02,370.94)}, rotate = 334.76] [color={rgb, 255:red, 0; green, 0; blue, 0 }  ][fill={rgb, 255:red, 0; green, 0; blue, 0 }  ][line width=0.75]      (0, 0) circle [x radius= 3.35, y radius= 3.35]   ;
\draw    (198.1,295.74) -- (354.33,296.74) ;
\draw [shift={(354.33,296.74)}, rotate = 0.37] [color={rgb, 255:red, 0; green, 0; blue, 0 }  ][fill={rgb, 255:red, 0; green, 0; blue, 0 }  ][line width=0.75]      (0, 0) circle [x radius= 3.35, y radius= 3.35]   ;
\draw [shift={(198.1,295.74)}, rotate = 0.37] [color={rgb, 255:red, 0; green, 0; blue, 0 }  ][fill={rgb, 255:red, 0; green, 0; blue, 0 }  ][line width=0.75]      (0, 0) circle [x radius= 3.35, y radius= 3.35]   ;
\draw    (198.1,330.83) -- (354.33,296.74) ;
\draw [shift={(354.33,296.74)}, rotate = 347.69] [color={rgb, 255:red, 0; green, 0; blue, 0 }  ][fill={rgb, 255:red, 0; green, 0; blue, 0 }  ][line width=0.75]      (0, 0) circle [x radius= 3.35, y radius= 3.35]   ;
\draw [shift={(198.1,330.83)}, rotate = 347.69] [color={rgb, 255:red, 0; green, 0; blue, 0 }  ][fill={rgb, 255:red, 0; green, 0; blue, 0 }  ][line width=0.75]      (0, 0) circle [x radius= 3.35, y radius= 3.35]   ;
\draw    (260,225) -- (335,237) ;
\draw [shift={(335,237)}, rotate = 9.09] [color={rgb, 255:red, 0; green, 0; blue, 0 }  ][fill={rgb, 255:red, 0; green, 0; blue, 0 }  ][line width=0.75]      (0, 0) circle [x radius= 3.35, y radius= 3.35]   ;
\draw    (259,240) -- (335,237) ;
\draw [shift={(335,237)}, rotate = 357.74] [color={rgb, 255:red, 0; green, 0; blue, 0 }  ][fill={rgb, 255:red, 0; green, 0; blue, 0 }  ][line width=0.75]      (0, 0) circle [x radius= 3.35, y radius= 3.35]   ;
\draw    (259.5,232.5) -- (335,237) ;
\draw [shift={(335,237)}, rotate = 3.41] [color={rgb, 255:red, 0; green, 0; blue, 0 }  ][fill={rgb, 255:red, 0; green, 0; blue, 0 }  ][line width=0.75]      (0, 0) circle [x radius= 3.35, y radius= 3.35]   ;
\draw    (778,171) ;
\draw [shift={(778,171)}, rotate = 0] [color={rgb, 255:red, 0; green, 0; blue, 0 }  ][fill={rgb, 255:red, 0; green, 0; blue, 0 }  ][line width=0.75]      (0, 0) circle [x radius= 3.35, y radius= 3.35]   ;
\draw [shift={(778,171)}, rotate = 0] [color={rgb, 255:red, 0; green, 0; blue, 0 }  ][fill={rgb, 255:red, 0; green, 0; blue, 0 }  ][line width=0.75]      (0, 0) circle [x radius= 3.35, y radius= 3.35]   ;
\draw    (620.42,259.64) -- (768.33,295.74) ;
\draw [shift={(768.33,295.74)}, rotate = 13.72] [color={rgb, 255:red, 0; green, 0; blue, 0 }  ][fill={rgb, 255:red, 0; green, 0; blue, 0 }  ][line width=0.75]      (0, 0) circle [x radius= 3.35, y radius= 3.35]   ;
\draw [shift={(620.42,259.64)}, rotate = 13.72] [color={rgb, 255:red, 0; green, 0; blue, 0 }  ][fill={rgb, 255:red, 0; green, 0; blue, 0 }  ][line width=0.75]      (0, 0) circle [x radius= 3.35, y radius= 3.35]   ;
\draw    (613.02,369.94) -- (768.33,296.74) ;
\draw [shift={(768.33,296.74)}, rotate = 334.76] [color={rgb, 255:red, 0; green, 0; blue, 0 }  ][fill={rgb, 255:red, 0; green, 0; blue, 0 }  ][line width=0.75]      (0, 0) circle [x radius= 3.35, y radius= 3.35]   ;
\draw [shift={(613.02,369.94)}, rotate = 334.76] [color={rgb, 255:red, 0; green, 0; blue, 0 }  ][fill={rgb, 255:red, 0; green, 0; blue, 0 }  ][line width=0.75]      (0, 0) circle [x radius= 3.35, y radius= 3.35]   ;
\draw    (612.1,294.74) -- (768.33,295.74) ;
\draw [shift={(768.33,295.74)}, rotate = 0.37] [color={rgb, 255:red, 0; green, 0; blue, 0 }  ][fill={rgb, 255:red, 0; green, 0; blue, 0 }  ][line width=0.75]      (0, 0) circle [x radius= 3.35, y radius= 3.35]   ;
\draw [shift={(612.1,294.74)}, rotate = 0.37] [color={rgb, 255:red, 0; green, 0; blue, 0 }  ][fill={rgb, 255:red, 0; green, 0; blue, 0 }  ][line width=0.75]      (0, 0) circle [x radius= 3.35, y radius= 3.35]   ;
\draw    (612.1,329.83) -- (768.33,295.74) ;
\draw [shift={(768.33,295.74)}, rotate = 347.69] [color={rgb, 255:red, 0; green, 0; blue, 0 }  ][fill={rgb, 255:red, 0; green, 0; blue, 0 }  ][line width=0.75]      (0, 0) circle [x radius= 3.35, y radius= 3.35]   ;
\draw [shift={(612.1,329.83)}, rotate = 347.69] [color={rgb, 255:red, 0; green, 0; blue, 0 }  ][fill={rgb, 255:red, 0; green, 0; blue, 0 }  ][line width=0.75]      (0, 0) circle [x radius= 3.35, y radius= 3.35]   ;
\draw    (768.33,295.74) -- (580.67,218.53) ;
\draw [shift={(580.67,218.53)}, rotate = 202.36] [color={rgb, 255:red, 0; green, 0; blue, 0 }  ][fill={rgb, 255:red, 0; green, 0; blue, 0 }  ][line width=0.75]      (0, 0) circle [x radius= 3.35, y radius= 3.35]   ;
\draw [shift={(768.33,295.74)}, rotate = 202.36] [color={rgb, 255:red, 0; green, 0; blue, 0 }  ][fill={rgb, 255:red, 0; green, 0; blue, 0 }  ][line width=0.75]      (0, 0) circle [x radius= 3.35, y radius= 3.35]   ;
\draw    (674,224) -- (749,236) ;
\draw [shift={(749,236)}, rotate = 9.09] [color={rgb, 255:red, 0; green, 0; blue, 0 }  ][fill={rgb, 255:red, 0; green, 0; blue, 0 }  ][line width=0.75]      (0, 0) circle [x radius= 3.35, y radius= 3.35]   ;
\draw    (673,239) -- (749,236) ;
\draw [shift={(749,236)}, rotate = 357.74] [color={rgb, 255:red, 0; green, 0; blue, 0 }  ][fill={rgb, 255:red, 0; green, 0; blue, 0 }  ][line width=0.75]      (0, 0) circle [x radius= 3.35, y radius= 3.35]   ;
\draw    (673.5,231.5) -- (749,236) ;
\draw [shift={(749,236)}, rotate = 3.41] [color={rgb, 255:red, 0; green, 0; blue, 0 }  ][fill={rgb, 255:red, 0; green, 0; blue, 0 }  ][line width=0.75]      (0, 0) circle [x radius= 3.35, y radius= 3.35]   ;
\draw    (674.5,217.5) -- (749,236) ;
\draw [shift={(749,236)}, rotate = 13.95] [color={rgb, 255:red, 0; green, 0; blue, 0 }  ][fill={rgb, 255:red, 0; green, 0; blue, 0 }  ][line width=0.75]      (0, 0) circle [x radius= 3.35, y radius= 3.35]   ;
\draw    (229,386) -- (356,350) ;
\draw [shift={(356,350)}, rotate = 344.17] [color={rgb, 255:red, 0; green, 0; blue, 0 }  ][fill={rgb, 255:red, 0; green, 0; blue, 0 }  ][line width=0.75]      (0, 0) circle [x radius= 3.35, y radius= 3.35]   ;
\draw    (229,405) -- (356,350) ;
\draw [shift={(356,350)}, rotate = 336.58] [color={rgb, 255:red, 0; green, 0; blue, 0 }  ][fill={rgb, 255:red, 0; green, 0; blue, 0 }  ][line width=0.75]      (0, 0) circle [x radius= 3.35, y radius= 3.35]   ;
\draw    (229.5,394.5) -- (356,350) ;
\draw [shift={(356,350)}, rotate = 340.62] [color={rgb, 255:red, 0; green, 0; blue, 0 }  ][fill={rgb, 255:red, 0; green, 0; blue, 0 }  ][line width=0.75]      (0, 0) circle [x radius= 3.35, y radius= 3.35]   ;
\draw    (231,377) -- (356,350) ;
\draw [shift={(356,350)}, rotate = 347.81] [color={rgb, 255:red, 0; green, 0; blue, 0 }  ][fill={rgb, 255:red, 0; green, 0; blue, 0 }  ][line width=0.75]      (0, 0) circle [x radius= 3.35, y radius= 3.35]   ;
\draw    (644,383) -- (771,347) ;
\draw [shift={(771,347)}, rotate = 344.17] [color={rgb, 255:red, 0; green, 0; blue, 0 }  ][fill={rgb, 255:red, 0; green, 0; blue, 0 }  ][line width=0.75]      (0, 0) circle [x radius= 3.35, y radius= 3.35]   ;
\draw    (644,402) -- (771,347) ;
\draw [shift={(771,347)}, rotate = 336.58] [color={rgb, 255:red, 0; green, 0; blue, 0 }  ][fill={rgb, 255:red, 0; green, 0; blue, 0 }  ][line width=0.75]      (0, 0) circle [x radius= 3.35, y radius= 3.35]   ;
\draw    (644.5,391.5) -- (771,347) ;
\draw [shift={(771,347)}, rotate = 340.62] [color={rgb, 255:red, 0; green, 0; blue, 0 }  ][fill={rgb, 255:red, 0; green, 0; blue, 0 }  ][line width=0.75]      (0, 0) circle [x radius= 3.35, y radius= 3.35]   ;
\draw    (646,374) -- (771,347) ;
\draw [shift={(771,347)}, rotate = 347.81] [color={rgb, 255:red, 0; green, 0; blue, 0 }  ][fill={rgb, 255:red, 0; green, 0; blue, 0 }  ][line width=0.75]      (0, 0) circle [x radius= 3.35, y radius= 3.35]   ;
\draw   (753,343.5) .. controls (753,336.04) and (761.95,330) .. (773,330) .. controls (784.05,330) and (793,336.04) .. (793,343.5) .. controls (793,350.96) and (784.05,357) .. (773,357) .. controls (761.95,357) and (753,350.96) .. (753,343.5) -- cycle ;
\draw   (756,290) .. controls (756,275.09) and (766.3,263) .. (779,263) .. controls (791.7,263) and (802,275.09) .. (802,290) .. controls (802,304.91) and (791.7,317) .. (779,317) .. controls (766.3,317) and (756,304.91) .. (756,290) -- cycle ;
\draw   (739,236) .. controls (739,227.16) and (743.7,220) .. (749.5,220) .. controls (755.3,220) and (760,227.16) .. (760,236) .. controls (760,244.84) and (755.3,252) .. (749.5,252) .. controls (743.7,252) and (739,244.84) .. (739,236) -- cycle ;
\draw    (815,195.34) ;
\draw [shift={(815,195.34)}, rotate = 0] [color={rgb, 255:red, 0; green, 0; blue, 0 }  ][fill={rgb, 255:red, 0; green, 0; blue, 0 }  ][line width=0.75]      (0, 0) circle [x radius= 3.35, y radius= 3.35]   ;
\draw [shift={(815,195.34)}, rotate = 0] [color={rgb, 255:red, 0; green, 0; blue, 0 }  ][fill={rgb, 255:red, 0; green, 0; blue, 0 }  ][line width=0.75]      (0, 0) circle [x radius= 3.35, y radius= 3.35]   ;
\draw    (784.5,157.34) ;
\draw [shift={(784.5,157.34)}, rotate = 0] [color={rgb, 255:red, 0; green, 0; blue, 0 }  ][fill={rgb, 255:red, 0; green, 0; blue, 0 }  ][line width=0.75]      (0, 0) circle [x radius= 3.35, y radius= 3.35]   ;
\draw [shift={(784.5,157.34)}, rotate = 0] [color={rgb, 255:red, 0; green, 0; blue, 0 }  ][fill={rgb, 255:red, 0; green, 0; blue, 0 }  ][line width=0.75]      (0, 0) circle [x radius= 3.35, y radius= 3.35]   ;
\draw    (824,184.34) ;
\draw [shift={(824,184.34)}, rotate = 0] [color={rgb, 255:red, 0; green, 0; blue, 0 }  ][fill={rgb, 255:red, 0; green, 0; blue, 0 }  ][line width=0.75]      (0, 0) circle [x radius= 3.35, y radius= 3.35]   ;
\draw [shift={(824,184.34)}, rotate = 0] [color={rgb, 255:red, 0; green, 0; blue, 0 }  ][fill={rgb, 255:red, 0; green, 0; blue, 0 }  ][line width=0.75]      (0, 0) circle [x radius= 3.35, y radius= 3.35]   ;
\draw    (798.5,149.34) ;
\draw [shift={(798.5,149.34)}, rotate = 0] [color={rgb, 255:red, 0; green, 0; blue, 0 }  ][fill={rgb, 255:red, 0; green, 0; blue, 0 }  ][line width=0.75]      (0, 0) circle [x radius= 3.35, y radius= 3.35]   ;
\draw [shift={(798.5,149.34)}, rotate = 0] [color={rgb, 255:red, 0; green, 0; blue, 0 }  ][fill={rgb, 255:red, 0; green, 0; blue, 0 }  ][line width=0.75]      (0, 0) circle [x radius= 3.35, y radius= 3.35]   ;
\draw    (815,157) ;
\draw [shift={(815,157)}, rotate = 0] [color={rgb, 255:red, 0; green, 0; blue, 0 }  ][fill={rgb, 255:red, 0; green, 0; blue, 0 }  ][line width=0.75]      (0, 0) circle [x radius= 3.35, y radius= 3.35]   ;
\draw [shift={(815,157)}, rotate = 0] [color={rgb, 255:red, 0; green, 0; blue, 0 }  ][fill={rgb, 255:red, 0; green, 0; blue, 0 }  ][line width=0.75]      (0, 0) circle [x radius= 3.35, y radius= 3.35]   ;
\draw    (783.5,191.34) ;
\draw [shift={(783.5,191.34)}, rotate = 0] [color={rgb, 255:red, 0; green, 0; blue, 0 }  ][fill={rgb, 255:red, 0; green, 0; blue, 0 }  ][line width=0.75]      (0, 0) circle [x radius= 3.35, y radius= 3.35]   ;
\draw [shift={(783.5,191.34)}, rotate = 0] [color={rgb, 255:red, 0; green, 0; blue, 0 }  ][fill={rgb, 255:red, 0; green, 0; blue, 0 }  ][line width=0.75]      (0, 0) circle [x radius= 3.35, y radius= 3.35]   ;
\draw    (823,170) -- (580.67,218.53) ;
\draw [shift={(580.67,218.53)}, rotate = 168.68] [color={rgb, 255:red, 0; green, 0; blue, 0 }  ][fill={rgb, 255:red, 0; green, 0; blue, 0 }  ][line width=0.75]      (0, 0) circle [x radius= 3.35, y radius= 3.35]   ;
\draw [shift={(823,170)}, rotate = 168.68] [color={rgb, 255:red, 0; green, 0; blue, 0 }  ][fill={rgb, 255:red, 0; green, 0; blue, 0 }  ][line width=0.75]      (0, 0) circle [x radius= 3.35, y radius= 3.35]   ;
\draw    (799,198) ;
\draw [shift={(799,198)}, rotate = 0] [color={rgb, 255:red, 0; green, 0; blue, 0 }  ][fill={rgb, 255:red, 0; green, 0; blue, 0 }  ][line width=0.75]      (0, 0) circle [x radius= 3.35, y radius= 3.35]   ;
\draw [shift={(799,198)}, rotate = 0] [color={rgb, 255:red, 0; green, 0; blue, 0 }  ][fill={rgb, 255:red, 0; green, 0; blue, 0 }  ][line width=0.75]      (0, 0) circle [x radius= 3.35, y radius= 3.35]   ;
\draw  [fill={rgb, 255:red, 245; green, 166; blue, 35 }  ,fill opacity=0.12 ][dash pattern={on 0.84pt off 2.51pt}] (533.9,268.09) .. controls (533.9,175.36) and (569.29,100.19) .. (612.94,100.19) .. controls (656.59,100.19) and (691.98,175.36) .. (691.98,268.09) .. controls (691.98,360.83) and (656.59,436) .. (612.94,436) .. controls (569.29,436) and (533.9,360.83) .. (533.9,268.09) -- cycle ;
\draw   (747,176.5) .. controls (747,144.19) and (771.62,118) .. (802,118) .. controls (832.38,118) and (857,144.19) .. (857,176.5) .. controls (857,208.81) and (832.38,235) .. (802,235) .. controls (771.62,235) and (747,208.81) .. (747,176.5) -- cycle ;
\draw   (453,221.75) -- (495,221.75) -- (495,216) -- (523,227.5) -- (495,239) -- (495,233.25) -- (453,233.25) -- cycle ;
\draw    (364.5,150.34) -- (260,205) ;

\draw (182.51,249.03) node [anchor=north west][inner sep=0.75pt]    {$u_{5}$};
\draw (173.34,284.13) node [anchor=north west][inner sep=0.75pt]    {$u_{6}$};
\draw (170.57,317.22) node [anchor=north west][inner sep=0.75pt]    {$u_{7}$};
\draw (169.65,357.33) node [anchor=north west][inner sep=0.75pt]    {$u_{8}$};
\draw (357.31,268.08) node [anchor=north west][inner sep=0.75pt]    {$v_{B}$};
\draw (217.45,246.03) node [anchor=north west][inner sep=0.75pt]    {$W_{2}$};
\draw (212.83,277.11) node [anchor=north west][inner sep=0.75pt]    {$W_{4}$};
\draw (213.76,304.18) node [anchor=north west][inner sep=0.75pt]    {$W_{6}$};
\draw (210.98,336.27) node [anchor=north west][inner sep=0.75pt]    {$W_{8}$};
\draw (372.45,26.43) node [anchor=north west][inner sep=0.75pt]    {$\text{Let } W=( W_{1} ,\cdots ,W_{11} ,\cdots )$};
\draw (334.58,48.5) node [anchor=north west][inner sep=0.75pt]    {$\text{Let }\mathcal{E}_{s=3} =\{\{u_{1} ,v_{A}\} ,\{u_{2} ,v_{A}\} ,\{u_{3} ,v_{A}\}\}$};
\draw (110.75,441.64) node [anchor=north west][inner sep=0.75pt]    {$\mathsf{wit}_{1}(\mathcal{E}_{s=3}) =\{u_{1} ,u_{2} ,u_{3}\}$};
\draw (333.75,400.64) node [anchor=north west][inner sep=0.75pt]    {$ \begin{array}{l}
\mathsf{wit}_{2}(\mathcal{E}_{s=3}) =\{v_{A}\}\\
\mathsf{wit}_{1}( v_{A}) =\{u_{1} ,u_{2} ,u_{3}\}\\
k_{v_{A}} =( 1+3+4) \ =\ 8
\end{array}$};
\draw  [fill={rgb, 255:red, 245; green, 166; blue, 35 }  ,fill opacity=0.27 ]  (140.75, 280.24) circle [x radius= 16.97, y radius= 15.56]   ;
\draw (129.75,271.64) node [anchor=north west][inner sep=0.75pt]    {$V_{1}$};
\draw  [fill={rgb, 255:red, 208; green, 2; blue, 27 }  ,fill opacity=0.27 ]  (411.75, 248.24) circle [x radius= 16.97, y radius= 15.56]   ;
\draw (400.75,239.64) node [anchor=north west][inner sep=0.75pt]    {$V_{2}$};
\draw (589.51,245.03) node [anchor=north west][inner sep=0.75pt]    {$u_{5}^{( 0)}$};
\draw (579.57,316.22) node [anchor=north west][inner sep=0.75pt]    {$u_{7}^{( 0)}$};
\draw (583.65,357.33) node [anchor=north west][inner sep=0.75pt]    {$u_{8}^{( 0)}$};
\draw (837.42,222.68) node [anchor=north west][inner sep=0.75pt]    {$C( v_{A})$};
\draw (768.31,270.08) node [anchor=north west][inner sep=0.75pt]    {$v_{B}^{( 0)}$};
\draw (630.4,172.9) node [anchor=north west][inner sep=0.75pt]    {$W_{7}$};
\draw (631.45,246.03) node [anchor=north west][inner sep=0.75pt]    {$W_{2}$};
\draw (626.83,276.11) node [anchor=north west][inner sep=0.75pt]    {$W_{4}$};
\draw (627.76,303.18) node [anchor=north west][inner sep=0.75pt]    {$W_{6}$};
\draw (624.98,335.27) node [anchor=north west][inner sep=0.75pt]    {$W_{8}$};
\draw  [draw opacity=0][fill={rgb, 255:red, 139; green, 87; blue, 42 }  ,fill opacity=0.27 ][dash pattern={on 4.5pt off 4.5pt}]  (644.75,68.24) -- (739.75,68.24) -- (739.75,101.24) -- (644.75,101.24) -- cycle  ;
\draw (647.75,72.64) node [anchor=north west][inner sep=0.75pt]    {$G\left( H^{*} ,\mathcal{E}_{s=3}\right)$};
\draw (281.45,250.03) node [anchor=north west][inner sep=0.75pt]    {$W_{10}$};
\draw (700.45,255.03) node [anchor=north west][inner sep=0.75pt]    {$W_{10}$};
\draw (582.51,280.03) node [anchor=north west][inner sep=0.75pt]    {$u_{5}^{( 0)}$};
\draw (540.51,217.03) node [anchor=north west][inner sep=0.75pt]    {$u_{4}^{( 0)}$};
\draw (788.31,125.08) node [anchor=north west][inner sep=0.75pt]  [font=\footnotesize]  {$v_{A}^{( 0)}$};
\draw (820.31,135.08) node [anchor=north west][inner sep=0.75pt]  [font=\footnotesize]  {$v_{A}^{( 1)}$};
\draw (830.31,156.08) node [anchor=north west][inner sep=0.75pt]  [font=\footnotesize]  {$v_{A}^{( 2)}$};
\draw (832.31,184.08) node [anchor=north west][inner sep=0.75pt]  [font=\footnotesize]  {$v_{A}^{( 3)}$};
\draw (817,198.74) node [anchor=north west][inner sep=0.75pt]  [font=\footnotesize]  {$v_{A}^{( 4)}$};
\draw (757,141.74) node [anchor=north west][inner sep=0.75pt]  [font=\footnotesize]  {$v_{A}^{( 8)}$};
\draw (750,163.74) node [anchor=north west][inner sep=0.75pt]  [font=\footnotesize]  {$v_{A}^{( 7)}$};
\draw (756,185.74) node [anchor=north west][inner sep=0.75pt]  [font=\footnotesize]  {$v_{A}^{( 6)}$};
\draw (780,203.74) node [anchor=north west][inner sep=0.75pt]  [font=\footnotesize]  {$v_{A}^{( 5)}$};
\draw  [fill={rgb, 255:red, 208; green, 2; blue, 27 }  ,fill opacity=0.27 ]  (828.75, 375.24) circle [x radius= 16.97, y radius= 15.56]   ;
\draw (817.75,366.64) node [anchor=north west][inner sep=0.75pt]    {$V'_{2}$};
\draw  [fill={rgb, 255:red, 245; green, 166; blue, 35 }  ,fill opacity=0.27 ]  (613.75, 122.24) circle [x radius= 16.97, y radius= 15.56]   ;
\draw (602.75,113.64) node [anchor=north west][inner sep=0.75pt]    {$V'_{1}$};
\draw (807.42,280.68) node [anchor=north west][inner sep=0.75pt]    {$C( v_{B})$};
\draw (166.87,101.63) node [anchor=north west][inner sep=0.75pt]    {$u_{1}$};
\draw (161.33,146.75) node [anchor=north west][inner sep=0.75pt]    {$u_{2}$};
\draw (162.25,177.84) node [anchor=north west][inner sep=0.75pt]    {$u_{3}$};
\draw (123.42,216.95) node [anchor=north west][inner sep=0.75pt]    {$u_{4}$};
\draw (379.42,121.68) node [anchor=north west][inner sep=0.75pt]    {$v_{A}$};
\draw (246.11,104.64) node [anchor=north west][inner sep=0.75pt]    {$W_{1}$};
\draw (209.13,163.8) node [anchor=north west][inner sep=0.75pt]    {$W_{9}$};
\draw (219.3,129.71) node [anchor=north west][inner sep=0.75pt]    {$W_{3}$};
\draw (204.4,202.9) node [anchor=north west][inner sep=0.75pt]    {$W_{7}$};
\draw (24.75,104.64) node [anchor=north west][inner sep=0.75pt]    {$\#_{W}( u_{1}) =1$};
\draw (25.83,137.73) node [anchor=north west][inner sep=0.75pt]    {$\#_{W}( u_{2}) =3$};
\draw (25.83,169.82) node [anchor=north west][inner sep=0.75pt]    {$\#_{W}( u_{3}) =4$};
\draw (324.69,174.83) node [anchor=north west][inner sep=0.75pt]    {$W_{5}$};
\draw  [draw opacity=0][fill={rgb, 255:red, 139; green, 87; blue, 42 }  ,fill opacity=0.27 ][dash pattern={on 4.5pt off 4.5pt}]  (276.75,70.24) -- (304.75,70.24) -- (304.75,95.24) -- (276.75,95.24) -- cycle  ;
\draw (279.75,74.64) node [anchor=north west][inner sep=0.75pt]    {$H^{*}$};
\draw (186.65,387.33) node [anchor=north west][inner sep=0.75pt]    {$\vdots $};
\draw (380.65,308.33) node [anchor=north west][inner sep=0.75pt]    {$\vdots $};
\draw (598.65,393.33) node [anchor=north west][inner sep=0.75pt]    {$\vdots $};
\draw (828.65,319.33) node [anchor=north west][inner sep=0.75pt]    {$\vdots $};

\end{tikzpicture}

%% file: non-cycle-split-auxiliary-tikz
\tikzset{every picture/.style={line width=0.75pt}} 

\begin{tikzpicture}[x=0.75pt,y=0.75pt,yscale=-1,xscale=1]

\draw  [fill={rgb, 255:red, 245; green, 166; blue, 35 }  ,fill opacity=0.12 ][dash pattern={on 0.84pt off 2.51pt}] (557,173) .. controls (557,92.37) and (584.53,27) .. (618.49,27) .. controls (652.45,27) and (679.98,92.37) .. (679.98,173) .. controls (679.98,253.63) and (652.45,319) .. (618.49,319) .. controls (584.53,319) and (557,253.63) .. (557,173) -- cycle ;
\draw  [fill={rgb, 255:red, 208; green, 2; blue, 27 }  ,fill opacity=0.1 ][dash pattern={on 0.84pt off 2.51pt}] (688,170.5) .. controls (688,93.46) and (727.85,31) .. (777,31) .. controls (826.15,31) and (866,93.46) .. (866,170.5) .. controls (866,247.54) and (826.15,310) .. (777,310) .. controls (727.85,310) and (688,247.54) .. (688,170.5) -- cycle ;
\draw    (738,123) ;
\draw [shift={(738,123)}, rotate = 0] [color={rgb, 255:red, 0; green, 0; blue, 0 }  ][fill={rgb, 255:red, 0; green, 0; blue, 0 }  ][line width=0.75]      (0, 0) circle [x radius= 3.35, y radius= 3.35]   ;
\draw [shift={(738,123)}, rotate = 0] [color={rgb, 255:red, 0; green, 0; blue, 0 }  ][fill={rgb, 255:red, 0; green, 0; blue, 0 }  ][line width=0.75]      (0, 0) circle [x radius= 3.35, y radius= 3.35]   ;
\draw   (582,119.5) .. controls (582,99.34) and (595.43,83) .. (612,83) .. controls (628.57,83) and (642,99.34) .. (642,119.5) .. controls (642,139.66) and (628.57,156) .. (612,156) .. controls (595.43,156) and (582,139.66) .. (582,119.5) -- cycle ;
\draw    (775,147.34) ;
\draw [shift={(775,147.34)}, rotate = 0] [color={rgb, 255:red, 0; green, 0; blue, 0 }  ][fill={rgb, 255:red, 0; green, 0; blue, 0 }  ][line width=0.75]      (0, 0) circle [x radius= 3.35, y radius= 3.35]   ;
\draw [shift={(775,147.34)}, rotate = 0] [color={rgb, 255:red, 0; green, 0; blue, 0 }  ][fill={rgb, 255:red, 0; green, 0; blue, 0 }  ][line width=0.75]      (0, 0) circle [x radius= 3.35, y radius= 3.35]   ;
\draw    (744.5,109.34) ;
\draw [shift={(744.5,109.34)}, rotate = 0] [color={rgb, 255:red, 0; green, 0; blue, 0 }  ][fill={rgb, 255:red, 0; green, 0; blue, 0 }  ][line width=0.75]      (0, 0) circle [x radius= 3.35, y radius= 3.35]   ;
\draw [shift={(744.5,109.34)}, rotate = 0] [color={rgb, 255:red, 0; green, 0; blue, 0 }  ][fill={rgb, 255:red, 0; green, 0; blue, 0 }  ][line width=0.75]      (0, 0) circle [x radius= 3.35, y radius= 3.35]   ;
\draw    (784,136.34) ;
\draw [shift={(784,136.34)}, rotate = 0] [color={rgb, 255:red, 0; green, 0; blue, 0 }  ][fill={rgb, 255:red, 0; green, 0; blue, 0 }  ][line width=0.75]      (0, 0) circle [x radius= 3.35, y radius= 3.35]   ;
\draw [shift={(784,136.34)}, rotate = 0] [color={rgb, 255:red, 0; green, 0; blue, 0 }  ][fill={rgb, 255:red, 0; green, 0; blue, 0 }  ][line width=0.75]      (0, 0) circle [x radius= 3.35, y radius= 3.35]   ;
\draw    (758.5,101.34) ;
\draw [shift={(758.5,101.34)}, rotate = 0] [color={rgb, 255:red, 0; green, 0; blue, 0 }  ][fill={rgb, 255:red, 0; green, 0; blue, 0 }  ][line width=0.75]      (0, 0) circle [x radius= 3.35, y radius= 3.35]   ;
\draw [shift={(758.5,101.34)}, rotate = 0] [color={rgb, 255:red, 0; green, 0; blue, 0 }  ][fill={rgb, 255:red, 0; green, 0; blue, 0 }  ][line width=0.75]      (0, 0) circle [x radius= 3.35, y radius= 3.35]   ;
\draw    (775,109) ;
\draw [shift={(775,109)}, rotate = 0] [color={rgb, 255:red, 0; green, 0; blue, 0 }  ][fill={rgb, 255:red, 0; green, 0; blue, 0 }  ][line width=0.75]      (0, 0) circle [x radius= 3.35, y radius= 3.35]   ;
\draw [shift={(775,109)}, rotate = 0] [color={rgb, 255:red, 0; green, 0; blue, 0 }  ][fill={rgb, 255:red, 0; green, 0; blue, 0 }  ][line width=0.75]      (0, 0) circle [x radius= 3.35, y radius= 3.35]   ;
\draw    (743.5,143.34) ;
\draw [shift={(743.5,143.34)}, rotate = 0] [color={rgb, 255:red, 0; green, 0; blue, 0 }  ][fill={rgb, 255:red, 0; green, 0; blue, 0 }  ][line width=0.75]      (0, 0) circle [x radius= 3.35, y radius= 3.35]   ;
\draw [shift={(743.5,143.34)}, rotate = 0] [color={rgb, 255:red, 0; green, 0; blue, 0 }  ][fill={rgb, 255:red, 0; green, 0; blue, 0 }  ][line width=0.75]      (0, 0) circle [x radius= 3.35, y radius= 3.35]   ;
\draw    (783,120) -- (762,128.5) ;
\draw [shift={(783,120)}, rotate = 157.96] [color={rgb, 255:red, 0; green, 0; blue, 0 }  ][fill={rgb, 255:red, 0; green, 0; blue, 0 }  ][line width=0.75]      (0, 0) circle [x radius= 3.35, y radius= 3.35]   ;
\draw    (759,150) ;
\draw [shift={(759,150)}, rotate = 0] [color={rgb, 255:red, 0; green, 0; blue, 0 }  ][fill={rgb, 255:red, 0; green, 0; blue, 0 }  ][line width=0.75]      (0, 0) circle [x radius= 3.35, y radius= 3.35]   ;
\draw [shift={(759,150)}, rotate = 0] [color={rgb, 255:red, 0; green, 0; blue, 0 }  ][fill={rgb, 255:red, 0; green, 0; blue, 0 }  ][line width=0.75]      (0, 0) circle [x radius= 3.35, y radius= 3.35]   ;
\draw   (707,128.5) .. controls (707,96.19) and (731.62,70) .. (762,70) .. controls (792.38,70) and (817,96.19) .. (817,128.5) .. controls (817,160.81) and (792.38,187) .. (762,187) .. controls (731.62,187) and (707,160.81) .. (707,128.5) -- cycle ;
\draw    (637,178) ;
\draw [shift={(637,178)}, rotate = 0] [color={rgb, 255:red, 0; green, 0; blue, 0 }  ][fill={rgb, 255:red, 0; green, 0; blue, 0 }  ][line width=0.75]      (0, 0) circle [x radius= 3.35, y radius= 3.35]   ;
\draw [shift={(637,178)}, rotate = 0] [color={rgb, 255:red, 0; green, 0; blue, 0 }  ][fill={rgb, 255:red, 0; green, 0; blue, 0 }  ][line width=0.75]      (0, 0) circle [x radius= 3.35, y radius= 3.35]   ;
\draw    (630,135.53) ;
\draw [shift={(630,135.53)}, rotate = 0] [color={rgb, 255:red, 0; green, 0; blue, 0 }  ][fill={rgb, 255:red, 0; green, 0; blue, 0 }  ][line width=0.75]      (0, 0) circle [x radius= 3.35, y radius= 3.35]   ;
\draw [shift={(630,135.53)}, rotate = 0] [color={rgb, 255:red, 0; green, 0; blue, 0 }  ][fill={rgb, 255:red, 0; green, 0; blue, 0 }  ][line width=0.75]      (0, 0) circle [x radius= 3.35, y radius= 3.35]   ;
\draw    (599.25,99.24) ;
\draw [shift={(599.25,99.24)}, rotate = 0] [color={rgb, 255:red, 0; green, 0; blue, 0 }  ][fill={rgb, 255:red, 0; green, 0; blue, 0 }  ][line width=0.75]      (0, 0) circle [x radius= 3.35, y radius= 3.35]   ;
\draw [shift={(599.25,99.24)}, rotate = 0] [color={rgb, 255:red, 0; green, 0; blue, 0 }  ][fill={rgb, 255:red, 0; green, 0; blue, 0 }  ][line width=0.75]      (0, 0) circle [x radius= 3.35, y radius= 3.35]   ;
\draw    (613.78,118.13) ;
\draw [shift={(613.78,118.13)}, rotate = 0] [color={rgb, 255:red, 0; green, 0; blue, 0 }  ][fill={rgb, 255:red, 0; green, 0; blue, 0 }  ][line width=0.75]      (0, 0) circle [x radius= 3.35, y radius= 3.35]   ;
\draw [shift={(613.78,118.13)}, rotate = 0] [color={rgb, 255:red, 0; green, 0; blue, 0 }  ][fill={rgb, 255:red, 0; green, 0; blue, 0 }  ][line width=0.75]      (0, 0) circle [x radius= 3.35, y radius= 3.35]   ;
\draw    (642,208) -- (665,195) ;
\draw [shift={(642,208)}, rotate = 330.52] [color={rgb, 255:red, 0; green, 0; blue, 0 }  ][fill={rgb, 255:red, 0; green, 0; blue, 0 }  ][line width=0.75]      (0, 0) circle [x radius= 3.35, y radius= 3.35]   ;
\draw    (637,178) .. controls (677,148) and (661.5,178.25) .. (701.5,148.25) .. controls (740.5,119) and (747.18,147.74) .. (770.66,126.02) ;
\draw [shift={(772.5,124.25)}, rotate = 135] [fill={rgb, 255:red, 0; green, 0; blue, 0 }  ][line width=0.08]  [draw opacity=0] (10.72,-5.15) -- (0,0) -- (10.72,5.15) -- (7.12,0) -- cycle    ;
\draw    (660.92,197.85) .. controls (697.91,171.22) and (696.38,194.55) .. (711,177) .. controls (726,159) and (765,213) .. (784,136.34) ;
\draw [shift={(658,200)}, rotate = 323.13] [fill={rgb, 255:red, 0; green, 0; blue, 0 }  ][line width=0.08]  [draw opacity=0] (10.72,-5.15) -- (0,0) -- (10.72,5.15) -- (7.12,0) -- cycle    ;
\draw    (642,208) -- (758,211) -- (641,232) -- (763,227) -- (642,253) -- (770,250) -- (642,271) -- (777,278) ;
\draw [shift={(777,278)}, rotate = 0] [color={rgb, 255:red, 0; green, 0; blue, 0 }  ][fill={rgb, 255:red, 0; green, 0; blue, 0 }  ][line width=0.75]      (0, 0) circle [x radius= 3.35, y radius= 3.35]   ;
\draw [shift={(642,208)}, rotate = 1.48] [color={rgb, 255:red, 0; green, 0; blue, 0 }  ][fill={rgb, 255:red, 0; green, 0; blue, 0 }  ][line width=0.75]      (0, 0) circle [x radius= 3.35, y radius= 3.35]   ;
\draw    (777,278) .. controls (555.37,322.33) and (583.1,225.97) .. (622.21,194.4) ;
\draw [shift={(624,193)}, rotate = 143.13] [fill={rgb, 255:red, 0; green, 0; blue, 0 }  ][line width=0.08]  [draw opacity=0] (10.72,-5.15) -- (0,0) -- (10.72,5.15) -- (7.12,0) -- cycle    ;
\draw    (618,195) -- (637,178) ;

\draw (797.42,174.68) node [anchor=north west][inner sep=0.75pt]    {$C( v_{A})$};
\draw (748.31,77.08) node [anchor=north west][inner sep=0.75pt]  [font=\footnotesize]  {$v_{A}^{( 0)}$};
\draw (780.31,87.08) node [anchor=north west][inner sep=0.75pt]  [font=\footnotesize]  {$v_{A}^{( 1)}$};
\draw (790.31,108.08) node [anchor=north west][inner sep=0.75pt]  [font=\footnotesize]  {$v_{A}^{( 2)}$};
\draw (792.31,136.08) node [anchor=north west][inner sep=0.75pt]  [font=\footnotesize]  {$v_{A}^{( 3)}$};
\draw (777,150.74) node [anchor=north west][inner sep=0.75pt]  [font=\footnotesize]  {$v_{A}^{( 4)}$};
\draw (717,93.74) node [anchor=north west][inner sep=0.75pt]  [font=\footnotesize]  {$v_{A}^{( 8)}$};
\draw (710,115.74) node [anchor=north west][inner sep=0.75pt]  [font=\footnotesize]  {$v_{A}^{( 7)}$};
\draw (716,137.74) node [anchor=north west][inner sep=0.75pt]  [font=\footnotesize]  {$v_{A}^{( 6)}$};
\draw (740,155.74) node [anchor=north west][inner sep=0.75pt]  [font=\footnotesize]  {$v_{A}^{( 5)}$};
\draw (605.87,83.53) node [anchor=north west][inner sep=0.75pt]    {$u_{1}$};
\draw (590.33,107.65) node [anchor=north west][inner sep=0.75pt]    {$u_{2}$};
\draw (607.25,125.73) node [anchor=north west][inner sep=0.75pt]    {$u_{3}$};
\draw (580.75,59.4) node [anchor=north west][inner sep=0.75pt]    {$\mathsf{wit}_{1}(\mathcal{E}_{s=3})$};
\draw (601.42,161.95) node [anchor=north west][inner sep=0.75pt]    {$u_{4}^{( 0)}$};
\draw (655.4,139.9) node [anchor=north west][inner sep=0.75pt]    {$W_{7}$};
\draw (705.4,185.9) node [anchor=north west][inner sep=0.75pt]    {$W_{12}$};
\draw (612.42,206.95) node [anchor=north west][inner sep=0.75pt]    {$u_{12}^{( 0)}$};
\draw (611.65,287.33) node [anchor=north west][inner sep=0.75pt]    {$\vdots $};
\draw (793.65,219.33) node [anchor=north west][inner sep=0.75pt]    {$\vdots $};

\end{tikzpicture}

%% file: cycle-split-auxiliary-tikz
\tikzset{every picture/.style={line width=0.75pt}} 

\begin{tikzpicture}[x=0.75pt,y=0.75pt,yscale=-1,xscale=1]

\draw  [fill={rgb, 255:red, 245; green, 166; blue, 35 }  ,fill opacity=0.12 ][dash pattern={on 0.84pt off 2.51pt}] (557,173) .. controls (557,92.37) and (584.53,27) .. (618.49,27) .. controls (652.45,27) and (679.98,92.37) .. (679.98,173) .. controls (679.98,253.63) and (652.45,319) .. (618.49,319) .. controls (584.53,319) and (557,253.63) .. (557,173) -- cycle ;
\draw  [fill={rgb, 255:red, 208; green, 2; blue, 27 }  ,fill opacity=0.1 ][dash pattern={on 0.84pt off 2.51pt}] (688,170.5) .. controls (688,93.46) and (727.85,31) .. (777,31) .. controls (826.15,31) and (866,93.46) .. (866,170.5) .. controls (866,247.54) and (826.15,310) .. (777,310) .. controls (727.85,310) and (688,247.54) .. (688,170.5) -- cycle ;
\draw    (738,123) ;
\draw [shift={(738,123)}, rotate = 0] [color={rgb, 255:red, 0; green, 0; blue, 0 }  ][fill={rgb, 255:red, 0; green, 0; blue, 0 }  ][line width=0.75]      (0, 0) circle [x radius= 3.35, y radius= 3.35]   ;
\draw [shift={(738,123)}, rotate = 0] [color={rgb, 255:red, 0; green, 0; blue, 0 }  ][fill={rgb, 255:red, 0; green, 0; blue, 0 }  ][line width=0.75]      (0, 0) circle [x radius= 3.35, y radius= 3.35]   ;
\draw   (582,119.5) .. controls (582,99.34) and (595.43,83) .. (612,83) .. controls (628.57,83) and (642,99.34) .. (642,119.5) .. controls (642,139.66) and (628.57,156) .. (612,156) .. controls (595.43,156) and (582,139.66) .. (582,119.5) -- cycle ;
\draw    (775,147.34) ;
\draw [shift={(775,147.34)}, rotate = 0] [color={rgb, 255:red, 0; green, 0; blue, 0 }  ][fill={rgb, 255:red, 0; green, 0; blue, 0 }  ][line width=0.75]      (0, 0) circle [x radius= 3.35, y radius= 3.35]   ;
\draw [shift={(775,147.34)}, rotate = 0] [color={rgb, 255:red, 0; green, 0; blue, 0 }  ][fill={rgb, 255:red, 0; green, 0; blue, 0 }  ][line width=0.75]      (0, 0) circle [x radius= 3.35, y radius= 3.35]   ;
\draw    (744.5,109.34) ;
\draw [shift={(744.5,109.34)}, rotate = 0] [color={rgb, 255:red, 0; green, 0; blue, 0 }  ][fill={rgb, 255:red, 0; green, 0; blue, 0 }  ][line width=0.75]      (0, 0) circle [x radius= 3.35, y radius= 3.35]   ;
\draw [shift={(744.5,109.34)}, rotate = 0] [color={rgb, 255:red, 0; green, 0; blue, 0 }  ][fill={rgb, 255:red, 0; green, 0; blue, 0 }  ][line width=0.75]      (0, 0) circle [x radius= 3.35, y radius= 3.35]   ;
\draw    (784,136.34) -- (762,128.5) ;
\draw [shift={(784,136.34)}, rotate = 199.61] [color={rgb, 255:red, 0; green, 0; blue, 0 }  ][fill={rgb, 255:red, 0; green, 0; blue, 0 }  ][line width=0.75]      (0, 0) circle [x radius= 3.35, y radius= 3.35]   ;
\draw    (758.5,101.34) ;
\draw [shift={(758.5,101.34)}, rotate = 0] [color={rgb, 255:red, 0; green, 0; blue, 0 }  ][fill={rgb, 255:red, 0; green, 0; blue, 0 }  ][line width=0.75]      (0, 0) circle [x radius= 3.35, y radius= 3.35]   ;
\draw [shift={(758.5,101.34)}, rotate = 0] [color={rgb, 255:red, 0; green, 0; blue, 0 }  ][fill={rgb, 255:red, 0; green, 0; blue, 0 }  ][line width=0.75]      (0, 0) circle [x radius= 3.35, y radius= 3.35]   ;
\draw    (775,109) ;
\draw [shift={(775,109)}, rotate = 0] [color={rgb, 255:red, 0; green, 0; blue, 0 }  ][fill={rgb, 255:red, 0; green, 0; blue, 0 }  ][line width=0.75]      (0, 0) circle [x radius= 3.35, y radius= 3.35]   ;
\draw [shift={(775,109)}, rotate = 0] [color={rgb, 255:red, 0; green, 0; blue, 0 }  ][fill={rgb, 255:red, 0; green, 0; blue, 0 }  ][line width=0.75]      (0, 0) circle [x radius= 3.35, y radius= 3.35]   ;
\draw    (743.5,143.34) ;
\draw [shift={(743.5,143.34)}, rotate = 0] [color={rgb, 255:red, 0; green, 0; blue, 0 }  ][fill={rgb, 255:red, 0; green, 0; blue, 0 }  ][line width=0.75]      (0, 0) circle [x radius= 3.35, y radius= 3.35]   ;
\draw [shift={(743.5,143.34)}, rotate = 0] [color={rgb, 255:red, 0; green, 0; blue, 0 }  ][fill={rgb, 255:red, 0; green, 0; blue, 0 }  ][line width=0.75]      (0, 0) circle [x radius= 3.35, y radius= 3.35]   ;
\draw    (782,120) ;
\draw [shift={(782,120)}, rotate = 0] [color={rgb, 255:red, 0; green, 0; blue, 0 }  ][fill={rgb, 255:red, 0; green, 0; blue, 0 }  ][line width=0.75]      (0, 0) circle [x radius= 3.35, y radius= 3.35]   ;
\draw    (759,150) ;
\draw [shift={(759,150)}, rotate = 0] [color={rgb, 255:red, 0; green, 0; blue, 0 }  ][fill={rgb, 255:red, 0; green, 0; blue, 0 }  ][line width=0.75]      (0, 0) circle [x radius= 3.35, y radius= 3.35]   ;
\draw [shift={(759,150)}, rotate = 0] [color={rgb, 255:red, 0; green, 0; blue, 0 }  ][fill={rgb, 255:red, 0; green, 0; blue, 0 }  ][line width=0.75]      (0, 0) circle [x radius= 3.35, y radius= 3.35]   ;
\draw   (707,128.5) .. controls (707,96.19) and (731.62,70) .. (762,70) .. controls (792.38,70) and (817,96.19) .. (817,128.5) .. controls (817,160.81) and (792.38,187) .. (762,187) .. controls (731.62,187) and (707,160.81) .. (707,128.5) -- cycle ;
\draw    (637,178) ;
\draw [shift={(637,178)}, rotate = 0] [color={rgb, 255:red, 0; green, 0; blue, 0 }  ][fill={rgb, 255:red, 0; green, 0; blue, 0 }  ][line width=0.75]      (0, 0) circle [x radius= 3.35, y radius= 3.35]   ;
\draw [shift={(637,178)}, rotate = 0] [color={rgb, 255:red, 0; green, 0; blue, 0 }  ][fill={rgb, 255:red, 0; green, 0; blue, 0 }  ][line width=0.75]      (0, 0) circle [x radius= 3.35, y radius= 3.35]   ;
\draw    (630,135.53) ;
\draw [shift={(630,135.53)}, rotate = 0] [color={rgb, 255:red, 0; green, 0; blue, 0 }  ][fill={rgb, 255:red, 0; green, 0; blue, 0 }  ][line width=0.75]      (0, 0) circle [x radius= 3.35, y radius= 3.35]   ;
\draw [shift={(630,135.53)}, rotate = 0] [color={rgb, 255:red, 0; green, 0; blue, 0 }  ][fill={rgb, 255:red, 0; green, 0; blue, 0 }  ][line width=0.75]      (0, 0) circle [x radius= 3.35, y radius= 3.35]   ;
\draw    (599.25,99.24) ;
\draw [shift={(599.25,99.24)}, rotate = 0] [color={rgb, 255:red, 0; green, 0; blue, 0 }  ][fill={rgb, 255:red, 0; green, 0; blue, 0 }  ][line width=0.75]      (0, 0) circle [x radius= 3.35, y radius= 3.35]   ;
\draw [shift={(599.25,99.24)}, rotate = 0] [color={rgb, 255:red, 0; green, 0; blue, 0 }  ][fill={rgb, 255:red, 0; green, 0; blue, 0 }  ][line width=0.75]      (0, 0) circle [x radius= 3.35, y radius= 3.35]   ;
\draw    (613.78,118.13) ;
\draw [shift={(613.78,118.13)}, rotate = 0] [color={rgb, 255:red, 0; green, 0; blue, 0 }  ][fill={rgb, 255:red, 0; green, 0; blue, 0 }  ][line width=0.75]      (0, 0) circle [x radius= 3.35, y radius= 3.35]   ;
\draw [shift={(613.78,118.13)}, rotate = 0] [color={rgb, 255:red, 0; green, 0; blue, 0 }  ][fill={rgb, 255:red, 0; green, 0; blue, 0 }  ][line width=0.75]      (0, 0) circle [x radius= 3.35, y radius= 3.35]   ;
\draw    (642,208) -- (665,195) ;
\draw [shift={(642,208)}, rotate = 330.52] [color={rgb, 255:red, 0; green, 0; blue, 0 }  ][fill={rgb, 255:red, 0; green, 0; blue, 0 }  ][line width=0.75]      (0, 0) circle [x radius= 3.35, y radius= 3.35]   ;
\draw    (637,178) .. controls (677,148) and (661.5,178.25) .. (701.5,148.25) .. controls (740.5,119) and (741.95,149.4) .. (765.17,127.77) ;
\draw [shift={(767,126)}, rotate = 135] [fill={rgb, 255:red, 0; green, 0; blue, 0 }  ][line width=0.08]  [draw opacity=0] (10.72,-5.15) -- (0,0) -- (10.72,5.15) -- (7.12,0) -- cycle    ;
\draw    (660.92,197.85) .. controls (697.91,171.22) and (696.38,194.55) .. (711,177) .. controls (726,159) and (765,213) .. (784,136.34) ;
\draw [shift={(658,200)}, rotate = 323.13] [fill={rgb, 255:red, 0; green, 0; blue, 0 }  ][line width=0.08]  [draw opacity=0] (10.72,-5.15) -- (0,0) -- (10.72,5.15) -- (7.12,0) -- cycle    ;
\draw    (642,208) -- (758,211) -- (641,232) -- (763,227) -- (642,253) -- (770,250) -- (642,271) -- (777,278) ;
\draw [shift={(777,278)}, rotate = 0] [color={rgb, 255:red, 0; green, 0; blue, 0 }  ][fill={rgb, 255:red, 0; green, 0; blue, 0 }  ][line width=0.75]      (0, 0) circle [x radius= 3.35, y radius= 3.35]   ;
\draw [shift={(642,208)}, rotate = 1.48] [color={rgb, 255:red, 0; green, 0; blue, 0 }  ][fill={rgb, 255:red, 0; green, 0; blue, 0 }  ][line width=0.75]      (0, 0) circle [x radius= 3.35, y radius= 3.35]   ;
\draw    (777,278) .. controls (555.37,322.33) and (583.1,225.97) .. (622.21,194.4) ;
\draw [shift={(624,193)}, rotate = 143.13] [fill={rgb, 255:red, 0; green, 0; blue, 0 }  ][line width=0.08]  [draw opacity=0] (10.72,-5.15) -- (0,0) -- (10.72,5.15) -- (7.12,0) -- cycle    ;
\draw    (618,195) -- (637,178) ;

\draw (797.42,174.68) node [anchor=north west][inner sep=0.75pt]    {$C( v_{A})$};
\draw (748.31,77.08) node [anchor=north west][inner sep=0.75pt]  [font=\footnotesize]  {$v_{A}^{( 0)}$};
\draw (780.31,87.08) node [anchor=north west][inner sep=0.75pt]  [font=\footnotesize]  {$v_{A}^{( 1)}$};
\draw (790.31,108.08) node [anchor=north west][inner sep=0.75pt]  [font=\footnotesize]  {$v_{A}^{( 2)}$};
\draw (792.31,136.08) node [anchor=north west][inner sep=0.75pt]  [font=\footnotesize]  {$v_{A}^{( 3)}$};
\draw (777,150.74) node [anchor=north west][inner sep=0.75pt]  [font=\footnotesize]  {$v_{A}^{( 4)}$};
\draw (717,93.74) node [anchor=north west][inner sep=0.75pt]  [font=\footnotesize]  {$v_{A}^{( 8)}$};
\draw (710,115.74) node [anchor=north west][inner sep=0.75pt]  [font=\footnotesize]  {$v_{A}^{( 7)}$};
\draw (716,137.74) node [anchor=north west][inner sep=0.75pt]  [font=\footnotesize]  {$v_{A}^{( 6)}$};
\draw (740,155.74) node [anchor=north west][inner sep=0.75pt]  [font=\footnotesize]  {$v_{A}^{( 5)}$};
\draw (605.87,83.53) node [anchor=north west][inner sep=0.75pt]    {$u_{1}$};
\draw (590.33,107.65) node [anchor=north west][inner sep=0.75pt]    {$u_{2}$};
\draw (607.25,125.73) node [anchor=north west][inner sep=0.75pt]    {$u_{3}$};
\draw (580.75,59.4) node [anchor=north west][inner sep=0.75pt]    {$\mathsf{wit}_{1}(\mathcal{E}_{s=3})$};
\draw (601.42,160.95) node [anchor=north west][inner sep=0.75pt]    {$u_{18}^{( 0)}$};
\draw (655.4,139.9) node [anchor=north west][inner sep=0.75pt]    {$W_{18}$};
\draw (705.4,185.9) node [anchor=north west][inner sep=0.75pt]    {$W_{12}$};
\draw (612.42,206.95) node [anchor=north west][inner sep=0.75pt]    {$u_{12}^{( 0)}$};
\draw (611.65,287.33) node [anchor=north west][inner sep=0.75pt]    {$\vdots $};
\draw (793.65,219.33) node [anchor=north west][inner sep=0.75pt]    {$\vdots $};

\end{tikzpicture}

%% file: tree-tikz
\tikzset{every picture/.style={line width=0.75pt}} 

\begin{tikzpicture}[x=0.75pt,y=0.75pt,yscale=-1,xscale=1]

\draw    (336,88) ;
\draw [shift={(336,88)}, rotate = 0] [color={rgb, 255:red, 0; green, 0; blue, 0 }  ][fill={rgb, 255:red, 0; green, 0; blue, 0 }  ][line width=0.75]      (0, 0) circle [x radius= 3.35, y radius= 3.35]   ;
\draw [shift={(336,88)}, rotate = 0] [color={rgb, 255:red, 0; green, 0; blue, 0 }  ][fill={rgb, 255:red, 0; green, 0; blue, 0 }  ][line width=0.75]      (0, 0) circle [x radius= 3.35, y radius= 3.35]   ;
\draw    (336,88) -- (374,106) ;
\draw [shift={(374,106)}, rotate = 25.35] [color={rgb, 255:red, 0; green, 0; blue, 0 }  ][fill={rgb, 255:red, 0; green, 0; blue, 0 }  ][line width=0.75]      (0, 0) circle [x radius= 3.35, y radius= 3.35]   ;
\draw [shift={(336,88)}, rotate = 25.35] [color={rgb, 255:red, 0; green, 0; blue, 0 }  ][fill={rgb, 255:red, 0; green, 0; blue, 0 }  ][line width=0.75]      (0, 0) circle [x radius= 3.35, y radius= 3.35]   ;
\draw    (336,88) -- (295,106) ;
\draw [shift={(295,106)}, rotate = 156.3] [color={rgb, 255:red, 0; green, 0; blue, 0 }  ][fill={rgb, 255:red, 0; green, 0; blue, 0 }  ][line width=0.75]      (0, 0) circle [x radius= 3.35, y radius= 3.35]   ;
\draw [shift={(336,88)}, rotate = 156.3] [color={rgb, 255:red, 0; green, 0; blue, 0 }  ][fill={rgb, 255:red, 0; green, 0; blue, 0 }  ][line width=0.75]      (0, 0) circle [x radius= 3.35, y radius= 3.35]   ;
\draw    (253,125) ;
\draw [shift={(253,125)}, rotate = 0] [color={rgb, 255:red, 0; green, 0; blue, 0 }  ][fill={rgb, 255:red, 0; green, 0; blue, 0 }  ][line width=0.75]      (0, 0) circle [x radius= 3.35, y radius= 3.35]   ;
\draw [shift={(253,125)}, rotate = 0] [color={rgb, 255:red, 0; green, 0; blue, 0 }  ][fill={rgb, 255:red, 0; green, 0; blue, 0 }  ][line width=0.75]      (0, 0) circle [x radius= 3.35, y radius= 3.35]   ;
\draw    (253,125) -- (291,143) ;
\draw [shift={(291,143)}, rotate = 25.35] [color={rgb, 255:red, 0; green, 0; blue, 0 }  ][fill={rgb, 255:red, 0; green, 0; blue, 0 }  ][line width=0.75]      (0, 0) circle [x radius= 3.35, y radius= 3.35]   ;
\draw [shift={(253,125)}, rotate = 25.35] [color={rgb, 255:red, 0; green, 0; blue, 0 }  ][fill={rgb, 255:red, 0; green, 0; blue, 0 }  ][line width=0.75]      (0, 0) circle [x radius= 3.35, y radius= 3.35]   ;
\draw    (253,125) -- (212,143) ;
\draw [shift={(212,143)}, rotate = 156.3] [color={rgb, 255:red, 0; green, 0; blue, 0 }  ][fill={rgb, 255:red, 0; green, 0; blue, 0 }  ][line width=0.75]      (0, 0) circle [x radius= 3.35, y radius= 3.35]   ;
\draw [shift={(253,125)}, rotate = 156.3] [color={rgb, 255:red, 0; green, 0; blue, 0 }  ][fill={rgb, 255:red, 0; green, 0; blue, 0 }  ][line width=0.75]      (0, 0) circle [x radius= 3.35, y radius= 3.35]   ;
\draw    (212,143) -- (185,166) ;
\draw [shift={(185,166)}, rotate = 139.57] [color={rgb, 255:red, 0; green, 0; blue, 0 }  ][fill={rgb, 255:red, 0; green, 0; blue, 0 }  ][line width=0.75]      (0, 0) circle [x radius= 3.35, y radius= 3.35]   ;
\draw [shift={(212,143)}, rotate = 139.57] [color={rgb, 255:red, 0; green, 0; blue, 0 }  ][fill={rgb, 255:red, 0; green, 0; blue, 0 }  ][line width=0.75]      (0, 0) circle [x radius= 3.35, y radius= 3.35]   ;
\draw    (374,106) -- (412,124) ;
\draw [shift={(412,124)}, rotate = 25.35] [color={rgb, 255:red, 0; green, 0; blue, 0 }  ][fill={rgb, 255:red, 0; green, 0; blue, 0 }  ][line width=0.75]      (0, 0) circle [x radius= 3.35, y radius= 3.35]   ;
\draw [shift={(374,106)}, rotate = 25.35] [color={rgb, 255:red, 0; green, 0; blue, 0 }  ][fill={rgb, 255:red, 0; green, 0; blue, 0 }  ][line width=0.75]      (0, 0) circle [x radius= 3.35, y radius= 3.35]   ;
\draw    (295,106) -- (254,124) ;
\draw [shift={(254,124)}, rotate = 156.3] [color={rgb, 255:red, 0; green, 0; blue, 0 }  ][fill={rgb, 255:red, 0; green, 0; blue, 0 }  ][line width=0.75]      (0, 0) circle [x radius= 3.35, y radius= 3.35]   ;
\draw [shift={(295,106)}, rotate = 156.3] [color={rgb, 255:red, 0; green, 0; blue, 0 }  ][fill={rgb, 255:red, 0; green, 0; blue, 0 }  ][line width=0.75]      (0, 0) circle [x radius= 3.35, y radius= 3.35]   ;
\draw    (415,125) ;
\draw [shift={(415,125)}, rotate = 0] [color={rgb, 255:red, 0; green, 0; blue, 0 }  ][fill={rgb, 255:red, 0; green, 0; blue, 0 }  ][line width=0.75]      (0, 0) circle [x radius= 3.35, y radius= 3.35]   ;
\draw [shift={(415,125)}, rotate = 0] [color={rgb, 255:red, 0; green, 0; blue, 0 }  ][fill={rgb, 255:red, 0; green, 0; blue, 0 }  ][line width=0.75]      (0, 0) circle [x radius= 3.35, y radius= 3.35]   ;
\draw    (412,124) -- (450,142) ;
\draw [shift={(450,142)}, rotate = 25.35] [color={rgb, 255:red, 0; green, 0; blue, 0 }  ][fill={rgb, 255:red, 0; green, 0; blue, 0 }  ][line width=0.75]      (0, 0) circle [x radius= 3.35, y radius= 3.35]   ;
\draw [shift={(412,124)}, rotate = 25.35] [color={rgb, 255:red, 0; green, 0; blue, 0 }  ][fill={rgb, 255:red, 0; green, 0; blue, 0 }  ][line width=0.75]      (0, 0) circle [x radius= 3.35, y radius= 3.35]   ;
\draw    (415,125) -- (374,143) ;
\draw [shift={(374,143)}, rotate = 156.3] [color={rgb, 255:red, 0; green, 0; blue, 0 }  ][fill={rgb, 255:red, 0; green, 0; blue, 0 }  ][line width=0.75]      (0, 0) circle [x radius= 3.35, y radius= 3.35]   ;
\draw [shift={(415,125)}, rotate = 156.3] [color={rgb, 255:red, 0; green, 0; blue, 0 }  ][fill={rgb, 255:red, 0; green, 0; blue, 0 }  ][line width=0.75]      (0, 0) circle [x radius= 3.35, y radius= 3.35]   ;
\draw    (291,143) -- (264,166) ;
\draw [shift={(264,166)}, rotate = 139.57] [color={rgb, 255:red, 0; green, 0; blue, 0 }  ][fill={rgb, 255:red, 0; green, 0; blue, 0 }  ][line width=0.75]      (0, 0) circle [x radius= 3.35, y radius= 3.35]   ;
\draw [shift={(291,143)}, rotate = 139.57] [color={rgb, 255:red, 0; green, 0; blue, 0 }  ][fill={rgb, 255:red, 0; green, 0; blue, 0 }  ][line width=0.75]      (0, 0) circle [x radius= 3.35, y radius= 3.35]   ;
\draw    (374,143) -- (347,166) ;
\draw [shift={(347,166)}, rotate = 139.57] [color={rgb, 255:red, 0; green, 0; blue, 0 }  ][fill={rgb, 255:red, 0; green, 0; blue, 0 }  ][line width=0.75]      (0, 0) circle [x radius= 3.35, y radius= 3.35]   ;
\draw [shift={(374,143)}, rotate = 139.57] [color={rgb, 255:red, 0; green, 0; blue, 0 }  ][fill={rgb, 255:red, 0; green, 0; blue, 0 }  ][line width=0.75]      (0, 0) circle [x radius= 3.35, y radius= 3.35]   ;
\draw    (451,144) -- (424,167) ;
\draw [shift={(424,167)}, rotate = 139.57] [color={rgb, 255:red, 0; green, 0; blue, 0 }  ][fill={rgb, 255:red, 0; green, 0; blue, 0 }  ][line width=0.75]      (0, 0) circle [x radius= 3.35, y radius= 3.35]   ;
\draw [shift={(451,144)}, rotate = 139.57] [color={rgb, 255:red, 0; green, 0; blue, 0 }  ][fill={rgb, 255:red, 0; green, 0; blue, 0 }  ][line width=0.75]      (0, 0) circle [x radius= 3.35, y radius= 3.35]   ;
\draw    (187,167) ;
\draw [shift={(187,167)}, rotate = 0] [color={rgb, 255:red, 0; green, 0; blue, 0 }  ][fill={rgb, 255:red, 0; green, 0; blue, 0 }  ][line width=0.75]      (0, 0) circle [x radius= 3.35, y radius= 3.35]   ;
\draw [shift={(187,167)}, rotate = 0] [color={rgb, 255:red, 0; green, 0; blue, 0 }  ][fill={rgb, 255:red, 0; green, 0; blue, 0 }  ][line width=0.75]      (0, 0) circle [x radius= 3.35, y radius= 3.35]   ;
\draw    (187,167) -- (198,196) ;
\draw [shift={(198,196)}, rotate = 69.23] [color={rgb, 255:red, 0; green, 0; blue, 0 }  ][fill={rgb, 255:red, 0; green, 0; blue, 0 }  ][line width=0.75]      (0, 0) circle [x radius= 3.35, y radius= 3.35]   ;
\draw [shift={(187,167)}, rotate = 69.23] [color={rgb, 255:red, 0; green, 0; blue, 0 }  ][fill={rgb, 255:red, 0; green, 0; blue, 0 }  ][line width=0.75]      (0, 0) circle [x radius= 3.35, y radius= 3.35]   ;
\draw    (185,166) -- (173,194) ;
\draw [shift={(173,194)}, rotate = 113.2] [color={rgb, 255:red, 0; green, 0; blue, 0 }  ][fill={rgb, 255:red, 0; green, 0; blue, 0 }  ][line width=0.75]      (0, 0) circle [x radius= 3.35, y radius= 3.35]   ;
\draw [shift={(185,166)}, rotate = 113.2] [color={rgb, 255:red, 0; green, 0; blue, 0 }  ][fill={rgb, 255:red, 0; green, 0; blue, 0 }  ][line width=0.75]      (0, 0) circle [x radius= 3.35, y radius= 3.35]   ;
\draw    (264,166) -- (275,195) ;
\draw [shift={(275,195)}, rotate = 69.23] [color={rgb, 255:red, 0; green, 0; blue, 0 }  ][fill={rgb, 255:red, 0; green, 0; blue, 0 }  ][line width=0.75]      (0, 0) circle [x radius= 3.35, y radius= 3.35]   ;
\draw [shift={(264,166)}, rotate = 69.23] [color={rgb, 255:red, 0; green, 0; blue, 0 }  ][fill={rgb, 255:red, 0; green, 0; blue, 0 }  ][line width=0.75]      (0, 0) circle [x radius= 3.35, y radius= 3.35]   ;
\draw    (264,166) -- (256,196) ;
\draw [shift={(256,196)}, rotate = 104.93] [color={rgb, 255:red, 0; green, 0; blue, 0 }  ][fill={rgb, 255:red, 0; green, 0; blue, 0 }  ][line width=0.75]      (0, 0) circle [x radius= 3.35, y radius= 3.35]   ;
\draw [shift={(264,166)}, rotate = 104.93] [color={rgb, 255:red, 0; green, 0; blue, 0 }  ][fill={rgb, 255:red, 0; green, 0; blue, 0 }  ][line width=0.75]      (0, 0) circle [x radius= 3.35, y radius= 3.35]   ;
\draw    (346,167) -- (358,197) ;
\draw [shift={(358,197)}, rotate = 68.2] [color={rgb, 255:red, 0; green, 0; blue, 0 }  ][fill={rgb, 255:red, 0; green, 0; blue, 0 }  ][line width=0.75]      (0, 0) circle [x radius= 3.35, y radius= 3.35]   ;
\draw [shift={(346,167)}, rotate = 68.2] [color={rgb, 255:red, 0; green, 0; blue, 0 }  ][fill={rgb, 255:red, 0; green, 0; blue, 0 }  ][line width=0.75]      (0, 0) circle [x radius= 3.35, y radius= 3.35]   ;
\draw    (346,167) -- (338,197) ;
\draw [shift={(338,197)}, rotate = 104.93] [color={rgb, 255:red, 0; green, 0; blue, 0 }  ][fill={rgb, 255:red, 0; green, 0; blue, 0 }  ][line width=0.75]      (0, 0) circle [x radius= 3.35, y radius= 3.35]   ;
\draw [shift={(346,167)}, rotate = 104.93] [color={rgb, 255:red, 0; green, 0; blue, 0 }  ][fill={rgb, 255:red, 0; green, 0; blue, 0 }  ][line width=0.75]      (0, 0) circle [x radius= 3.35, y radius= 3.35]   ;
\draw    (424,167) -- (435,198) ;
\draw [shift={(435,198)}, rotate = 70.46] [color={rgb, 255:red, 0; green, 0; blue, 0 }  ][fill={rgb, 255:red, 0; green, 0; blue, 0 }  ][line width=0.75]      (0, 0) circle [x radius= 3.35, y radius= 3.35]   ;
\draw [shift={(424,167)}, rotate = 70.46] [color={rgb, 255:red, 0; green, 0; blue, 0 }  ][fill={rgb, 255:red, 0; green, 0; blue, 0 }  ][line width=0.75]      (0, 0) circle [x radius= 3.35, y radius= 3.35]   ;
\draw    (424,167) -- (416,197) ;
\draw [shift={(416,197)}, rotate = 104.93] [color={rgb, 255:red, 0; green, 0; blue, 0 }  ][fill={rgb, 255:red, 0; green, 0; blue, 0 }  ][line width=0.75]      (0, 0) circle [x radius= 3.35, y radius= 3.35]   ;
\draw [shift={(424,167)}, rotate = 104.93] [color={rgb, 255:red, 0; green, 0; blue, 0 }  ][fill={rgb, 255:red, 0; green, 0; blue, 0 }  ][line width=0.75]      (0, 0) circle [x radius= 3.35, y radius= 3.35]   ;
\draw  [dash pattern={on 0.84pt off 2.51pt}]  (173,194) -- (173,260) -- (264,166) ;
\draw [shift={(264,166)}, rotate = 314.07] [color={rgb, 255:red, 0; green, 0; blue, 0 }  ][fill={rgb, 255:red, 0; green, 0; blue, 0 }  ][line width=0.75]      (0, 0) circle [x radius= 3.35, y radius= 3.35]   ;
\draw    (173,194) -- (173,260) ;
\draw [color={rgb, 255:red, 208; green, 2; blue, 27 }  ,draw opacity=1 ]   (173,260) ;
\draw [shift={(173,260)}, rotate = 0] [color={rgb, 255:red, 208; green, 2; blue, 27 }  ,draw opacity=1 ][fill={rgb, 255:red, 208; green, 2; blue, 27 }  ,fill opacity=1 ][line width=0.75]      (0, 0) circle [x radius= 3.35, y radius= 3.35]   ;

\draw (175,263.4) node [anchor=north west][inner sep=0.75pt]  [font=\footnotesize]  {$u_{8} =u_{5}$};
\draw (437,201.4) node [anchor=north west][inner sep=0.75pt]  [font=\footnotesize]  {$v^{\prime ( j'_{'7})}_{7}$};
\draw (405,205.4) node [anchor=north west][inner sep=0.75pt]  [font=\footnotesize]  {$v_{7}^{( j_{7})}$};
\draw (360,200.4) node [anchor=north west][inner sep=0.75pt]  [font=\footnotesize]  {$v^{\prime ( j'_{6})}_{6}$};
\draw (323,203.4) node [anchor=north west][inner sep=0.75pt]  [font=\footnotesize]  {$v_{6}^{( j_{6})}$};
\draw (402,152.4) node [anchor=north west][inner sep=0.75pt]  [font=\footnotesize]  {$u_{7}$};
\draw (322,153.4) node [anchor=north west][inner sep=0.75pt]  [font=\footnotesize]  {$u_{6}$};
\draw (243,149.4) node [anchor=north west][inner sep=0.75pt]  [font=\footnotesize]  {$u_{5}$};
\draw (159,149.4) node [anchor=north west][inner sep=0.75pt]  [font=\footnotesize]  {$u_{4}$};
\draw (267,200.4) node [anchor=north west][inner sep=0.75pt]  [font=\footnotesize]  {$v^{\prime ( j'_{'5})}_{5}$};
\draw (236,200.4) node [anchor=north west][inner sep=0.75pt]  [font=\footnotesize]  {$v_{5}^{( j_{5})}$};
\draw (190,200.4) node [anchor=north west][inner sep=0.75pt]  [font=\footnotesize]  {$v^{\prime ( j'_{4})}_{4}$};
\draw (142,199.4) node [anchor=north west][inner sep=0.75pt]  [font=\footnotesize]  {$v_{4}^{( j_{4})}$};
\draw (461,116.4) node [anchor=north west][inner sep=0.75pt]  [font=\footnotesize]  {$v^{\prime ( j'_{3})}_{3}$};
\draw (352,115.4) node [anchor=north west][inner sep=0.75pt]  [font=\footnotesize]  {$v_{3}^{( j_{3})}$};
\draw (300,121.4) node [anchor=north west][inner sep=0.75pt]  [font=\footnotesize]  {$v^{\prime ( j'_{2})}_{2}$};
\draw (186,112.4) node [anchor=north west][inner sep=0.75pt]  [font=\footnotesize]  {$v_{2}^{( j_{2})}$};
\draw (419,99.4) node [anchor=north west][inner sep=0.75pt]  [font=\footnotesize]  {$u_{3}$};
\draw (234,96.4) node [anchor=north west][inner sep=0.75pt]  [font=\footnotesize]  {$u_{2}$};
\draw (381,70.4) node [anchor=north west][inner sep=0.75pt]  [font=\footnotesize]  {$v^{\prime \left( j_{1}^{'}\right)}_{1}$};
\draw (270,80.4) node [anchor=north west][inner sep=0.75pt]  [font=\footnotesize]  {$v_{1}^{( j_{1})}$};
\draw (302,57.4) node [anchor=north west][inner sep=0.75pt]  [font=\footnotesize]  {$u_{1} =u_{\mathsf{root}}$};

\end{tikzpicture}

%% file: General-case.tex
\section{General Case}
\label{sec:general-case}

We prove Lemma \ref{maintech} for the general case.
Let $\calS=(\calS_1,\ldots,\calS_N)$ be a \valid move sequence of length $N=n\log^{10}n$ that consists of both $1$-moves and $2$-moves.
We will consider two cases and deal with them separately: 
  (1) the number of $1$-moves in $\calS$ is at least $N/\log^5 n$; and (2) the number of $1$-moves is at most $N/\log^5 n$.

\subsection{Case $1$} 

We consider the case when there are at least $N/\log^5 n$ many $1$-moves. In this case we show that there is a window $W$ of $\calS$ such that $\rank_{\arcs}(W)$ is large.
The arguments used in this case are similar to those used in \cite{roglin2008complexity,bibak2019improving,coordination}. Given a window $W$ of $\calS$, we write $V_2(W)$ to denote the set of nodes $u\in V(W)$ such that at least two $1$-moves in $W$ are $\{u\}$. 

\begin{lemma}
There is a window $W$ of $\calS$ such that 
$$|V_2(W)|=\Omega\left(\frac{\emph{\len}(W)}{\log^6 n}\right).$$
\end{lemma}
\begin{proof}
Any 1-move that is not the first 1-move of the vertex generates a new arc of $\calS$, so the total number of arcs is at least
$
|\adarcs(\calS)|\ge N/\log^5 n-n.
$
Define the length of an arc $\alpha$ $(i,j)$, $\len(\alpha)$, to be $j-i$.
Partition all arcs based on their length, for any integer $i$ that $0\le i\le \lfloor \log_2 N\rfloor$, define
\[\adarcs_i(\calS) := \Big\{\alpha: \alpha\in \adarcs(\calS),\ \len(\alpha)\in [2^i,2^{i+1})\Big\}.\]
Since $\sum_{i=0}^{\lfloor \log_2 N\rfloor} |\adarcs_i(\calS)| \ge N/\log^5 n-n$, there exists $i^*$ such that 
\[|\adarcs_{i^*}(\calS)|\ge \frac{N/\log^5 n-n}{\log_2 N+1}\ge  \frac{N}{10\log^6 n}.\]

Let $W_r'$ be a window of length $2^{i^*+2}$ starting at a uniformly random position in $\{-2^{i^*+2}+1,\cdots , N\}$, and $W_r=W_r'\cap [N]$. For any arc $\alpha\in \adarcs_{i^*}(\calS)$, there are $\len(W_r')-\len(\alpha)$ possible starting points for $W_r$ to contain $\alpha$. So
\[
\Pr\big[\alpha\in \adarcs(W_r)\big]\ge \frac{\len(W_r')-\len(\alpha)}{N+2^{i^*+2}-1}\ge \frac{2^{i^*+1}}{N+2^{i^*+2}}.
\]
From linearity of expectation, and $2^{i^*+2}\le 4N$,
\[
\E\big[|\adarcs_{i^*}(W_r)|\big]\ge |\adarcs_{i^*}(\calS)|\cdot \frac{2^{i^*+1}}{N+2^{i^*+2}}\ge \frac{2^{i^*+1}}{50\log^6 n}\ge \frac{\len(W_r)}{100\log^6 n}.
\]
We can pick $W$ so that $|\adarcs_{i^*}(W)|\ge  {\len(W)}/{100\log^6 n}$. By definition of an arc, any vertex in one of the arcs in $\adarcs_{i^*}(W)$ must be in $V_2(W)$. On the other hand, any arc $\alpha\in \adarcs_{i^*}(W)$ has length at least $2^{i^*}\ge \len(W)/4$. So any vertex can have at most 4 arcs in $\adarcs_{i^*}(W)$. We have 
\[V_2(W)\ge \text{\# vertices in }\adarcs_{i^*}(W) \ge |\adarcs_{i^*}(W)|/4\ge \frac{\len(W)}{400\log^6 n}\]
This finishes the proof of the lemma.
\end{proof}

\begin{lemma}
We have $\emph{\rank}_{\emph{\arcs}}(W)\ge \Omega(|V_2(W)|)$.
\end{lemma}
\begin{proof}
Let $u_1,u_2,\cdots, u_k$ be the vertices in $V_2(W)$, and $\alpha_j=(s_j,e_j)$ be the arc of $u_j$ formed by its first and second 1-move. Since sequence $W$ is a \valid move sequence, there exists a vertex $v_j\not=u_j$ that moved odd number of times between $s_j$ and $e_j$, i.e., $\tau_{s_j-1}(v_j)=-\tau_{e_j-1}(v_j)$. Pick an arbitrary such $v_j$ for each $j$. Take a subset $U$ of $V_2(W)$ by the following process: 
\begin{itemize}
    \item $V\leftarrow V_2(W)$
    \item $U\leftarrow \emptyset$
    \item For $j$ from 1 to $k$, if $u_j\in  V$, $V\leftarrow V\backslash \{u_j,v_j\}$, $U\leftarrow U\cup\{u_j\}$.
\end{itemize}
In each step we delete at most two element from $V$ and add one element to $U$, so $|U|\ge |V_2(W)|/2$. Let $U=\{u_{i_1},u_{i_2},\cdots ,u_{i_m}\}$, ordered by the sequence they are added. By the process, for any $u_{i_j}\in U$, and any $j'>j$, $v_{i_{j'}}\not= u_{i_j}$.

Recall 
\[
\imp_{\tau_0,W}(\alpha_{i_j})_{(u_{i_j},v_{i_j})} = -\tau_{s_{i_j}-1}(v_{i_j})+\tau_{e_{i_j}-1}(v_{i_j}) = 2\tau_{e_{i_j}-1}(v_{i_j})\not=0.
\]
And for any $j'>j$, $u_{i_j}\not= u_{i_j}$, $v_{i_{j'}}\not=u_{i_j}$, so $\imp_{\tau_0,W}(\alpha_{i_{j'}})_{(u_{i_j},v_{i_j})}=0$. Consider the matrix formed by taking the $j-$th column to be $\imp_{\tau_0,W}(\alpha_{i_j})$. The row indexed by $(u_{i_j},v_{i_j})$ would be of the form
\[
(\underbrace{*,*,\cdots,*,}_\text{$j-1$ unknown numbers} 2\tau_{e_{i_j}-1}(v_{i_j})\not=0,0,\cdots,0).
\]
 This means the matrix has a lower triangular square submatrix of size at least $m\ge |V_2(W)|/2$. So we have $\rank_{\arcs}(W)\ge |V_2(W)|/2$.
\end{proof}

\subsection{Case $2$}

Let $\calS$ be a \valid move sequence of length $N$ with no more than $N/\log^5n$ 
many $1$-moves.
Let $W$ be a window of $\calS$. We write $\#_W(u)$ to denote the number of moves (including both $1$-moves and $2$-moves) that $u$ appears in $W$, and  write 
  $\#_W^2(u)$ to denote the number of $2$-moves that $u$ appears in $W$.

We start by showing a lemma similar to Lemma \ref{lem:sequence-selection} in Section \ref{sec:good-window}.

\begin{lemma}\label{lem:sequence-selection2}
Let $\calS$ be a \valid move sequence of length $N= n\log^{10} n$ with no more than $N/\log^5 n$ many $1$-moves. 
Then there exists a window $W$ of $\calS$ 
  such that at least $\Omega(\emph{\len}(W)/\log n)$ many moves of $W$ are $2$-moves $W_i=\{u,v\}$ that satisfy 
\begin{equation}\label{heheeq2}
\log^{3} n\le \#^2_W(u)\le \#_W(u)\le 2\log^{7} n\quad\text{and}\quad
\#^2_W(v)\ge \log^{3} n
\end{equation}
\end{lemma}
\begin{proof} 
We would like to apply Lemma \ref{lem:sequence-selection} (which works on move sequences that consist of $2$-moves only).
To this end, we let
  $\calS'$ be the move sequence obtained from $\calS$ by removing all its $1$-moves.
Let $N':=\len(\calS')\ge (1-1/\log^5 n)N$.
Applying Lemma \ref{lem:sequence-selection} on $\calS'$,\footnote{Note that $\calS'$ has length not exactly $N$ but $(1-1/\log^5 n)N$ but the statement of Lemma  \ref{lem:sequence-selection}  still holds.}   there must exist a positive integer $L$ and among the $N'-L+1$ windows $W'$ of $\calS'$ of length $L$, at least $\Omega(1/\log n)$-fraction of them satisfy that $\Omega(L/\log n)$ many
  $2$-moves $\{u,v\}$ in it satisfy 
\begin{equation}\label{goalgoal}
\log^3 n\le \#^2_{W'}(u)\le   \log^{7} n\quad\text{and}\quad
\#^2_{W'}(v)\ge \log^{3} n.
\end{equation}
{Let's denote these windows of $\calS'$ by $W_1',\ldots,W_s'$ for some $s=\Omega( (N'-L+1)/\log n )$.}
For each $W_i'$ we let $\calS_{k_i}$ (or $\calS_{\ell_i}$) to denote the move in $\calS$ that corresponds to the first (or last, respectively) move in $W_i'$, and
  let $W_i$ denote the window $(\calS_{k_i},\ldots,\calS_{\ell_i})$ of $\calS$.

If $L\ge N'/2$, we can just take $W$ to be $W_1$.
We note that the number of $2$-moves in $W_1$ that satisfy (\ref{goalgoal}) is at least $\Omega(L/\log n)=\Omega(\len(W_1)/\log n)$ given that $L\ge N'/2$. 
On the other hand, the number of $u\in V(W_1)$ that appears in at least $\log^7n$ $1$-moves of $W_1$ is at most $(N/\log^5n)/\log^7n=O(N/\log^{12} n)$.
Thus, the number of $2$-moves $\{u,v\}$ in $W_1$ that satisfy (\ref{goalgoal}) but not $\#_{W_1}(u)\le 2\log^7 n$ is at most
$$
\log^7 n\cdot O\left(\frac{N}{\log^{12}n}\right)=o\left(\frac{L}{\log n}\right).
$$
So we assume below that $L<N'/2$.

We claim that $W_i$ can satisfy the condition of the lemma if $\len(W_i)\le (1+1/\log^2 n)L$.
To see this is the case, we note that the number of nodes $u\in W_i$ that appears in at least $\log^7 n$ many $1$-moves is at most $(L/\log^2 n)/\log^7 n=O(L/\log^9 n)$.
Thus, the number of $2$-moves $\{u,v\}$ in $W_i$ that satisfy (\ref{goalgoal}) but not $\#_{W_i}(u)\le 2\log^7 n$ is at most
$$
\log^7 n\cdot O\left(\frac{N}{\log^{9} n}\right)=o\left(\frac{L}{\log n}\right).
$$

So it suffices to show that $\len(W_i)\le (1+1/\log^2 n)\cdot  \len(W_i') $ for some window $W_i$.
Assume this is not the case. Then 
  the total number of $1$-moves in $W_1,\ldots,W_s$ is at least
$$
\frac{1}{\log^2 n}\cdot \sum_{i\in [s]} \len(W_i')\ge \Omega\left( \frac{(N'-L+1)L}{\log^3 n}\right)=\Omega\left(\frac{NL}{\log^3 n}\right)
$$
using $L< N'/2$.
However, each $1$-move can only appear in no more than $L$ many windows of length $L$. Given that there are only $N/\log^5 n$ many $1$-moves, the same number can be upper bounded by
$$
\frac{N}{\log^5 n}\cdot L,
$$
a contradiction.
This finishes the proof of the lemma.
\end{proof}

So we now have a \valid move sequence $W$ of length $L$, as a window of the original \valid  sequence $\calS$, such that the number of $2$-moves in $W$ that satisfy (\ref{heheeq2}) is at least $\Omega(L/\log n)$.
The rest of the proof follows the same arguments used in Section \ref{sec:finding-cycles}.
We give a sketch below.

First we define the same auxiliary graph $H=(V(W),E)$ such that there is a one-to-one correspondence between $E$ and $2$-moves in $W$.
Note that the degree of a node $u$ in $H$ is the same as $\#_W^2(u)$.

We then show that there are disjoint sets of nodes $V_1,V_2\subset V(W)$ and a subset of edges $E'\subseteq E$ that satisfy conditions similar to those of Lemma \ref{hehelemma51}: 

\begin{lemma} \label{lemma64}
There are two disjoint sets of nodes $V_1,V_2\subset V(W)$ and a subset of edges $E'\subseteq E$ such that 
\begin{enumerate}
\item Every edge in $E'$ has one node in $V_1$ and the other node in $V_2$;
\item $|V_1\cup V_2|=O(L/\log^3 n)$ and $|E'|=\Omega(L/\log n)$;
\item $\#_W (u)\le 2\log^7 n$ for every node $u\in V_1$.
\end{enumerate}
\end{lemma}

The proof is exactly the same as that of Lemma \ref{hehelemma51}, except that we define $V$ to be the set of nodex $v$ with $\#_W^2(v)\ge \log^3 n$ and $V_h$ to be the set of nodes $v$ with $\#_W(v)\ge 2\log^7 n$.

Next we focus on $E''$, which contains all edges in $E'$ of the same sign (or edges in $E'$ of different signs, whichever contains more edges).
Similarly every cycle in $E''$ corresponds to a dependent cycle of $W$. 
The case when $E''$ contains many parallel edges can be handled exactly the same way as in Lemma \ref{lemma54}.
So we may delete all parallel edges from $E''$, and finish the proof using Lemma \ref{main}.

The proof of Lemma \ref{main} for the general case is very similar, with the following changes:
\begin{flushleft}\begin{enumerate}
\item In the definition of $k_v$ for each $v\in \wit_2(\mathcal{E}_s)$, we need it to be the number of moves (including both $1$-moves and $2$-moves) in $W$ that involve at least one node in $\wit_1(v)$.
This can still be bounded from above by $|\wit_1(v)|\cdot 2\log^7 n$ since we have $\#_W(u)\le \log^7 n$ for all $u\in V_1$ as promised in Lemma \ref{lemma64} above. 
\item As we commented earlier, Claim \ref{claimhehe} works even when $W$ consists of both $1$-moves and $2$-moves.
\end{enumerate}\end{flushleft}
This finishes the proof of  
  Lemma \ref{maintech} for the general case.

%% file: csp.tex
\clearpage
\section{Binary Max-CSP and Function Optimization Problems}\label{csp}
We recall the definition of binary maximum constraint satisfaction problems, and more generally function optimization problems.
\begin{definition}
An instance of Binary Max-CSP (Constraint Satisfaction Problem), or MAX 2-CSP, consists
of a set $V = \{ x_1, \ldots, x_n \}$ of variables that can take values over $\{0,1\}$
and a set $C = \{c_1, \ldots, c_m \}$ of constraints 
with given respective weights $w_1, \ldots, w_m$,
where each constraint is a predicate on a pair of variables.
The MAX 2-CSP problem is: given an instance, find an assignment that
maximizes the sum of the weights of the satisfied constraints.
\end{definition}
Several problems can be viewed as special cases of Binary Max-CSP where the
predicates of the constraints are restricted to belong to
a fixed family ${\cal P}$ of predicates; 
this restricted version is denoted Max-CSP(${\cal P}$).
For example, the Max Cut problem in graphs is equivalent to Max-CSP(${\cal P}$) where
${\cal P}$ contains only the ``not-equal'' predicate
($x \neq y$, where $x,y$ are the two variables).
The Max Directed Cut problem, where the input graph is directed and we seek a partition of the nodes into two parts $N_1, N_2$ that maximizes the total weight of the edges directed from $N_1$ to $N_2$, corresponds to the case that ${\cal P}$ contains only the $<$ predicate (i.e. $x < y$).
MAX 2SAT corresponds to the case that ${\cal P}$ consists of all 4 possible clauses on two variables.

A generalization of MAX 2-CSP is the class of {\em Binary function optimization problems} (BFOP) where instead of constraints (predicates) we have functions on two arguments that take values in $\{0, 1, \ldots, d \}$ instead of $\{0,1\}$, where $d$ is a fixed constant (or even is polynomially bounded).
For convenience and consistency with the notation of configurations in the Max Cut problem, we will use in the following  $\{-1,1\}$ as the domain of the variables instead of $\{0,1\}$.
That is, the problem is: Given a set $V = \{ x_1, \ldots, x_n \}$ of variables with domain $D=\{-1,1\}$,  a set $F = \{f_1, \ldots, f_m \}$ of functions, where each $f_i$ is a function of a pair $(x_{i_1}, x_{i_2})$ of variables,
and given respective weights $w_1, \ldots, w_m$, find an assignment $\tau: V \rightarrow D$ to the variables  that maximizes $\sum_{i=1}^m w_i \cdot f_i (\tau(x_{i_1}),\tau(x_{i_2}))$.

Even though a function in BFOP (or a constraint in Max-2CSP) has two arguments, its value may depend on only one of them, i.e. it may be
essentially a unary function (or constraint). More generally, it may be that the two arguments of the function can be decoupled and the function can be separated into two unary functions.
We say that a binary function $f(x,y)$ is {\em separable} if there are unary functions $f_1, f_2$ such that $f(x,y) = f_1(x)+f_2(y)$ for all values of $x,y$; otherwise $f$ is {\em nonseparable}.
For binary domains there is a simple criterion for separability: a function $f(x,y)$ is separable if and only if $f(-1,-1)+f(1,1) = f(-1,1) +f(1,-1)$ \cite{chen2020}.
If in a given BFOP instance some binary functions are separable, then we can decompose them into the equivalent unary functions. Thus, we may assume, without loss of generality, that a given BFOP instance has unary and binary functions, where all the binary functions are nonseparable.  
We say that an instance is {\em complete}, if every pair of variables appear as the arguments of  a (nonseparable) binary function in the instance.

The 2-FLIP local search algorithm can be applied to a MAX 2-CSP or BFOP problem to compute a locally optimal assignment that cannot be improved by flipping the value of any one or two variables.
We will show that the smoothed complexity of 2-FLIP for any complete MAX 2-CSP or BFOP instance is (at most) quasipolynomial.

\cspthm*

\begin{proof}
Consider a (complete) instance $I$ of a BFOP problem
with $n$ variables and $m$ functions, 
and a sequence $\calS$ of moves of 2-FLIP starting from an initial configuration.
The proof follows the same structure as the proof for Max Cut. The only thing that changes is the improvement vector in each step, which depends on the
specific functions of the instance: the vector has one coordinate for each function $f_i$ in the instance and the entry is equal to the change in the value of the function resulting from the move. Arcs and cycles of $\calS$  are defined in the same way as in Max Cut, and the improvement vectors of arcs and cycles are defined in an analogous way from the improvement vectors of the moves. 

The heart of the proof for Max Cut is Lemma \ref{maintech} which showed that there is a window $\cal{W}$ and a set of arcs or a set of cycles of $\cal{W}$ whose improvement vectors have rank $\Omega(\frac{\len (\cal{W})}{\log^{10} n})$. We will show that the lemma holds for any BFOP problem.

We associate with the BFOP instance $I$ the graph $G$
where the nodes correspond to the variables of $I$ and the edges correspond to the binary functions of $I$; since $I$ is a complete instance, the graph $G$ is the complete graph, possibly with multiple edges connecting the same pair of nodes (if there are
multiple functions with the same pair of arguments).
We will identify the variables of $I$ with the nodes of $G$ and the functions of $I$ with the edges of $G$.

In the general case of the Max Cut problem, in Case 1 where there is a large number of 1-moves, we identified a window $\cal{W}$ and a large set $A'$ of arcs in the window whose set of improvement vectors are linearly independent. The argument relied only on the zero-nonzero structure of the improvement vectors: it showed that the matrix $M$ formed by these vectors and a set of rows corresponding to a certain set $E'$ of witness edges is a lower triangular matrix with nonzero diagonal. Take a set $F'$ of functions of $I$ that contains for each edge $\{u,v\} \in E'$ a function $f_k(u,v)$ with this pair as arguments (it exists because the instance $I$ is complete), and form the matrix  $M'$ with the set $F'$ as rows and the set $A'$ of arcs as columns.
We will show that the matrix $M'$ has the same zero-nonzero structure as $M$, thus it also has full rank.

Consider an arc of the move sequence $\calS$ corresponding to two moves $\calS_i =\{u\}$, $\calS_j =\{u\}$, $i<j$, and a function $f_k$ of $I$. If $u$ is not one of the arguments of the function, then the corresponding entry of the improvement vector of the arc is obviously 0.
If $u$ is one of the argument, i.e. the $k$-th function is $f_k(u,v)$  (similarly if it is $f_k(v,u)$), then the corresponding entry of the improvement vector of the arc is
$\gamma_i(u) [f_k(\gamma_i(u), \gamma_i(v)) -f_k(-\gamma_i(u), \gamma_i(v))] -  
\gamma_j(u) [f_k(\gamma_j(u), \gamma_j(v)) -f_k(-\gamma_j(u), \gamma_j(v))]$.
If $v$ moves an even number of times between $\calS_i$ and $\calS_j$, then $\gamma_i(v)=\gamma_j(v)$ and it follows that the entry is 0, both in the case that $\gamma_i(u)=\gamma_j(u)$ and 
in the case that $\gamma_i(u)=-\gamma_j(u)$.
On the other hand, if $v$ moves an odd number of times between $\calS_i$ and $\calS_j$, then $\gamma_i(v)=-\gamma_j(v)$ and it follows that the $k$-th entry of the improvement vector is $\gamma_i(u) [f_k(\gamma_i(u), \gamma_i(v)) -f_k(-\gamma_i(u), \gamma_i(v))] -  
\gamma_j(u) [f_k(\gamma_j(u), -\gamma_i(v)) -f_k(-\gamma_j(u), -\gamma_i(v))]$.
Letting $\gamma_i(u)=a, \gamma_i(v)=b$,
the entry is  $a[f_k(a,b) + f_k(-a,-b) - f_k(-a,b) - f_k(a,-b) ]$ (both when $\gamma_i(u)=\gamma_j(u)$
and when $\gamma_i(u) = -\gamma_j(u)$);
this quantity is nonzero because $f_k$ is nonseparable.
Thus, the entry for $f_k(u,v)$ of the improvement vector of the arc is nonzero exactly when the entry of the arc in the Max Cut problem for the edge $(u,v)$ is nonzero. 
It follows that the matrix $M'$ has the same zero-nonzero structure as $M$, thus it also has full rank.

In Case 2 of the Max Cut problem, where the number of 2-moves is very large, there were two subcases.
In the first subcase, where there are many parallel edges in the graph that we associated with the window of the move sequence, we found a large set of 2-cycles whose improvement vectors were linearly independent. In the other case, where there are "few" parallel edges, we constructed a large set of cycles (of length $O(\log n))$, again with linearly independent improvement vectors. In both cases, the proof of linear independence relied again only on the zero-nonzero structure of the vectors, and not on the precise value of the entries.
We will argue that in both cases, the corresponding vectors of these cycles in the BFOP instance $I$ have the same zero-nonzero structure.

In the first subcase we found many 2-cycles $(u_1,v_1), \ldots, (u_k,v_k)$, and corresponding "witness" edges $(u_1,z_1), \ldots, (u_k,z_k)$ such that the matrix $M$ with rows corresponding to the witness edges and columns corresponding to the 2-cycles in the Max Cut problem is lower triangular with nonzero diagonal. 
The nodes $u_i$ are distinct (the $v_i$ and the $z_i$ may not be distinct) and $z_i \neq u_i, v_i$ for all $i$.
For each witness pair $(u_i,z_i)$ pick a function $f_{r_i}$ of instance $I$ with this pair of variables as arguments, in either order, say wlog the function is $f_{r_i}(u_i,z_i)$.
Consider the matrix $M'$ with rows corresponding to the functions $f_{r_i}(u_i,z_i)$ and columns corresponding to the 2-cycles $(u_i,v_i)$.
Note that the entry $M(j,i)$ is nonzero if one of the nodes $u_j, z_j$ is in $\{u_i, v_i\}$ and the other node appears an odd number of times between the two moves $\{u_i, v_i\}$, and it is 0 otherwise, i.e. if $\{u_j, z_j\} \cap \{u_i, v_i\} = \emptyset$, or if one of $u_j, z_j$ is in $\{u_i, v_i\}$ and the other node appears an even number of times between the two moves $\{u_i, v_i\}$. Importantly it cannot be that $\{u_j, z_j\} = \{u_i, v_i\}$ because $u_j \neq u_i, v_i$.
Examining the value $M'(j,i)$ in the same way as in the case of arcs above, we observe that 
if $M(j,i)=0$ then also $M'(j,i)=0$, and if $M(j,i)\neq 0$ then also $M'(j,i)\neq 0$. Thus, $M'$ has the same zero-nonzero structure as $M$ and hence it has also full rank.

In the second subcase of Case 2, we found many cycles $C_1, \ldots C_k$ and corresponding witness edges $\{u_i,v_i\}$ such that for every $i$, (1) $C_i$ does not contain any $u_j$ for $j<i$, nor $v_i$,
(2) $C_i$ has exactly two edges incident to $u_i$ and node $v_i$ appears an odd number of times between the two moves corresponding to these two edges, (3) if $C_i$ contains $v_j$ for some $j<i$ (the cycle $C_i$ may go more than once through $v_j$), then
$u_j$ does not appear between any pair of moves that correspond to consecutive edges of the cycle $C_i$ incident to $v_i$.
We used these properties in Max Cut to show that the matrix $M$ whose rows correspond to the witness edges and the columns correspond to the cycles $C_i$ is lower triangular with nonzero diagonal.
As before, for each witness pair $(u_i,v_i)$ pick a function $f_{r_i}$ of instance $I$ with this pair of variables as arguments, and let $M'$ be the matrix with these functions as rows and the cycles $C_i$ as columns. 
We can use the above properties to show that the matrix $M'$ is also lower triangular with nonzero diagonal.
Property (2) and the fact that $v_i \notin C_i$ (from property (1)) imply that $M'(i,i) \neq 0$ for all $i$.
Properties (1) and (3) can be used to show that $M'(j,i) =0$ for all $j <i$. Therefore, $M'$ has full rank.

Once we have Lemma \ref{maintech} for the BFOP instance $I$, the rest of the proof is the same as for Max Cut. 
The only difference is that, if the maximum value of a function in $I$ is $d$ (a constant, or even polynomial in $n$), then the maximum absolute value of the objective function is $md$ instead of $n^2$ that it was in Max Cut.
\end{proof}

%% file: Conclusion.tex
\clearpage
\section{Conclusions}\label{conclusions}
We analyzed the smoothed complexity of the SWAP algorithm for Graph Partitioning and the 2-FLIP algorithm for Max Cut and showed that with high probability the algorithms terminate in quasi-polynomial time for any pivoting rule. The same result holds more generally for the class of maximum binary constraint satisfaction problems (like Max-2SAT, Max Directed Cut, and others). We have not made any attempt currently to optimize the exponent of $\log n$ in the bound, but we believe that with a more careful analysis the true exponent will be low.  There are several interesting open questions raised by this work. We list some of them below.

1. Can our bounds be improved to polynomial? In the case of the 1-FLIP algorithm in the full perturbation model (i.e. when all edges of $K_n$ are perturbed) a polynomial bound was proved in \cite{angel2017local}. Can a similar result be shown for 2-FLIP and SWAP?

2. Can our results be extended to the structured smoothed model, i.e., when we are given a graph $G$ and only the edges of $G$ are perturbed?  In the case of 1-FLIP we know that this holds \cite{etscheid2015smoothed,chen2020}, but 2-FLIP is much more challenging.

3. We saw in this paper how to analyze local search when one move flips simultaneously two nodes. This is a qualitative step up from the case of single flips, that creates a number of obstacles which had to be addressed. This involved the introduction of nontrivial new techniques in the analysis of the sequence of moves, going from sets to graphs. Dealing with local search that flips 3 or more nodes will require extending the methods further to deal with hypergraphs.  We hope that our techniques will form the basis for handling local search algorithms that flip multiple nodes in one move, e.g. $k$-FLIP for higher $k$, and even more ambitiously powerful methods like Kernighan-Lin that perform a deep search in each iteration and flip/swap an unbounded number of nodes.

4. Can our results be extended to Max $k$-Cut or $k$-Graph Partitioning where the graph is partitioned into $k>2$ parts? In the case of 1-FLIP for Max $k$-Cut quasi-polynomial bounds were shown in \cite{bibak2019improving}.

5. Can similar results be shown for Max-CSP with constraints of higher arities, for example Max 3SAT? 
No bounds are known even for 1-FLIP.  In fact, analyzing 1-FLIP for Max 3SAT seems to present challenges that have similarities with those encountered in the analysis of 2-FLIP for Max 2SAT and Max Cut, so it is possible that the techniques developed in this paper will be useful also in addressing this problem.

%% file: Appendix-Preliminaries-Missing-Proofs.tex
\section{Missing Proofs from Section~\ref{sec:setup}}
\CancellationOfExternalEdgesArc*

\begin{proof}
Note that $\tau_i(u)=\gamma_i(u)$, $\tau_j(u)=\gamma_j(u)$, and by definition, $\imp_{\tau_0,\calS}(i)=\imp_{\gamma_0,\calS}(i)^*$ is the projection of $\imp_{\gamma_0,\calS}(i)$ on $E(\calS)$. 

Thus, for any edge $e\in E(\calS)$, $w_e$ is the same as $(\tau_i(u)\cdot \imp_{\tau_0,\calS}(i)-\tau_j(u)\cdot \imp_{\tau_0,\calS}(j))_e$. For any edge $e\notin E(\calS)$, if $e=\{v_1,v_2\}$ doesn't contain $u$, the improvement vectors are 0 on $e$ and
correspondingly $w[u,i,j]_e=0$. For the last case, let us assume $e=\{u,v\}$ where $v$ is inactive. We have that
\[
\imp_{\gamma_0,\calS}(i)_e = \gamma_{i-1}(v)\gamma_{i-1}(u) = -\gamma_{i-1}(v)\gamma_{i}(u)
\ \text{ \& }\ 
\imp_{\gamma_0,\calS}(j)_e = -\gamma_{j-1}(v)\gamma_{j}(u).\]
Since $v$ is not active, $\gamma_{i}(v) = \gamma_{0}(v)$ for any $i\in[\ell]$. So, we get that
\[(\gamma_i(u)\cdot \imp_{\gamma_0,\calS}(i)-\gamma_j(u)\cdot 
\imp_{\gamma_0,\calS} (j))_e = -\gamma_{i-1}(v) + \gamma_{j-1}(v)=0.\]
This finishes the proof of the lemma.
\end{proof}

\CancellationOfExternalEdgesCycle*\begin{proof}
Again recall that $\tau_i(u)=\gamma_i(u)$, $\tau_j(u)=\gamma_j(u)$, and by definition, $\imp_{\tau_0,\calS}(i)=\imp_{\gamma_0,\calS}(i)^*$ is the projection of $\imp_{\gamma_0,\calS}(i)$ on $E(\calS)$.

Thus, for edge $e\in E(\calS)$, $w[C]_e$ is the same as $\sum_{j\in [t]} b_j\cdot \imp_{\tau_0,\calS}(c_j)_e$. For edge $e\notin E(\calS)$, if $e$ doesn't contain $u_i$ for any $i\in [t]$, the improvement vectors are 0 on $e$ and correspondingly $w[C]_e=0$. For the last case, let us assume $e=(u ,v)$ where $v$ is inactive. Then we have (where index $0$ corresponds to $t$)
\begin{align*}
    \sum_{j\in [t]} b_j\cdot \imp_{\gamma_0,\calS}(c_j)_e = & \sum_{i\in [t]: u_i=u}\Big(b_{i-1}\imp_{\gamma_0,\calS}(c_{i-1})_e+b_{i}\imp_{\gamma_0,\calS}(c_{i})_e\Big) \\
    = & \sum_{i\in [t]:u_i=u}\Big( -b_{i-1}\gamma_{c_{i-1}-1}(v)\gamma_{c_{i-1}}(u_i)-b_{i}\gamma_{c_{i}-1}(v)\gamma_{c_i}(u_{i})\Big).
\end{align*}
Since $v$ is inactive, $\gamma_{c_{i-1}-1}(v)=\gamma_{c_{i}-1}(v)=\gamma_0(v)$. Each term above is equal to 
\[
-\gamma_0(v)(b_{i-1}\gamma_{c_{i-1}}(u_i)+ b_{i}\gamma_{c_i}(u_{i}))=0.
\]
This finishes the proof of the lemma.
%
\end{proof}

%% file: Appendix-Rank-Invariance.tex
\section{Rank Invariance of Improving vectors over Initial Configuration}\label{sec:rank-invariance}
In this section we prove that the rank of the improvement vectors for the set of $\emph{1-move}(\calS)$, $\emph{2-move}(\calS)$, $\arcs(\calS)$ and $\cycles(\calS)$ is independent of the initial configuration $\gamma_0$ of the vertices.
\begin{proof}[Proof of Lemmas~\ref{lem:independence-init-rank-arcs}, \ref{lem:independence-init-dependent-cycles} \& \ref{lem:independence-init-rank-cycles}]

We start by recalling the following useful facts:

\begin{fact}\label{fact:config-invariance}
Let $\gamma_0,\gamma_0'\in \{\pm 1\}^n$ be two arbitrary initial configurations. Then, it holds that \[ \gamma_0(v)\gamma_0'(v)=\gamma_i(v)\gamma_i'(v) \quad \text{ for any } i\in[\ell],\]
where $\gamma_{i},\gamma_{i}'$ is obtained from $\gamma_{i-1},\gamma_{i-1}'$ by flipping nodes in $\calS_{i}$, for a move sequence $\calS=(\calS_1,\ldots,\calS_\ell)$. 
\end{fact}

\begin{fact}\label{fact:rank-invariance}
Let $A$ be a  $(k_1\times k_2)$ real-valued matrix $A$ and $B,C$ are full-rank $(k_1\times k_1)$ and $(k_2\times k_2)$ squared matrices correspondingly. Then, it holds that $\rank(A)=\rank(A^\top)$
and $ \rank(A)=\rank(BAC)$.
\end{fact}

\noindent Let $
\mathcal{M}_{ \emph{$1$-move}}(\gamma_0,\calS),
\mathcal{M}_{ \emph{$2$-move}}(\gamma_0,\calS)$ be the matrices whose columns are the improvement vectors  of $1$-move  for a given initial configuration $\gamma_0$ and let $\mathcal{M}_{ \emph{$1$-move}}(\gamma_0,\calS),
\mathcal{M}_{ \emph{$2$-move}}(\gamma_0,\calS)$ be the submatrices of $\mathcal{M}_{ \emph{$1$-move}}(\tau_0,\calS),
\mathcal{M}_{ \emph{$2$-move}}(\tau_0,\calS)$ including only the rows which correspond to $E(\calS)$. 

Schematically, we have that for the \emph{1-move} case:
\[
\mathcal{M}_{ \emph{$1$-move}}(\tau_0,\calS)=
\begin{bmatrix}
 & \vdots  & \\
\cdots  & (\imp_{\tau{0} ,\mathcal{S}}( \calS_k=\{u\})_{e=\{u,v\}}=\tau_{k-1}(u)\tau_{k-1}(v) & \cdots \\
 & \vdots  & \\
\cdots  & (\imp_{\tau{0} ,\mathcal{S}}( \calS_k=\{u\})_{e=\{z,w\}}=0 & \cdots \\
 & \vdots  & \\
\end{bmatrix}_{|E(\calS)|,|\emph{1-move}(\calS)|}
\]
and for the \emph{2-move} case:
\[
\mathcal{M}_{ \emph{$2$-move}}(\tau_0,\calS)=
\begin{bmatrix}
 & \vdots  & \\
\cdots  & (\imp_{\tau{0} ,\mathcal{S}}( \calS_{k'}=\{u,v\})_{e=\{u,v\}}=0 & \cdots \\
\cdots  & (\imp_{\tau{0} ,\mathcal{S}}( \calS_{k'}=\{u,v\})_{e=\{u,z\}}=\tau_{{k'}-1}(u)\tau_{{k'}-1}(z) & \cdots \\
\cdots  & (\imp_{\tau{0} ,\mathcal{S}}( \calS_{k'}=\{u,v\})_{e=\{v,w\}}=\tau_{{k'}-1}(v)\tau_{{k'}-1}(w) & \cdots \\
 & \vdots  & \\
\cdots  & (\imp_{\tau{0} ,\mathcal{S}}( \calS_{k'}=\{u,v\})_{e=\{z,w\}}=0 & \cdots \\ & \vdots  & \\
\end{bmatrix}_{|E(\calS)|,|\emph{2-move}(\calS)|}
\]
Notice that we can derive $\mathcal{M}_{ \emph{$1$-move}}(\tau_0',\calS),\mathcal{M}_{ \emph{$2$-move}}(\tau_0',\calS)$ by multiplying $\mathcal{M}_{ \emph{$1$-move}}(\tau_0,\calS),\mathcal{M}_{ \emph{$2$-move}}(\tau_0,\calS)$ from left by
the squared diagonal $|E_n|\times|E_n|$ matrix $\mathcal{D}[\tau_0,\tau_0']$ such that
\[\mathcal{D}[\tau_0,\tau_0']_{(e,e)=((u,v),(u,v))}=\tau_0(u)\tau_0'(u)\tau_0(v)\tau_0'(v)\text{ and } \mathcal{D}[\tau_0,\tau_0']_{(e,e')}=0\text{ for $e\neq e'$.}\] 
Indeed,   we have that for an entry $(e,k)$ representing an edge $e=(u,v)$ and the $k$-th $\mu$-move 
\begin{align*}
(\mathcal{D}[\tau_0,\tau_0']\mathcal{M}_{ \emph{$\mu$-move}}(\tau_0,\calS))_{e,k}&=& &
\begin{cases}
\tau_0(u)\tau_0'(u)\tau_0(v)\tau_0'(v)\tau_k(u)\tau_k(v) & \text{if }\mathcal{M}_{ \emph{$\mu$-move}}(\tau_0,\calS))_{e,k}=\tau_k(u)\tau_k(v)\\
0& \text{if }\mathcal{M}_{ \emph{$\mu$-move}}(\tau_0,\calS))_{e,k}=0
\end{cases}\\
&\overset{\footnotesize Fact~\ref{fact:config-invariance}}{=}& &
\begin{cases}
\tau_k(u)\tau_k'(u)\tau_k(v)\tau_k'(v)\tau_k(u)\tau_k(v) & \text{if }\mathcal{M}_{ \emph{$\mu$-move}}(\tau_0,\calS))_{e,k}=\tau_k(u)\tau_k(v)\\
0& \text{if }\mathcal{M}_{ \emph{$\mu$-move}}(\tau_0,\calS))_{e,k}=0
\end{cases}
\end{align*}
Since  $\tau_i(v)^2$ equals 1 for any $v\in V_n$ for every $i\in[\ell]$, we get that 
\begin{align*}
(\mathcal{D}[\tau_0,\tau_0']\mathcal{M}_{ \emph{$\mu$-move}}(\tau_0,\calS))_{e,k}&=& &
\begin{Bmatrix*}[l]
\tau_k'(u)\tau_k'(v)& \text{if }\mathcal{M}_{ \emph{$\mu$-move}}(\tau_0,\calS))_{e,k}=\tau_k(u)\tau_k(v)\\
0&  \text{if }\mathcal{M}_{ \emph{$\mu$-move}}(\tau_0,\calS))_{e,k}=0
\end{Bmatrix*} = (\mathcal{M}_{ \emph{$\mu$-move}}(\tau_0',\calS))_{e,k}
\end{align*}
 for any $\mu\in\{1,2\}$. Since $\mathcal{D}[\tau_0,\tau_0']$ is a full-rank matrix, leveraging Fact~\ref{fact:rank-invariance}, the above argument proves that $\rank(
\mathcal{M}_{ \emph{$\mu$-move}}(\tau_0,\calS))$ is independent of the initial configuration.

More interestingly, in order to prove Lemma~\ref{lem:independence-init-rank-arcs}, it suffices to prove that the rank of the column matrix with the improvement vectors of arcs is independent of the initial configuration. 
In fact, let \[
M_{\arcs(\mathcal{S})}(\tau_0,\calS)) =\begin{bmatrix}
\vdots  &  & \vdots  &  & \vdots \\
\imp_{\tau _{0} ,\mathcal{S}}( \alpha _{1}) & \cdots  & \imp_{\tau _{0} ,\mathcal{S}}( \alpha _{k}) & \cdots  & \imp_{\tau _{0} ,\mathcal{S}}( \alpha _{|\arcs(\mathcal{S}) |})\\
\vdots  &  & \vdots  &  & \vdots 
\end{bmatrix}
\]
By definition, we have that for any $\alpha\in\arcs(\calS)$:
\[\imp_{\tau_0,\calS}(\alpha)= \tau_i(u)\cdot \imp_{\tau_0,\calS}(i)-\tau_j(u)\cdot 
  \imp_{\tau_0,\calS}(j)\in \mathbb{Z}^{E(\calS)},\]
or equivalently $M_{\arcs(\mathcal{S})}(\tau_0,\calS))=\mathcal{M}_{ \emph{$1$-move}}(\tau_0,\calS))\cdot \mathcal{T}(\tau_0,\calS)$, where $\mathcal{T}(\tau_0,\calS)$ is a sparse ($|\emph{1-move}(\calS)| \times |\arcs(S)|$) rectangular matrix such that \[\mathcal{T}(\tau_0,\calS)_{k-\text{th \emph{1-move}},\alpha=(i,j)}=\begin{cases}
(-1)^{\tfrac{(k-i)}{(j-i)}}\tau_k(u) & k\in\{i,j\}\\
0 & \text{otherwise}
\end{cases},\text{ where $u$ is the corresponding node of arc $\alpha$.}\]
Schematically, we have that matrix ${\mathcal{T}(\tau_0,\calS)}$ includes the $(\tau_i(node(\alpha)),-\tau_j(node(\alpha)))_{\alpha=(i,j)\in \arcs(\calS)}$ pairs expanded to $\{0,\pm 1\}^{\emph{1-move}(\calS)}$ .
\[\footnotesize
{\mathcal{T}(\tau_0,\calS)}=
\begin{bmatrix}
0 & \tau _{i}( node( \alpha _{1} =( i,j))) & 0 & \cdots  & -\tau _{j}( node( \alpha _{1} =( i,j))) & 0 & \cdots  & 0\\
\vdots  & 0 & \vdots  & \vdots  & 0 &  &  & \vdots \\
\vdots  & \vdots  & 0 & \vdots  & \vdots  &  &  & \vdots \\
0 & 0 & \tau _{i'}( node( \alpha _{arcs(\mathcal{S})}) =( i',j'))) & 0 & \cdots  & \tau _{i'}( node( \alpha _{arcs(\mathcal{S})}) =( i',j'))) & 0 & 0
\end{bmatrix}^{\top }
\]
Again, notice that we can derive $\mathcal{T}(\tau_0')$ by multiplying $\mathcal{T}(\tau_0)$ from right by
the squared diagonal $|\arcs(S)|\times|\arcs(S)|$ matrix $\mathcal{D'}$ such that
\[\mathcal{D}'[\tau_0,\tau_0']_{(\alpha,\alpha)=((i,j),(i,j))}=\tau_0(u)\tau_0'(u)\text{ where $u$ is the node of $\alpha$ and } \mathcal{D}'[\tau_0,\tau_0']_{(\alpha,\alpha')}=0\text{ for $\alpha\neq \alpha'$.}\] 
Indeed,   we have that for an entry $(k,\alpha)$ representing an arc $\alpha=(i,j)$ and the $k$-th $1$-move 

\begin{align*}
(\mathcal{T}(\tau_0)\mathcal{D}'[\tau_0,\tau_0'])_{k,\alpha}&=& &
\begin{Bmatrix*}[l]
\tau_0(u)\tau_0'(u)\tau_i(u)=\tau_i(u)\tau_i'(u)\tau_i(u)=\tau_i'(u)  & \text{if }(\mathcal{T}(\tau_0))_{k,\alpha}=\tau_i(u)\\
-\tau_0(u)\tau_0'(u)\tau_j(u)=\tau_j(u)\tau_j'(u)\tau_j(u)=-\tau_j'(u) & \text{if }(\mathcal{T}(\tau_0))_{k,\alpha}=-\tau_j(u)\\
0& \text{if }(\mathcal{T}(\tau_0))_{k,\alpha}=0
\end{Bmatrix*}=(\mathcal{T}(\tau_0'))_{k,\alpha}
\end{align*}
where first equality leverages Fact~\ref{fact:config-invariance} and the second one uses the fact
$\tau_i(v)^2$ equals 1 for any $v\in V_n$ for every $i\in[\ell]$. To sum-up, it holds that 
\begin{align*}
    M_{\arcs(\mathcal{S})}(\tau_0',\calS))&=\mathcal{M}_{ \emph{$1$-move}}(\tau_0',\calS))\cdot \mathcal{T}(\tau_0',\calS)=
\left(\mathcal{D}[\tau_0,\tau_0']\mathcal{M}_{ \emph{$1$-move}}(\tau_0,\calS))\right)\cdot \left(\mathcal{T}(\tau_0,\calS)\mathcal{D}'[\tau_0,\tau_0']\right)\\
&=\mathcal{D}[\tau_0,\tau_0']\left(\mathcal{M}_{ \emph{$1$-move}}(\tau_0,\calS))\right)\cdot \mathcal{T}(\tau_0,\calS))\mathcal{D}'[\tau_0,\tau_0']=
\mathcal{D}[\tau_0,\tau_0']M_{\arcs(\mathcal{S})}(\tau_0',\calS))\mathcal{D}'[\tau_0,\tau_0']
\end{align*}

Given that $\mathcal{D}'[\tau_0,\tau_0']$ and $\mathcal{D}[\tau_0,\tau_0']$ are full-rank matrices, leveraging Fact~\ref{fact:rank-invariance}, we can prove that the above argument proves that $\rank(
\mathcal{M}_{\arcs(S)}(\tau_0,\calS))$ is independent of the initial configuration.

Now, in order to prove Lemma~\ref{lem:independence-init-dependent-cycles}, we recall the definition of a dependent cycle $C$ of size $t$, namely let $C=\left(c_1=\{u_1,u_2\},c_2=\{u_2,u_3\}\ldots,c_t=\{u_t,u_1\}\right)$ in a move sequence $\calS$. For an initial configuration $\tau_0\in \{\pm 1\}^{V(\calS)}$,
we say that cycle  is \emph{dependent} with respect to $\tau_0$ if there exists $b\in \{\pm 1\}^t$ such that 
\[\begin{bmatrix}
\tau _{1}( u_{1}) & 0 & \cdots  & \cdots  & 0 & \tau _{t}( u_{1})\\
\tau _{1}( u_{2}) & \tau _{2}( u_{2}) & \cdots  & \cdots  & \cdots  & 0\\
0 & \tau _{2}( u_{3}) & \tau _{3}( u_{3}) & \cdots  & \cdots  & 0\\
\vdots  & \vdots  & \ddots  & \ddots  & \ddots  & 0\\
\vdots  & \vdots  & \vdots  & \ddots  & \ddots  & 0\\
0 & 0 & 0 & 0 & \tau _{t-1}( u_{t}) & \tau _{t}( u_{t})
\end{bmatrix}\begin{bmatrix}
b_{1}\\
b_{2}\\
\vdots \\
\vdots \\
\vdots \\
b_{t}
\end{bmatrix} =\begin{bmatrix}
0\\
\vdots \\
\vdots \\
\vdots \\
\vdots \\
0
\end{bmatrix}\equiv \Delta(\tau_0)\cdot b = 0_t
\]
Again, notice that we can derive $\Delta(\tau_0')$ by multiplying $\Delta(\tau_0)$ from left by
the squared diagonal $|C|\times|C|$ matrix $\mathcal{D''}$ such that
\[\mathcal{D''}[\tau_0,\tau_0']_{(u_k,u_k')}=\tau_0(u)\tau_0'(u)\text{ for $k=k'$ and } \mathcal{D''}[\tau_0,\tau_0']_{(u_k,u_k')}=0\text{ for $k\neq k'$.}\] 
Indeed,   we have that for an entry $(k,c_i)$ representing a 2-\emph{move} $c_i=\{u_{i \mod t},u_{i+1 \mod t}\}$ and the $u_k$-th node of the cycle $C$:

\begin{align*}
(\mathcal{D}''[\tau_0,\tau_0']\Delta(\tau_0))_{k,c_i}&=& &
\begin{Bmatrix*}[l]
\tau_0(u_k)\tau_0'(u_k)\tau_i(u_k)=\tau_i(u_k)\tau_i'(u_k)\tau_i(u_k)=\tau_i'(u_k)  & \text{if }(\Delta(\tau_0))_{k,c_i}=\tau_i(u_k)\\
0& \text{if }(\Delta(\tau_0))_{k,c_i}=0
\end{Bmatrix*}=(\Delta(\tau_0'))_{k,c_i}
\end{align*}

Since $\mathcal{D}''[\tau_0,\tau_0']$ is full-rank matrix, we get that $\null(\Delta(\tau_0))=\null(\Delta(\tau_0'))$. Therefore if there exists a non-zero vector $b$ such that $\Delta(\tau_0) \cdot b = 0 $ then $\Delta(\tau_0') \cdot b = 0 $ as well, completing the proof for
Lemma~\ref{lem:independence-init-dependent-cycles}.

For the last case of Lemma~\ref{lem:independence-init-rank-cycles}, we start by the following observations:
\begin{enumerate}
    \item If $\Delta(\tau_0)\cdot b=0$, then for any $b'=\lambda b$, it also holds that $\Delta(\tau_0)\cdot b'=0$ for any non-zero $\lambda$ constant.
    \item If $\Delta(\tau_0)\cdot b=0$ and $b_k=0$ for some $k\in[t]$, then $b$ is the zero vector, $b=0$.
    \item More precisely, the vector that belong to the (right) null space of $\Delta(\tau_0)$, i.e., all the vectors $b$ such that $\Delta(\tau_0)\cdot b=0$, are of the following form:
    \[\Delta(\tau_0)\cdot b=0\Leftrightarrow b=b_1\cdot\left(1,-\frac{\tau_{c_1}(u_2)}{\tau_{c_2}(u_2)},\cdots, (-1)^{k-1}\frac{\prod_{i\in[2:k]}\tau_{c_{i-1}}(u_i)}{\prod_{i\in[2:k]}\tau_{c_{i}}(u_i)},\cdots, 
    (-1)^{t-1}\frac{\prod_{i\in[2:t]}\tau_{c_{i-1}}(u_i)}{\prod_{i\in[2:t]}\tau_{c_{i}}(u_i)}\right)^\top\]
    \item The term $\frac{\prod_{i\in[2:k]}\tau_{c_{i-1}}(u_i)}{\prod_{i\in[2:k]}\tau_{c_{i}}(u_i)}$ is independent of the initial configuration.
\end{enumerate}
Items (1)-(3) are simple linear algebras derivation. Item (4) holds since \[
\frac{\prod_{i\in[2:k]}\tau_{c_{i-1}}(u_i)}{\prod_{i\in[2:k]}\tau_{c_{i}}(u_i)}=
\frac{\prod_{i\in[2:k]}\tau_{c_{i-1}}(u_i)}{\prod_{i\in[2:k]}\tau_{c_{i}}(u_i)}\times 
\underbrace{\frac{\prod_{i\in[2:k]}\tau'_{c_{i-1}}(u_i)\tau_{c_{i-1}}(u_i)}{\prod_{i\in[2:k]}\tau'_{c_{i}}(u_i)\tau_{c_{i}}(u_i)}}_{=1 \text{ from Fact~\ref{fact:config-invariance}} }=
\frac{\prod_{i\in[2:k]}\tau'_{c_{i-1}}(u_i)}{\prod_{i\in[2:k]}\tau'_{c_{i}}(u_i)}\]
Hence, if $b$ is a cancelling vector with $b_1=1$, then from items (1),(2),(3),(4) we have that $b$ is unique for every cycle $C$ and independent initial configuration $\tau_0$.

More interestingly, in order to prove Lemma~\ref{lem:independence-init-rank-cycles}, it suffices to prove that the rank of the column matrix with the improvement vectors of cycles is independent of the initial configuration. 
In fact, let \[
M_{\cycles(\mathcal{S})}(\tau_0,\calS)) =\begin{bmatrix}
\vdots  &  & \vdots  &  & \vdots \\
\imp_{\tau _{0} ,\mathcal{S}}( C_{1}) & \cdots  & \imp_{\tau _{0} ,\mathcal{S}}( C _{k}) & \cdots  & \imp_{\tau _{0} ,\mathcal{S}}( C_{|\cycles(\mathcal{S}) |})\\
\vdots  &  & \vdots  &  & \vdots 
\end{bmatrix}
\]
By definition, we have that for any $C\in\cycles(\calS)$:
\[
\imp_{\tau_0,\calS}(C=(c_1,\cdots,c_t))=\sum_{j\in [t]} b_j(C)\cdot \imp_{\tau_0,\calS}(c_j).\]
or equivalently $M_{\cycles(\mathcal{S})}(\tau_0,\calS))=\mathcal{M}_{ \emph{$2$-move}}(\tau_0,\calS))\cdot \mathcal{B}(\calS)$ 
 where $\mathcal{T}(\tau_0,\calS)$ is a sparse ($|\emph{2-move}(\calS)| \times |\cycles(S)|$) rectangular matrix such that \[\mathcal{B}(\calS)_{\rho-\text{th \emph{$2-$move}},C}=\begin{cases}
b_\rho(C) & c_\rho\in C\\
0 & \text{otherwise}
\end{cases},\text{ where $u$ is the corresponding node of arc $\alpha$.}\]
Schematically, we have that matrix ${\mathcal{B}(\calS)}$ includes the $(b(C))_{C\in \cycles(\calS)}$ vectors expanded to $\{0,\pm 1\}^{\emph{2-move}(\calS)}$ .
\[\footnotesize
{\mathcal{B}(\calS)}=
\begin{bmatrix}
0           & b_{1}(C_1) & 0& b_{\rho}(C_1) & 0 & b_{|C_1|}(C_1) & 0 & \cdots  & 0\\
b_{1}(C_2)& 0 & \cdots & b_{\rho'}(C_2) & \vdots &0 &\cdots & b_{|C_2|}(C_2)& 0 \\
& 0 & \cdots &  & \vdots & & \cdots && 0 \\
b_{1}(C_{|\cycles(\calS)|})& b_{\rho''}(C_{|\cycles(\calS)|}) & 0 & \cdots & \cdots &0 &\cdots & b_{|C_{|\cycles(\calS)}|}(C_{|\cycles(\calS)|})& 0 \\
\end{bmatrix}^{\top }
\]
Having noticed the above ones, it is easy to see that
 $\rank(
\mathcal{M}_{\arcs(S)}(\tau_0,\calS))$ is independent of the initial configuration, since
\begin{align*}
    \mathcal{M}_{\arcs(S)}(\tau_0',\calS)&=\mathcal{M}_{ \emph{$2$-move}}(\tau_0',\calS) \cdot
{\mathcal{B}(\calS)}=(\mathcal{D}[\tau_0,\tau_0']\cdot\mathcal{M}_{ \emph{$2$-move}}(\tau_0,\calS)) {\mathcal{B}(\calS)}\\&=\mathcal{D}[\tau_0,\tau_0']\cdot(\mathcal{M}_{ \emph{$2$-move}}(\tau_0,\calS) {\mathcal{B}(\calS)})=
\mathcal{D}[\tau_0,\tau_0']\cdot \mathcal{M}_{\arcs(S)}(\tau_0,\calS).
\end{align*}
Thus, by Fact~\ref{fact:rank-invariance} we get that  $\rank(
\mathcal{M}_{\arcs(S)}(\tau_0,\calS))=\rank(
\mathcal{M}_{\arcs(S)}(\tau_0',\calS))$, concluding also the proof of Lemma~\ref{lem:independence-init-rank-cycles}.
\end{proof}

%% file: main.bbl
\begin{thebibliography}{10}

\bibitem{spielman2004smoothed}
Daniel~A Spielman and Shang-Hua Teng.
\newblock Smoothed analysis of algorithms: Why the simplex algorithm usually
  takes polynomial time.
\newblock {\em Journal of the ACM (JACM)}, 51(3):385--463, 2004.

\bibitem{beier2003random}
Ren{\'e} Beier and Berthold V{\"o}cking.
\newblock Random knapsack in expected polynomial time.
\newblock In {\em Proceedings of the thirty-fifth annual ACM symposium on
  Theory of computing}, pages 232--241, 2003.

\bibitem{englert2016smoothed}
Matthias Englert, Heiko R{\"o}glin, and Berthold V{\"o}cking.
\newblock Smoothed analysis of the 2-opt algorithm for the general tsp.
\newblock {\em ACM Transactions on Algorithms (TALG)}, 13(1):1--15, 2016.

\bibitem{roglin2007smoothed}
Heiko R{\"o}glin and Berthold V{\"o}cking.
\newblock Smoothed analysis of integer programming.
\newblock {\em Mathematical programming}, 110(1):21--56, 2007.

\bibitem{dadush2018friendly}
Daniel Dadush and Sophie Huiberts.
\newblock A friendly smoothed analysis of the simplex method.
\newblock In {\em Proceedings of the 50th Annual ACM SIGACT Symposium on Theory
  of Computing}, pages 390--403, 2018.

\bibitem{bhaskara2014smoothed}
Aditya Bhaskara, Moses Charikar, Ankur Moitra, and Aravindan Vijayaraghavan.
\newblock Smoothed analysis of tensor decompositions.
\newblock In {\em Proceedings of the forty-sixth annual ACM symposium on Theory
  of computing}, pages 594--603, 2014.

\bibitem{farrell2016smoothed}
Brendan Farrell and Roman Vershynin.
\newblock Smoothed analysis of symmetric random matrices with continuous
  distributions.
\newblock {\em Proceedings of the American Mathematical Society},
  144(5):2257--2261, 2016.

\bibitem{blum2002smoothed}
Avrim Blum and John Dunagan.
\newblock Smoothed analysis of the perceptron algorithm for linear programming.
\newblock 2002.

\bibitem{arthur2011smoothed}
David Arthur, Bodo Manthey, and Heiko R{\"o}glin.
\newblock Smoothed analysis of the k-means method.
\newblock {\em Journal of the ACM (JACM)}, 58(5):1--31, 2011.

\bibitem{arthur2006worst}
David Arthur and Sergei Vassilvitskii.
\newblock Worst-case and smoothed analysis of the icp algorithm, with an
  application to the k-means method.
\newblock In {\em 2006 47th Annual IEEE Symposium on Foundations of Computer
  Science (FOCS'06)}, pages 153--164. IEEE, 2006.

\bibitem{sivakumar2020structured}
Vidyashankar Sivakumar, Steven Wu, and Arindam Banerjee.
\newblock Structured linear contextual bandits: A sharp and geometric smoothed
  analysis.
\newblock In {\em International Conference on Machine Learning}, pages
  9026--9035. PMLR, 2020.

\bibitem{boodaghians2020smoothed}
Shant Boodaghians, Joshua Brakensiek, Samuel~B Hopkins, and Aviad Rubinstein.
\newblock Smoothed complexity of 2-player nash equilibria.
\newblock In {\em 2020 IEEE 61st Annual Symposium on Foundations of Computer
  Science (FOCS)}, pages 271--282. IEEE, 2020.

\bibitem{coordination}
Shant Boodaghians, Rucha Kulkarni, and Ruta Mehta.
\newblock Smoothed efficient algorithms and reductions for network coordination
  games.
\newblock In Thomas Vidick, editor, {\em 11th Innovations in Theoretical
  Computer Science Conference, {ITCS} 2020, January 12-14, 2020, Seattle,
  Washington, {USA}}, volume 151 of {\em LIPIcs}, pages 73:1--73:15. Schloss
  Dagstuhl - Leibniz-Zentrum f{\"{u}}r Informatik, 2020.

\bibitem{beier2022smoothed}
Ren{\'e} Beier, Heiko R{\"o}glin, Clemens R{\"o}sner, and Berthold V{\"o}cking.
\newblock The smoothed number of pareto-optimal solutions in bicriteria integer
  optimization.
\newblock {\em Mathematical Programming}, pages 1--37, 2022.

\bibitem{xia2020smoothed}
Lirong Xia.
\newblock The smoothed possibility of social choice.
\newblock {\em Advances in Neural Information Processing Systems},
  33:11044--11055, 2020.

\bibitem{JAMS89}
David~S. Johnson, Cecilia~R. Aragon, Lyle~A. McGeoch, and Catherine Schevon.
\newblock Optimization by simulated annealing: An experimental evaluation; part
  i, graph partitioning.
\newblock {\em Oper. Res.}, 37(6):865--892, 1989.

\bibitem{NPC}
M.R. Garey, D.S. Johnson, and L.~Stockmeyer.
\newblock Some simplified np-complete graph problems.
\newblock {\em Theor. Comput. Sci.}, pages 237--267, 1976.

\bibitem{KL70}
Brian~W. Kernighan and Shen Lin.
\newblock An efficient heuristic procedure for partitioning graphs.
\newblock {\em Bell Syst. Tech. J.}, 49(2):291--307, 1970.

\bibitem{johnson1988easy}
David~S Johnson, Christos~H Papadimitriou, and Mihalis Yannakakis.
\newblock How easy is local search?
\newblock {\em Journal of Computer and System Sciences}, 37(1):79--100, 1988.

\bibitem{schaffer1991simple}
Alejandro~A Sch{\"a}ffer and Mihalis Yannakakis.
\newblock Simple local search problems that are hard to solve.
\newblock {\em SIAM journal on Computing}, 20(1):56--87, 1991.

\bibitem{elsasser2011settling}
Robert Els{\"a}sser and Tobias Tscheuschner.
\newblock Settling the complexity of local max-cut (almost) completely.
\newblock In {\em International Colloquium on Automata, Languages, and
  Programming}, pages 171--182. Springer, 2011.

\bibitem{etscheid2015smoothed}
Michael Etscheid and Heiko R\"{o}glin.
\newblock Smoothed analysis of local search for the maximum-cut problem.
\newblock {\em ACM Trans. Algorithms}, 13(2):25:1--25:12, March 2017.

\bibitem{angel2017local}
Omer Angel, S{\'e}bastien Bubeck, Yuval Peres, and Fan Wei.
\newblock Local max-cut in smoothed polynomial time.
\newblock In {\em Proceedings of the 49th Annual ACM SIGACT Symposium on Theory
  of Computing}, pages 429--437. ACM, 2017.

\bibitem{bibak2019improving}
Ali Bibak, Charles Carlson, and Karthekeyan Chandrasekaran.
\newblock Improving the smoothed complexity of flip for max cut problems.
\newblock In {\em Proceedings of the Thirtieth Annual ACM-SIAM Symposium on
  Discrete Algorithms}, pages 897--916. SIAM, 2019.

\bibitem{chen2020}
Xi~Chen, Chenghao Guo, Emmanouil~V. Vlatakis-Gkaragkounis, Mihalis Yannakakis,
  and Xinzhi Zhang.
\newblock Smoothed complexity of local max-cut and binary max-csp.
\newblock In {\em Proceedings of the 52nd Annual ACM SIGACT Symposium on Theory
  of Computing}, STOC 2020, page 1052–1065, New York, NY, USA, 2020.
  Association for Computing Machinery.

\bibitem{FM82}
Charles~M. Fiduccia and Robert~M. Mattheyses.
\newblock A linear-time heuristic for improving network partitions.
\newblock In {\em Proceedings of the 19th Design Automation Conference}, pages
  175--181. {ACM/IEEE}, 1982.

\bibitem{roglin2008complexity}
Heiko R{\"o}glin.
\newblock {\em The complexity of Nash equilibria, local optima, and Pareto
  optimal solutions}.
\newblock PhD thesis, Aachen, Techn. Hochsch., Diss., 2008, 2008.

\end{thebibliography}
